\documentclass[12pt]{article}

\usepackage{graphicx}
\usepackage{epsfig}
\usepackage{psfrag}
\usepackage{wrapfig}
\usepackage[all]{xy}
\usepackage{bbm}

\usepackage{fullpage}
\usepackage{setspace}
\usepackage{url}
\usepackage{xcolor}
\usepackage{subcaption}

\usepackage{mathtools}
\usepackage{float}

\usepackage{footmisc}
\newcommand\wordcount{
    \immediate\write18{texcount -sum -1 \jobname.tex > 'count.txt'}
\input{count.txt}words}

\usepackage{natbib}

\usepackage{amsmath}
\usepackage{amsfonts}
\usepackage{amsthm}
\usepackage{amssymb}
\usepackage{bbm}
\usepackage{array}
\usepackage{enumitem}

\newtheorem{theorem}{Theorem}[section]
\newtheorem{lemma}{Lemma}[section]
\newtheorem{remark}{Remark}[section]
\newtheorem{prop}{Proposition}

\newtheorem{assn}{C}
\usepackage{appendix}
\usepackage{titletoc}
\usepackage[toctitles]{titlesec}
\usepackage[hidelinks]{hyperref}
\hypersetup{colorlinks,linkcolor={red},citecolor={blue},urlcolor={red}}  

\newcommand\DoToC{%
	\startcontents
	\printcontents{}{2}{\textbf{Contents}\vskip3pt\hrule\vskip5pt}
	\vskip3pt\hrule\vskip5pt
}

\def\Var{{\rm Var}\,}
\def\AVar{{\rm AVar}\,}
\def\E{{\rm E}\,}
\def\arg{{\rm arg}\,}
\def\Cov{{\rm Cov}\,}
\def\N{{\rm N}\,}

\newcommand{\X}{\mathbf{X}}
\newcommand{\Y}{\mathbf{Y}}
\newcommand{\Z}{\mathbf{Z}}
\newcommand{\I}{\mathbf{I}}
\newcommand{\V}{\textnormal{V}}
\newcommand{\W}{\mathbf{W}}

\newcommand{\M}{\mathbf{M}}

\newcommand{\z}{\mathbf{z}}

\newcommand{\mutd}{\bar{\mu}^1(d)}
\newcommand{\mucd}{\bar{\mu}^0(d)}

\newcommand{\muit}{\mu_i(\mathbf{1};d)}
\newcommand{\muic}{\mu_i(\mathbf{0};d)}
\newcommand{\mujt}{\mu_j(\mathbf{1};d)}
\newcommand{\mujc}{\mu_j(\mathbf{0};d)}
\newcommand{\muitt}{\mu_i(\mathbf{1_i},\mathbf{1_j};d)}
\newcommand{\muicc}{\mu_i(\mathbf{0_i},\mathbf{0_j};d)}
\newcommand{\muitc}{\mu_i(\mathbf{1_i},\mathbf{0_j};d)}
\newcommand{\muict}{\mu_i(\mathbf{0_i},\mathbf{1_j};d)}
\newcommand{\mujtt}{\mu_j(\mathbf{1_i},\mathbf{1_j};d)}

\newcommand{\sumij}{\sum_{i;j\in\mathcal{B}(i;d),j\not =i}}
\newcommand{\muia}{\mu_i(\mathbf{a};d)}

\newcommand{\mujb}{\mu_j(\mathbf{b};d)}
\newcommand{\muiab}{\mu_i(\mathbf{a_i},\mathbf{b_j};d)}
\newcommand{\mujab}{\mu_j(\mathbf{a_i},\mathbf{b_j};d)}
\newcommand{\node}{\mathcal{S}}

\DeclarePairedDelimiter\floor{\lfloor}{\rfloor}

\newcommand{\HT}{\textnormal{HT}}
\newcommand{\HA}{\textnormal{HA}}
\newcommand{\IPW}{\textnormal{IPW}}
\newcommand{\PD}{\textnormal{PD}}
\newcommand{\HAC}{\textnormal{HAC}}
\newcommand{\VhatHACd}{\widehat{\textnormal{V}}_{\HAC}(d)}
\newcommand{\SAH}{\textnormal{SAH}}
\newcommand{\SM}{\textnormal{sm}}
\newcommand{\obs}{\textnormal{obs}}
\newcommand{\Taylor}{\textnormal{L}}
\newcommand{\C}{O_i}
\newcommand{\Cj}{O_j}

\newtheorem{definition}{Definition}

\begin{document}
\author{Ye Wang,
Cyrus Samii, 
Haoge Chang, and
P. M. Aronow\thanks{  
Wang is Assistant Professor, Department of Political Science, University of North Carolina, Chapel Hill, USA (Email: yewang@unc.edu). 
Samii (contact author) is Associate Professor, Department of Politics, New York University, New York, USA (Email: cds2083@nyu.edu). 
Chang is Assistant Professor, Department of Economics, Columbia University, New York, USA (Email: hc3615@columbia.edu).
Aronow is Associate Professor, Departments of Statistics and Data Science, Political Science, Biostatistics and Economics, Yale University, New Haven, USA (Email: p.aronow@yale.edu).
For their comments and suggestions, we thank Kirill Borusyak, Stephen Cole, Alexander Demin, Naoki Egami, Jiawei Fu, Michael Hudgens, Peter Hull, Molly Roberts, Fredrik S\"avje, Davide Viviano, and seminar participants at Harris School at University of Chicago, New York University Abu Dhabi, Princeton University, Rochester University, Texas A\&M., UNC, and UCSD}.}

\title{Design-Based Inference for Spatial Experiments under Unknown Interference}
\maketitle

\thispagestyle{empty}
\setcounter{page}{0}

\clearpage

\begin{center}
{\Large Design-Based Inference for Spatial Experiments under Unknown Interference}
\end{center}
\vspace{2em}
\begin{abstract} 
We consider design-based causal inference for spatial experiments in which treatments may have effects that bleed out and feed back in complex ways. Such spatial spillover effects violate the standard ``no interference'' assumption for standard causal inference methods. The complexity of spatial spillover effects also raises the risk of misspecification and bias in model-based analyses. We offer an approach for robust inference in such settings without having to specify a parametric outcome model. We define a spatial ``average marginalized effect'' (AME) that characterizes how, in expectation, units of observation that are a specified distance from an intervention location are affected by treatment at that location, averaging over effects emanating from other intervention nodes. We show that randomization is sufficient for non-parametric identification of the AME even if the nature of interference is unknown. Under mild restrictions on the extent of interference, we establish asymptotic distributions of estimators and provide methods for both sample-theoretic and randomization-based inference. We show conditions under which the AME recovers a structural effect. We illustrate our approach with a simulation study.  Then we re-analyze a randomized field experiment and a quasi-experiment on forest conservation, showing how our approach offers robust inference on policy-relevant spillover effects.
\vspace{1em}
\end{abstract}

\clearpage
\doublespace

\section{Introduction}
Consider a spatial experiment where an intervention is randomly assigned to specific locations in a geographic space. 
Then, we observe how outcomes are distributed over this geography. Figure~\ref{fig:illustrations} illustrates the generic structure of such experiments. 
The left panel presents a hypothetical point-intervention experiment. 
An experimental design could treat half of such points to receive the intervention (gray shaded points), with the rest of the points remaining in a control condition without intervention (unshaded points). The shading in the background raster indicates outcome values. 
The right panel illustrates a similar experiment, but with interventions assigned to polygons instead of points. 

A recent example of a spatial experiment is from \cite{Jayachandran267}, who study a forest conservation intervention by comparing forest cover outcomes in villages that hosted the intervention to those that did not.  
(In this case, the villages are the polygons.) 
For this application, and others like it, a major concern is the possibility of detrimental spillover effects: while the intervention may reduce deforestation within a village in which it is applied, it may simply push the deforestation into other nearby areas. 
The dashed lines in Figure~\ref{fig:illustrations} display possible zones into which effects may bleed out in our hypothetical examples. 
The extent and manner in which such spillover effects occur could be shaped by how the intervention is distributed over the entire space.
The spillover effects might be smaller when there is a high saturation of treatment rather than a low saturation.
The situation at hand is one in which the potential outcomes at a point in space depend not only on that point's treatment status but rather on how treatments are distributed elsewhere---what is referred to as ``interference'' in the causal inference literature (\citealp[p. 19]{cox58}).

\begin{figure}\centering 
\includegraphics[width=.75\textwidth]{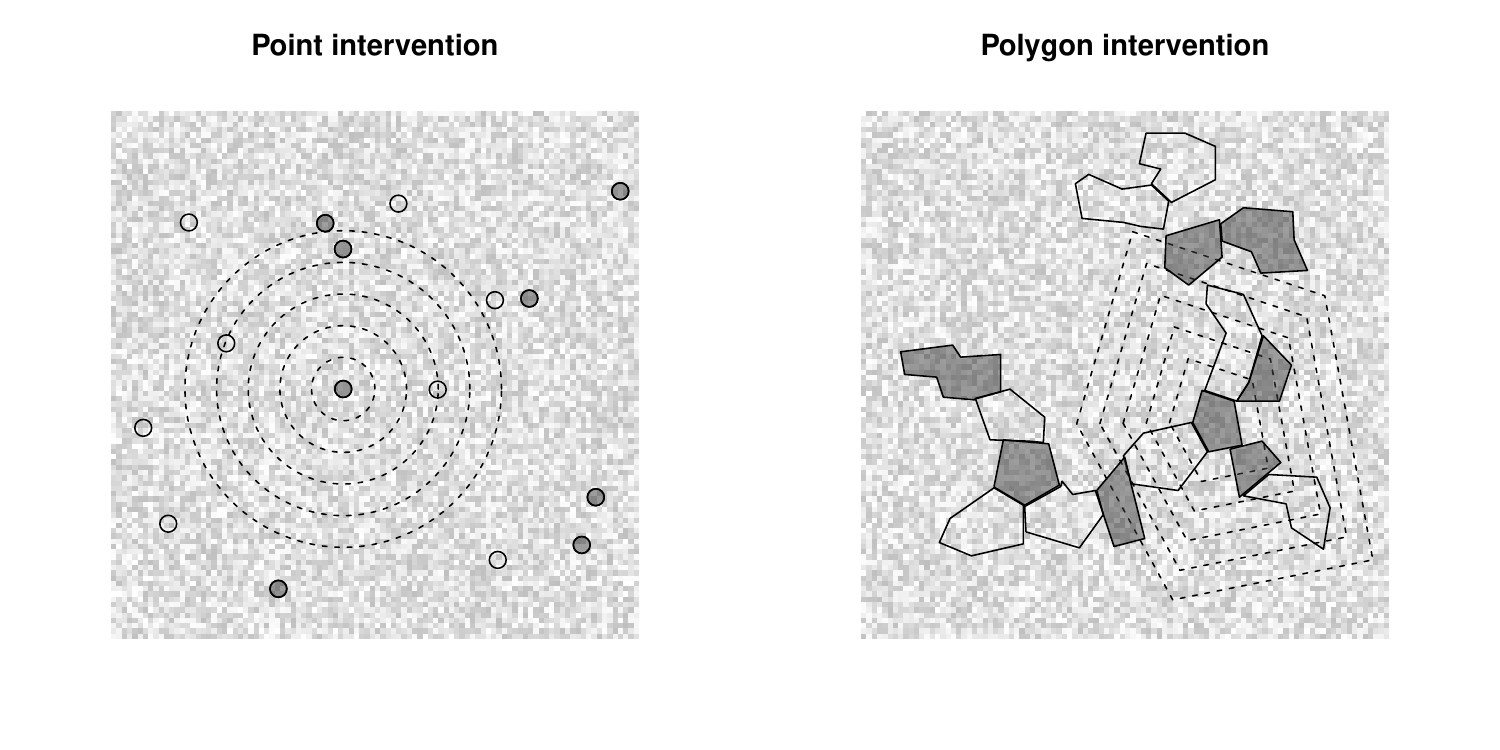}
\caption{Illustrations of hypothetical spatial experiments in which interventions are applied to points (left) or polygons (right). The background raster captures the geographic outcome data. Effects may bleed out in space, as illustrated by the concentric dashed lines.}
\label{fig:illustrations}
\end{figure}  

As \cite{reich2021review} discuss, in the current literature on spatial causal inference, established ways to approach spillover effects expand the set of potential outcomes to account for different types of direct and indirect exposure to a treatment.
As an alternative to parametric model-based approaches, these potential outcomes-based approaches are design-based, meaning that they aim for robustness by leveraging what is known about the experimental design while remaining agnostic about features of the outcome distributions.\footnote{See \citet{imbens_rubin15} for a general argument and \cite{gibbons2012mostly} for a discussion particular to spatial analysis.} 
One such approach is the ``exposure mapping'' method of \cite{aronow_samii2017_interference}, in which one defines discrete potential outcomes for points in the geography depending on whether the points are within a treated polygon and whether the points are within some set distance from any other treated points. 
In circumstances where the geography is partitioned into distinct areas within which spillover effects are contained, an alternative and somewhat more agnostic approach focuses on marginalized ``direct'' and ``indirect'' effects as defined by  \cite{hudgens_halloran08} or \cite{zigler-papadogeorgou2018-bipartite}.

The challenge in applied settings is that the appropriate exposure mapping may be unknown and spillovers may not be neatly contained within discrete areas. 
Spillover effects may bleed out in more complex ways and interact with each other.
These challenges apply to the interference problem in non-spatial settings as well, e.g., in settings of network interference, leading to recent theoretical work to consider what kinds of inference might be achievable under ``unknown interference'' \citep{savje2021average, li2022random, hu2022average}.
We bring these ideas to the spatial context in order to approach problems such as the concern about detrimental spillovers in the forest conservation application.
We show that when interference is present, rather than estimating the ATE, contrasts between treated and control areas capture a more nuanced, but nonetheless, policy-relevant quantity that we call the ``average marginalized effect'' (AME).
We can define AMEs at different distance intervals from where the intervention is applied, thereby capturing how effects transmit in space in potentially non-monotonic ways.
The AME is similar to the average direct effect of \citet{halloran1995causal} (see also \citet{hudgens_halloran08}, \citet{vanderweele2011effect}, \citet{savje2021average}, \citet{li2022random}, and \citet{hu2022average}), except that it is indexed by distance for application to the spatial case.
The AME measures how, on average, outcomes within the specified distance interval from an intervention node are affected by activating a treatment at that node, taking into account ambient effects emanating from treatments at other intervention nodes.  
Taking into account the ambient effects also allows us to capture the consequences of more or less uniform, and more or less dense, distributions of treatments.
There is a direct mapping from the AME to effects that are assumed by parametric models of spatial effects: when effects emanating from different intervention nodes are additive, the AME recovers the average of these additive effects.  
If spatial effects are not simply additive but exhibit complex interactions or feedback, the AME still yields an interpretable and policy-relevant quantity. 
This interpretation is robust and does not depend on, e.g., a parametric spatial lag, autoregressive, or weighting structure \citep{golgher2016interpret}.

We also develop inferential methods to account for the dependencies that interference creates.  
We work with a Horvitz-Thompson estimator and a Hajek estimator for the AME. 
We show that these estimators are consistent and asymptotically normal under weak restrictions on the degree of interdependence induced by interference. 
We further prove that the commonly-used spatial heteroscedasticity and autocorrelation consistent (HAC) variance estimator of \cite{conley99_spatial} provides conservative estimates for the true variance of these estimators under conditions that are often satisfied in practice.

Our analysis is related to a few streams of current methodological research on spatial causal inference and interference. 
First, our approach draws most directly on recent design-based analyses of causal effects under interference that consider estimands marginalized over the randomization distribution, as in \cite{hudgens_halloran08}, \cite{savje2021average}, \cite{papadogeorgou2020causal},  \cite{li2022random}, and \citet{hu2022average}.  
As in these approaches, our inference does not require that we specify the functional form of all interference-induced potential outcomes precisely (i.e., the ``exposure mapping'' of \cite{aronow_samii2017_interference}). It thus skirts the issue of non-overlap caused by a potentially high-dimensional set of potential outcomes \citep{leung2022rate}.
Second, our analysis is related to recent work on non-parametric estimation of spatial effects, including work on ``bipartite causal inference'' by \cite{zigler-papadogeorgou2018-bipartite} and on cluster-randomized designs by \cite{leung2022rate}. 
These works focus on cases where points of intervention are far enough apart to yield disjoint clusters that interfere with each other minimally. 
Such designs are appealing, but they are not always feasible.
Similar to $m$-dependence for a time series, we assume hard limits on the extent of interference, but we do not assume that the set of units can be partitioned into a set of disjoint clusters. 
Third, our inferential results rely on the contributions of \citet{ogburn2017causal}. 
We also draw connections to inferential results in the spatial regression literature \citep{arbia2006spatial, jenish2016spatial, kelejian_piras2017}. 
We justify regression estimators on spatial data from the design-based perspective and provide causal interpretations for the coefficients. 
We clarify the connection to \cite{conley99_spatial}'s spatial HAC variance estimator. 

We begin by developing the formal inferential setting and main theoretical results, using a toy example to illustrate concepts.  
We then consider extensions and refinements. 
 We provide simulation evidence of the performance of our proposed estimators and then turn to the forest conservation application that we used above to motivate the analysis.

\section{Setting}\label{Section:setting}


Suppose a set of intervention nodes $\node = \{1,...,N\}$\label{notation:node}. Each node $i \in\node$ can be either a point or a collection of points (e.g., a polygon) that resides in a two-dimensional set $\mathcal{X}\label{notation:outcomeset}$ indexed by $x=(x_1,x_2)$ (e.g., latitude and longitude). An experimental design assigns a binary treatment $Z_i \in \{0,1\}$ to each intervention node. The ordered vector of experimental assignment variables is $\Z \equiv (Z_1,...,Z_N)\label{notation:randomtreatment}$, and the {\it ex post} realized assignment from the experiment is given by $\z \equiv (z_1,...,z_\N)\in \{0,1\}^N\label{notation:realizedtreatment}$.  The experimental design fixes the set of possible assignment vectors as well as a probability distribution over that set. Our analysis considers the case of Bernoulli randomization for each $Z_i$, i.e., (possibly weighted) coin flips to determine treatment status at each node.\footnote{We provide a discussion on complete randomization at the end of Section \ref{ame}.  } 
	

Potential outcomes at any point $x \in \mathcal{X}$ are defined for each value of $\z$, $(Y_x(\z))_{\z \in \{0,1\}^N}\label{notation:potentialoutcome}$. Given a realized treatment assignment $\z$, we observe the corresponding potential outcome at $x$:
\begin{equation}
Y_x = \sum_{\z \in \{0,1\}^N} Y_x(\z)I(\Z=\z)\label{notation:observedoutcome},
\end{equation}
where $I(\cdot)$ is the indicator function and $I(\Z=\z)$ takes on the value 1 if the random assignments $\Z$ realize to be $\z$, and 0 otherwise. Data for points in $\mathcal{X}$ may come in various formats including raster data or data on a discrete set of points in $\mathcal{X}$.  Let $\Y(\z) = (Y_x(\z))_{x\in \mathcal{X}}\label{notation:potentialoutcome_vec}$ denote the full set of potential outcomes when $\Z= \z$ and $\Y = (Y_x)_{x\in \mathcal{X}}\label{notation:observedoutcome_vec}$ denote the full set of realized outcomes.

We map the full set of potential outcomes $\Y(\z)$ for all points in the outcome space $\mathcal{X}$ back to the intervention nodes  by defining the ``circle average'' function. For $i\in\mathcal{S}$, we define

\begin{equation}\label{def:circleaverage}
\mu_i(\Y(\z);\Omega_d) = \frac{\int_{x:d_i(x) \in \Omega_d} Y_x(\z)\text{d}\zeta}{\int_{x:d_i(x) \in \Omega_d} \text{d}\zeta}.
\end{equation}

In the expression above, $d_i(x)$ measures the distance between point $x$ and intervention node $i$. When $i$ is a point located at $x(i)$, $d_i(x) = \gamma\left(x(i), x\right)$, where $\gamma(\cdot,\cdot)$ is a well-defined metric (e.g., Euclidean, geodesic, or a least-cost distance) that satisfies triangular inequality. If $i$ is a collection of points, then $d_i(x) = \min_{x' \in i} ||x' - x||$, the minimal distance between $x$ and points belonging to $i$. $\Omega_d$ is a set of distance values and $\zeta$ is a suitable measure on $\mathcal{X}$. Therefore, $\mu_i(\Y(\z);\Omega_d) $ is the average outcome across points whose distance to $i$ falls in $\Omega_d$. 

If the points in $\mathcal{X}$ are dense and spaced evenly, then $\Omega_d$ could be a singleton, $\Omega_d=\{d\}$. The circle average amounts to taking the average across points along the edge of a circle of radius $d$ around $i$. If the points are spaced such that there are few or no points precisely at the edge of the circle, $\Omega_d$ could be a \textit{donut} where $\Omega_d= [d - \kappa , d]$ (with $\kappa$ a user-chosen constant dictating the donut's thickness), or a \textit{disk} where $\Omega_d=[0,d]$. By considering a collection of sets across different d values, $\{\Omega_d\}_{d \in \mathcal{D}}$, we will be able to examine how the circle average's value varies over the geography.
When it does not cause confusion, we write $\mu_i(\Y(\z);\Omega_d)$ simply as $\mu_i(\Y(\z);d)$. Similarly, the realized circle average for intervention node $i$ at $d$ is

\begin{equation}\label{def:obs_circleaverage}
\mu_i(\Y;d) = \sum_{\z \in \{0,1\}^N} \mu_i(\Y(\z);d)I(\Z=\z).
\end{equation}

This representation allows us to see how an experiment is a process of sampling potential circle averages for intervention nodes, and therefore allows us to apply sample theoretic results in our analysis below.

The left plot in Figure \ref{fig:null-raster} illustrates a toy example for a point intervention ($N=4$) and raster outcome data.  The plot shows a ``null raster'' for which none of the intervention nodes has been assigned to treatment and so $\z=(0,0,0,0)$. The outcomes are $Y_x(0,0,0,0)$ for all $x$ in the space.  As we can see, outcomes are defined for any point $x$ in the space, although outcomes are constant within raster cells. This is a feature of the raster data.  Other types of data may exhibit finer levels of granularity --- e.g., data produced from kriging interpolation that varies smoothly in space.  We take these outcome data, and any coarsening or smoothing operations that they incorporate, as fixed.  For our design-based inference, the only source of stochastic variation is from $\Z$. For data that are smoothed using kriging, one could tune smoothing parameters on auxiliary data so that they are fixed with respect to $\Z$. White circles around the intervention nodes demonstrate one possible way to construct the circle average. We use the Euclidean distance and take averages across all the raster cells passed by the edge of the circles. Note that we do not prohibit circles around different nodes to intersect with each other.

\begin{figure}\centering
\includegraphics[width=0.33\textwidth]{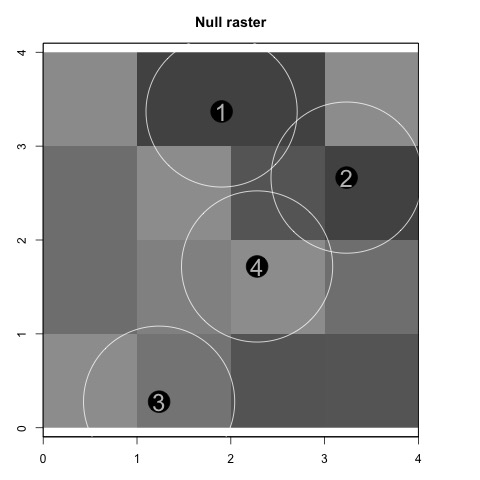}
\includegraphics[width=0.33\textwidth]{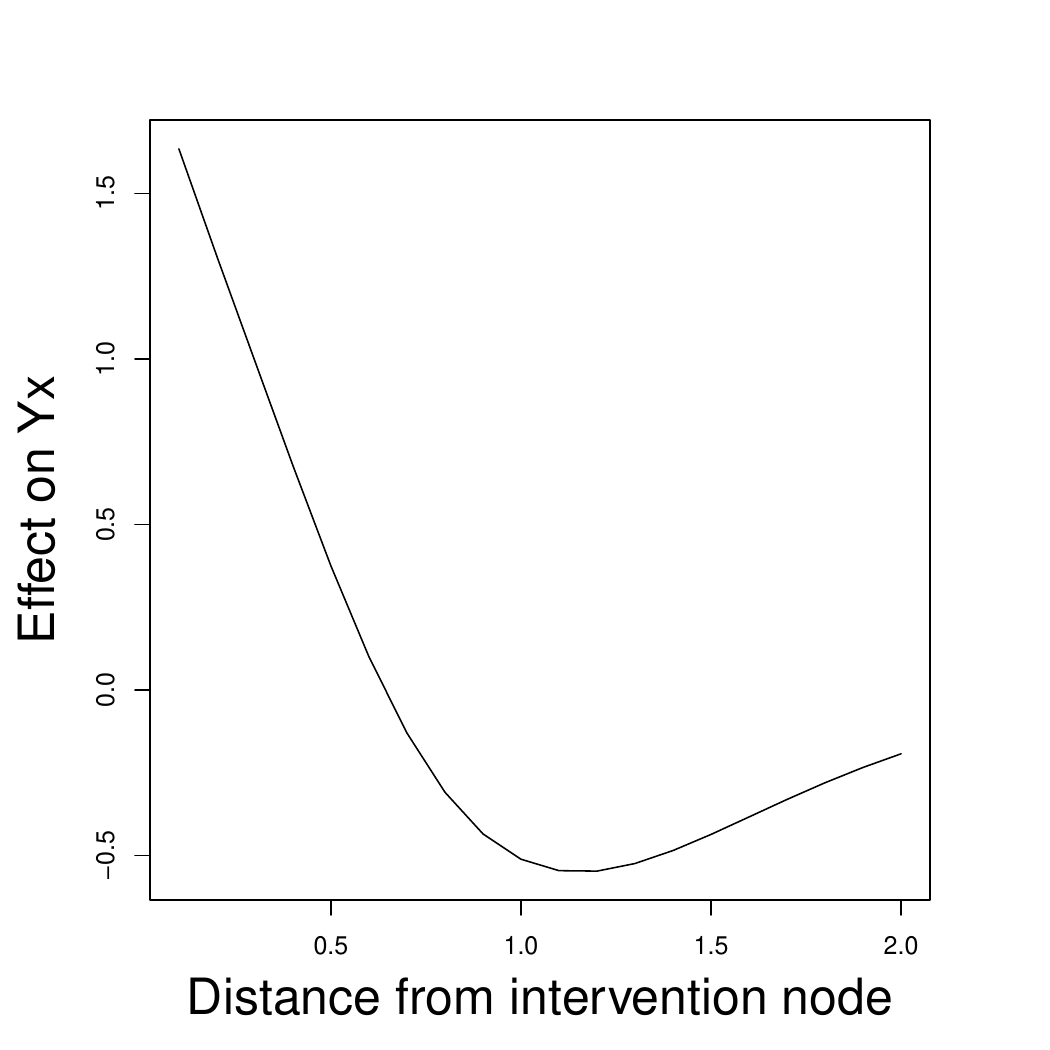}
\caption{Left: Illustration of a ``null raster,'' with $N=4$ intervention nodes (points), none of which are assigned to treatment.  Raster cells are colored according to outcome levels. White circles around the nodes are where circle averages are computed. Lighter colors represent larger outcome values. Right: a possible effect function that is non-monotonic in distance.}
\label{fig:null-raster}
\end{figure}

Spatial effects can exhibit considerable complexity. For the sake of illustration, our toy example supposes that treatments transmit effects that are non-monotonic in distance and that effects from different intervention nodes accumulate in an additive manner.  The right plot in Figure \ref{fig:null-raster} illustrates such an effect function. Then, the net result would depend on how treatments are distributed over the intervention points.  Figure \ref{fig:po-raster} illustrates how outcomes would be affected over different allocations of the treatment given that effects take the form as in Figure \ref{fig:null-raster}.  In the analysis below, we do not assume that effects are additive or homogeneous in form---this is done here merely to provide a simple illustration.

\begin{figure}[h]\centering
\includegraphics[width=.8\textwidth]{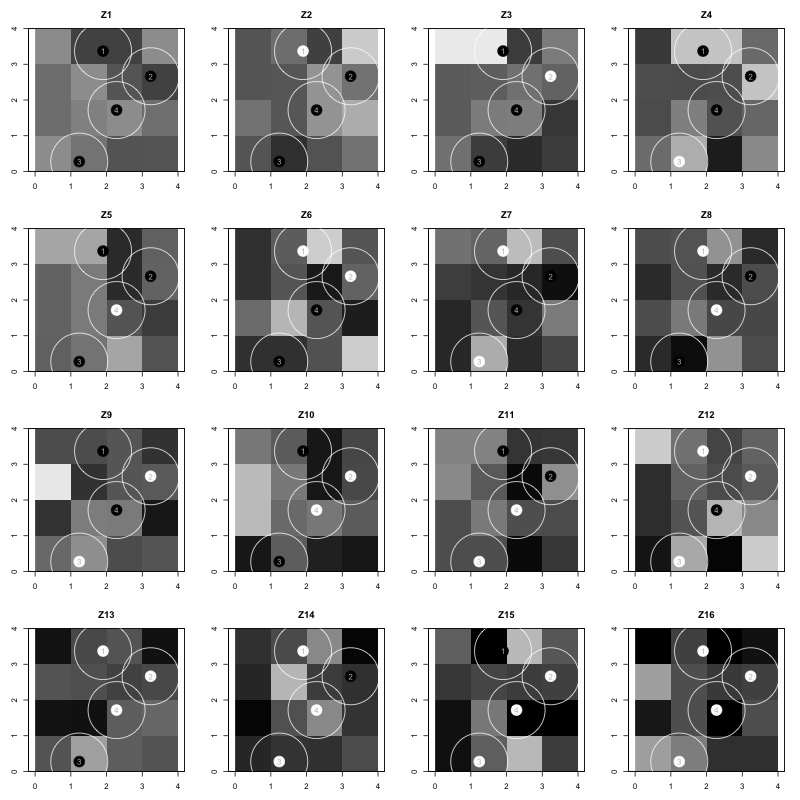}
\caption{Illustration of how outcomes are affected given different treatment allocations given the effect function in Figure \ref{fig:null-raster}.  Treated intervention points are white, while non-treated intervention points are black.}
\label{fig:po-raster}
\end{figure}

\section{Defining a marginal spatial effect}\label{ame}
As the potential outcome notation for $Y_x(\z)$ indicates, the outcome at any point may depend on the full vector of realized treatment assignments $\z$.  Similarly, the potential outcome notation for the circle average, $\mu_i(\Y(\z);d)$, shows that the realized circle average for node $i$ may depend on treatment assignments for nodes other than $i$.  As such, the circle averages are potentially subject to causal ``interference.''  

We now define a spatial effect that we call the ``average marginalized effect'' (AME). The AME is a marginal effect that accounts for interference.  It can be defined for any distance value $d$ and enables researchers to examine how the impacts of the intervention nodes on the outcome of their neighboring points vary in space.  The usual definition of a unit-level treatment effect takes the difference between a unit's potential outcomes under one treatment condition versus under another treatment condition. A unit-level marginal effect is different because it takes the difference between the {\it average} of a unit's potential outcomes over a set of potential outcomes versus the average over another set.  We apply this idea to the spatial setting.  In doing so, we consider effects that may bleed out in ways that are not necessarily contained within pre-defined strata, as in \cite{hudgens_halloran08}, or summarized by a simple statistic, as in \cite{aronow_samii2017_interference}. 

To define the spatial AME, let us first rewrite the potential outcome at point $x$ as $Y_x(z_i, \mathbf{z}_{-i})$, where $\mathbf{z}_{-i}$ is a vector equaling $\mathbf{z}$ except that the value for intervention node $i$ is omitted.  This allows us to pay special attention to how variation in treatment at node $i$ relates to potential outcomes at point $x$, given the variation in treatment values in $\mathbf{z}_{-i}$.   We can marginalize over variation in $\mathbf{z}_{-i}$ to define an ``individualistic'' average of potential outcomes for point $x$, holding the treatment at intervention node $i$ to treatment value $z$:
\begin{equation}\label{notation:AME_at_x}
Y_{ix}(z_i; \eta) = \E_{\Z_{-i}}[Y_x(z_i, \Z_{-i})]=
\sum_{\mathbf{z}_{-i}\in\{0,1\}^{N-1}} Y_x(z_i, \mathbf{z}_{-i})\mathrm{Pr}(\mathbf{Z}_{-i}=\mathbf{z}_{-i}; \eta),
\end{equation}
where $\eta$ is an experimental design parameter that is an index for the distribution of $\mathbf{Z}$ (that is, the probability of treatment assignments).  This is the individualistic marginal potential outcome at point $x$ given that node $i$ is assigned to treatment condition $z$, marginalizing over possible assignments to other nodes.   We can use Figure \ref{fig:po-raster} to illustrate.  To construct $Y_{1x}(0; \eta)$, one would take a weighted average of the potential outcomes at point $x$ under assignments labeled in the figure as Z1, Z3, Z4, Z5, Z9, Z10, Z11, and Z15, where the weights would be proportional to the probability of each assignment.

We can define a similar marginal quantity at the level of the circle averages:
\begin{equation}\label{notation:AME_at_i}
\mu_{i}(z_i; d, \eta) = \E_{\Z_{-i}}[\mu_i(\Y(z_i, \mathbf{Z}_{-i}); d)] = \sum_{\mathbf{z}_{-i}\in \{0,1\}^{N-1}} \mu_i(\Y(z_i, \mathbf{z}_{-i}); d)\mathrm{Pr}(\mathbf{Z}_{-i}=\mathbf{z}_{-i};\eta),
\end{equation}
where we use $\Y(z_i, \mathbf{z}_{-i})$ to denote the vector of potential outcomes over points in $\mathcal{X}$ that obtain under treatment assignment $(z_i, \mathbf{z}_{-i})$. 
This is the potential circle average at distance $d$ around node $i$, given that $i$ is assigned to treatment condition $z_i$, marginalizing over possible assignments to other nodes.

We can now define an {\it individual marginalized effect} at point $x$ of intervening on node $i$, allowing other nodes to vary as they otherwise would under $\eta$:
\begin{equation}\label{notation:individual_ame_x}
\tau_{ix}(\eta) = Y_{ix}(1;\eta) - Y_{ix}(0; \eta).
\end{equation}

This defines the response at point $x$ of switching node $i$ from no treatment to active treatment, averaging over possible treatment assignments to nodes other than $i$. At the level of circle averages, we can define 
\begin{equation}\label{notation:individual_ame_i}
\tau_{i}(d; \eta) = \mu_{i}(1; d, \eta) - \mu_{i}(0; d, \eta),
\end{equation}
which is the average of individual responses for points along the circle at distance $d$ around node $i$.  Using Figure \ref{fig:po-raster} to illustrate, one would construct $\tau_{1}(d; \eta)$ by working with the $d$-radius circle averages around intervention node $1$, taking the difference between the mean of the circle averages under assignments Z2, Z6, Z7, Z8, Z12, Z13, Z14, and Z16 minus the mean of circle averages under assignments Z1, Z3, Z4, Z5, Z9, Z10, Z11, and Z15.

Finally, define the {\it average marginalized effect} (AME) for distance $d$ by taking the mean over the intervention nodes:

\begin{equation}\label{def:AME}
	\textnormal{AME}(d;\eta)= \frac{1}{N} \sum_{i=1}^N \tau_{i}(d; \eta)
\end{equation}

The interpretation of the AME for distance $d$ is the average effect of switching a node $i \in \node$ to treatment on points at distance $d$ from that node, marginalized over possible realizations of treatment statuses in other intervention nodes. The distribution of these possible realizations of treatment statuses depends on the experimental design. When $d = 0$, the AME captures the direct effect generated by the treatment at the location of intervention, in a way similar to the ``expected average treatment effect'' in \cite{savje2021average}. For $d > 0$, the AME resembles the ``average indirect causal effect'' in \cite{hu2022average} but is defined for specific distance values. In practice, researchers may select a series of distance values, $\{d_l\}_{l=1}^L$, based on the resolution of $\mathcal{X}$ and the potential magnitude of spillover effects. The resulting collection of AMEs demonstrates how effects vary over the distance from an intervention node.

Before ending this section, we note that our analysis focuses on experimental designs with Bernoulli assignment, in which case the possible assignments consists of the $2^N$ possible vectors that could be obtained from $N$ (possibly differentially weighted) coin flips.  This allows for a clean definition of causal effects \citep{savje2021average}. This is because Bernoulli assignment ensures that $(1, \mathbf{z}_{-i})$ and $(0, \mathbf{z}_{-i})$ each has positive probability of occurring.  In this case, the marginal quantities $Y_{ix}(1; \eta)$ and $Y_{ix}(0; \eta)$ are defined by marginalizing over the same sets of $\mathbf{z}_{-i}$ values, and the individualistic response has a clear {\it ceteris paribus} interpretation.  Things are different under completely randomized assignment, where a fixed number $N_1$ of nodes are assigned to treatment. Then, for $Y_{ix}(1; \eta)$, one marginalizes over assignments with $N_1-1$ units assigned to treatment, while for $Y_{ix}(0; \eta)$, one marginalizes over assignments with $N_1$ units assigned to treatments.   As $N$ grows, differences between AMEs in Bernoulli and complete random assignment typically become negligible when interference is local, as shown in \cite{savje2021average}.  

\section{Inferential assumptions}\label{Section:inference}

In this section, we lay out assumptions on the experimental design and potential outcomes, including restrictions on the extent of interference for the inferential results in Section \ref{est-inf}. Our asymptotic analysis considers a sequence of sets indexed by the sample size $N$. The set of intervention nodes is denoted as  $\node_N$ and the set of outcome points $\mathcal{X}$. Note that assumptions below are assumed to hold uniformly for all large sample sizes. 
We begin with the following assumptions:
\begin{assn}{(Bernoulli design)}\label{assn:bern-des}
$(Z_1,...,Z_N)$ is a vector of independent $\text{Bernoulli}(p)$ draws.
\end{assn}

\begin{assn}{(Bounded potential outcomes)}\label{assn:bounded-y} $|Y_x(\z)| < b$ for some finite constant $b$ and all $x \in \mathcal{X}$ and $\z \in \{0,1\}^N$.
\end{assn}


Assumption C\ref{assn:bern-des} defines the experimental design.
As discussed above, condition C\ref{assn:bern-des} ensures that individualistic responses are {\it ceteris paribus} for variation in treatment assignment at a given node.
We work with the assumption that the assignment probability, $p$, is constant over intervention nodes and discuss the extension to cases where assignment probabilities vary in Section 6.4 below. Assumption C\ref{assn:bounded-y} is a common regularity condition on the potential outcomes. It ensures the boundedness of higher-order moments for the distribution of functions of the potential outcomes.

Our next assumption follows \citet{savje2021average} by using a dependency graph to characterize interference-induced dependencies among the circle averages defined at a specific distance value $d$.  Let $I_{ij}(d)$ be an indicator for whether assignment at intervention node $j$ interferes with the $d$-radius circle average at node $i$:
\begin{equation}
I_{ij}(d) = \left\{
\begin{array}{ll}
	1 & \text{ if } \mu_i(\mathbf{Y}(\mathbf{z});d) \ne \mu_i(\mathbf{Y}(\mathbf{z}');d) \text{ for some } \mathbf{z}, \mathbf{z}' \in \{0,1\}^N \text{ such that } \mathbf{z}_{-j} = \mathbf{z}_{-j}'\\
	1 & \text{ if } i = j,\\
	0 & \text{ otherwise}.
\end{array}
\right.
\end{equation}

Then, let $s_{ij}(d)$ be an indicator for whether $d$-radius circle averages at $i$ and $j$ are subject to interference from treatment at some common intervention node $\ell$ (which could be $i$, $j$, or some other third intervention node):
\begin{equation}
	s_{ij}(d) = \left \{ 
\begin{array}{ll}
	1 & \text{ if } I_{i\ell }(d)I_{j\ell }(d) = 1 \text{ for some } \ell \in \node_N,\\
	0 & \text{ otherwise }.
\end{array}
\right.
\end{equation}

If $s_{ij}(d) =  1$, then circle averages at $i$ and $j$ will vary together whenever there is variation in treatment values at the relevant $\ell$s, meaning non-independence over possible values of $\mathbf{Z}$. 

Using this dependency graph, our third assumption is a restriction on the extent of interference dependencies for circle averages at $d$. Let us denote the distance between two intervention nodes $i$ and $j$ as $d_{ij}$.\footnote{As before, under point intervention, $d_{ij} = \gamma(x(i) , x(j))$. Under polygon intervention, $d_{ij} = \min_{x \in i, x' \in j} \gamma(x , x')$.} Then we have:

\begin{assn}{(Local interference.)}\label{assn:local-inf} Let $h:[0,\infty)\to [0,\infty)$ be a function independent of sample sizes. For each $d$ and all large sample sizes $N$, and all pairs of intervention nodes $i$ and $j$ in $\node_N$, there exists a constant $h(d)$ such that if $d_{ij} > h(d)$, then  $s_{ij}(d)=0$.
\end{assn}

Assumption C\ref{assn:local-inf} means that there are hard limits to the spatial extent of the interference: nodes that are beyond some distance from each other have no interference-induced dependencies. Assumption C\ref{assn:local-inf} is an assumption on the possible extent of spillovers. For intervention node $j$ to satisfy C\ref{assn:local-inf} with respect to $i$, it would require that the outcomes $Y_x(\z)$ used to construct the circle average of $i$ at distance $d$ is unaffected by \textit{not only} $j$'s treatment value but also the treatment values at any intervention nodes that affect the circle average of $j$ at distance $d$. 
For example, consider a simple case where the outcomes used to construct the circle average of $j$ at distance $d$ depend on $Z_j$ and another assignment variable $Z_k$. If $s_{ij}(d)=0$, this implies that the outcomes used to construct the circle average of $i$ at distance $d$ cannot depend on either $Z_j$ or $Z_k$. In other words, $i$ and $j$ cannot share any sources of variation in their circle averages at distance $d$.  This assumption hence implies that for $i,j$ with $d_{ij} > h(d)$,  $\mu_i(Y(z, \Z_{-i});d)$ and $\mu_j(Y(z', \Z_{-j});d)$ are independent with each other. 

To give a more interpretable sufficient condition, suppose that the intervention at every node has no influence on the outcome points more than $\bar{d}$ away from it. Then, C\ref{assn:local-inf} would be satisfied if the distance between the boundary of  the radius-$d$ circle around node $i$ and the radius-d circle around node $j$ is larger than $2\bar{d}$. In this case, $h(d)$ in C\ref{assn:local-inf} can be set to $2\bar{d} + 2d$. When nodes $i$ and $j$ are $2\bar{d} + 2d$ apart, they neither interfere with each other, nor do they have any common neighbors that influence both nodes. To further illustrate the idea, a simple example is given in Figure \ref{fig:C3}. In this example with three intervention nodes,  intervention node 2 is on the radius-$d$ circle centering node 3. Because nodes 1 and 3 are sufficiently far apart, neither node 3 and node 2 affect the d-circle average of node 1 and vice versa.


\begin{figure}\centering
\includegraphics[width=0.5\textwidth]{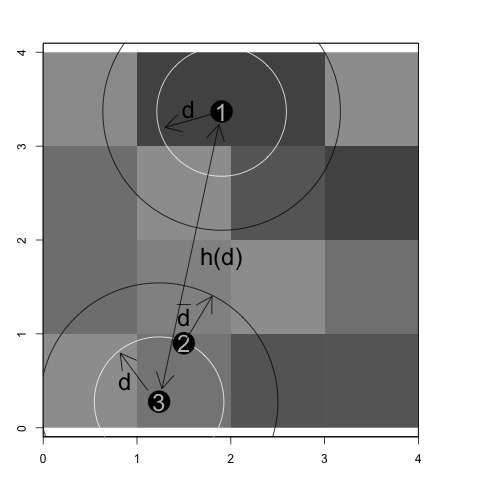}
\caption{Illustration of condition C3. White circles around the nodes are circle averages defined at the distance value $d$. Black circles with a radius of $d+\bar{d}$ depict the maximal range that interference can happen for the $d$-circle averages. As the distance between node 1 and node 3 is larger than $2\bar{d}+2d$,  the circle averages of the two nodes do not depend on each other, even when a third node lies between them.}
\label{fig:C3}
\end{figure}

A final assumption defines an increasing-domain asymptotic growth process in which the number of independent pairs of intervention nodes increases. Define $\node_{N}(i; d) \equiv \{j \in \node_N: d_{ij} \leq d\}$, the set of intervention nodes whose distance to node $i$ is less than $d$. We have the following asymptotic restriction on the spacing of the intervention nodes:  
	\begin{assn}{(Intervention node spacing)}\label{assn:spacing} Let $b:[0,\infty)\to [0,\infty)$ be a function independent of sample sizes. For each $d$ and all large sample sizes $N$, the sequence of intervention nodes satisfies $\sup_{i\in \node_N}|\node_{N}(i; d)| \leq b(d)$.
\end{assn}

C\ref{assn:spacing} ensures that as the size of intervention node grows, the number of intervention nodes that reside within a given distance of a node is bounded. It  is satisfied when the intervention nodes are deliberately chosen such that they are adequately spaced out geographically. For example, for point-intervention experiments, we can choose the set of nodes $\node_N$ from a meshgrid where the distance between any two points $i$ and $j$ on the grid is bounded from below, i.e. $d_{ij} \geq d_0$. In practice, researchers can first divide the space into disjoint areas and select one intervention node from each area to make C\ref{assn:spacing} plausible.\footnote{See \cite{leung2022rate, leung2023design} for algorithms that enable researchers to select intervention nodes in space such that assumption C\ref{assn:spacing} is likely satisfied.} For polygon-intervention experiments, one may require that the size of each polygon is larger than a threshold value (thus ensuring adequate spacing between non-adjacent polygons).\footnote{ It should be noted that the assumptions presented here are not the most general possible for the Wald-type inference we study below. One can relax our assumptions, for example, to allow the $h(\cdot)$ in C\ref{assn:local-inf} and the $b(\cdot)$ in C\ref{assn:spacing} to slowly increase with $N$. However, in such cases the convergence rate of the estimators will no longer be $\sqrt{N}$. We leave such cases for a future study. }

We note that C\ref{assn:local-inf} and C\ref{assn:spacing} should not be confused with each other. C\ref{assn:local-inf} is an assumption on the extent of spillover effects. C\ref{assn:spacing} is an assumption on the spacings of the intervention nodes which prevents nodes from concentrating densely in a particular region. To see how the two assumptions fit together, we complete our specification of the extent of interference. Let us define the  neighborhood $\mathcal{B}(i;d)$ that includes all the nodes whose circle averages at $d$ may interfere with that of node $i$: 
\begin{equation}
\mathcal{B}_N(i;d) = \{j \in \node_N: d_{ij}  \le h(d)\}.
\end{equation}
From C\ref{assn:local-inf}, we know that $s_{ij}(d) = 0$ for $j \not \in \mathcal{B}_N(i;d)$. From C\ref{assn:spacing}, we know that $|\mathcal{B}_N(i;d)|\leq b(h(d))$.
Define $c _i(d)=|\mathcal{B}_N(i;d)| $ and $c_{N}(d)= \max_{i \in \node_N} c_i(d)$.
Conditions C\ref{assn:local-inf} and C\ref{assn:spacing} imply the following condition:
\begin{description}
\item[C   4a.] (Limited local interference with respect to $d$) For each $d$, \label{assn:d-bound} $c_N(d)=O(1) $.
\end{description}
We note that this asymptotic implication is only \textit{pointwise} in $d$: for each $d$, we have $c_N(d)=O(1)$ but it is not true that $\sup_{d\in\mathbb{R}}c_N(d)=O(1)$. In practice, researchers should focus on distance values within a moderate range, such that the area covered by any circle is not excessive relative to the whole geography. In our analyses below, we suppress the subscripts for asymptotic sequences unless they are needed to add clarity.

\section{Estimation and inference}\label{est-inf}

In this section, we study the Horvitz-Thompson (HT) and Hajek estimators for the AME and establish their consistency and asymptotic normality. We recommend the use of the Hajek estimator, as it is usually more efficient in practice compared with the HT estimator.\footnote{We provide a theoretical efficiency comparison between the HT and Hajek estimators in Section \ref{Section:HAHTComparison}.} We express the Hajek estimator as a regression estimator and propose a variance estimator based on the spatial heteroscedasticity and autocorrelation consistent (spatial HAC) estimator of \citet{conley99_spatial}. Hence, our proposed Hajek estimator for the AME is equivalent to a regression of the circle average on a constant term and the treatment indicator. The inference is carried out using a spatial HAC standard error estimator and normal approximation.  All proofs are contained in the appendix.  

Consider the following Horvitz-Thompson (HT) estimator:
\begin{align}
\widehat{\tau}_{\HT}(d) = \frac{1}{Np}\sum_{i=1}^N Z_i \mu_i(\Y; d) - \frac{1}{N(1-p)}\sum_{i=1}^N (1-Z_i) \mu_i(\Y; d).\label{eq:ht}
\end{align}

The terms on the right-hand side consist of design parameters $N$ and $p$, assignment indicators $\{Z_i\}_{i=1}^N$, and observed circle average $ \{\mu_i(\Y; d)\}_{i=1}^N$ as defined in (\ref{def:circleaverage}). Hence the quantity is computable from observed data alone.

Our first two results show that $\widehat{\tau}_{\HT}(d)$ is unbiased for  the AME at distance $d$ under C\ref{assn:bern-des}, and is consistent and asymptotically normal under C\ref{assn:bern-des}-C\ref{assn:spacing}. 

\begin{prop}[Unbiasedness]\label{prop:identification}
Under C\ref{assn:bern-des}, 
\begin{equation}
\E_{\Z}\left[	\widehat{\tau}_{\HT}(d)  \right]=\textnormal{AME}(d;\eta),
\end{equation}
where the expectation is taken over the random assignment variables.
\end{prop}


Let $N(0,1)$ denote the standard Gaussian distribution with mean 0 and variance 1.
\begin{prop}[Asymptotic Distribution for the Horvitz-Thompson estimator]\label{prop:ht-asymptotics}
Under C\ref{assn:bern-des}-C\ref{assn:spacing} and if $N\times\Var(\widehat{\tau}_{\HT}(d))$ is uniformly bounded below for all large $N$, then, as $N \rightarrow \infty$ , 	
\begin{equation}
	\frac{\widehat{\tau}_{\HT}(d) - \textnormal{AME}(d;\eta)}{\sqrt{\Var(\widehat{\tau}_{\HT}(d))}} \overset{d}{\rightarrow}  N(0,1),\footnote{ $\Var(\widehat{\tau}_{\HT}(d))$ is characterized in Lemma \ref{lemma:ht-var}.}
\end{equation}

\end{prop}

The Hajek estimator is a refinement to the Horvitz-Thompson estimator and, in this setting, is equivalent to a difference-in-means estimator: instead of using $pN$ and $(1-p)N$ as the denominator, the Hajek estimator replaces them with $N_1=\sum_{i=1}^n Z_i$ and $N_0=N-N_1$, respectively:
\begin{align}
	\widehat{\tau}_{\HA}(d) = \frac{1}{N_1}\sum_{i=1}^{N}Z_i \mu_i(\Y; d) - \frac{1}{N_0}\sum_{i=1}^{N}(1-Z_i) \mu_i(\Y; d) \label{eq:hajek}
	\end{align}
The Hajek estimator is usually more efficient in pracitce, in terms of variances , to a Horvitz-Thompson estimator.\footnote{ Note that this is not always true theoretically: there can be potential outcomes for which the asymptotic variance of the Hajek estimator is larger than that of a Horvitz-Thompson estimator. Nevertheless, we believe such potential outcomes often do not arise in practice. We provide a theoretical comparison in Section \ref{Section:HAHTComparison}.}

\begin{prop}[Asymptotic Distribution for the Hajek estimator]\label{prop:hajek-asymptotics}	
		Let $\AVar(\widehat{\tau}_{\HA}(d))$ denote the asymptotic variance of the Hajek estimator.\footnote{The precise definition and characterization of the asymptotic variance of the Hajek estimator is stated in Lemma \ref{lemma:ha-var-homo}.}
		Under C\ref{assn:bern-des}-C\ref{assn:spacing} and if $N\times \AVar(\widehat{\tau}_{\HA}(d))$ is uniformly bounded below for all large $N$, then, as $N \rightarrow \infty$  
		\begin{equation}
			\frac{\widehat{\tau}_{\HA}(d) - \textnormal{AME}(d;\eta)}{\sqrt{\AVar(\widehat{\tau}_{\HA}(d))}} \overset{d}{\rightarrow}  N(0,1).
		\end{equation}
\end{prop}

As a simple difference in means, the Hajek estimator is algebraically equivalent to a least square regression of the circle averages on an intercept and treatment indicators of the intervention nodes:
\begin{equation}
\left(\begin{array}{c}\widehat{\mu}_0(d) \\ \widehat{\tau}_{\HA}(d)\end{array}\right) = \arg \min_{(\mu_0, \tau)} \sum_{i=1}^N \left(\mu_i(\Y; d) -\mu_0-\tau Z_i\right)^2.
\end{equation}
Our approach to variance estimation borrows from the spatial econometrics literature and works with the spatial heteroskedasticity and autocorrelation consistent (spatial HAC) variance estimator of \citet{conley99_spatial}. This estimator takes the form,
\begin{align}\label{HAC_formula}
\widehat{\Sigma}_{\HAC}(d)= (\X' \X)^{-1}\left(\sum_{i=1}^N \sum_{j=1}^N\X_i' \X_j \hat{e}_i \hat{e}_j K\left(\frac{d_{ij}}{\tilde{d}} \right)\right)(\X' \X)^{-1},
\end{align} where 
$\X  = \begin{pmatrix}
1  & 1 &  \hdots, 1 \\
Z_1 & Z_2 & \hdots Z_N\\
\end{pmatrix}'\in \mathbb{R}^{N\times 2}$ and $\X_i $ denotes the $i$th row of the matrix. The $\hat{e}_i$'s are the residuals from the regression, where $\hat{e}_i=\mu_i(\Y;d)-\widehat{\mu}_0(d)-\widehat{\tau}_{\HA}(d)Z_i$.  $K(\cdot)$ is a kernel function. $\tilde{d}$ is a cutoff value and our setup suggests setting it at $\tilde{d}=h(d)$. 
The (2,2)-entry of the estimator $\widehat{\Sigma}_{\HAC}(d)$ is our estimator for the variance of $\widehat{\tau}_{\HA}(d)$ and we denote it as $\widehat{\textnormal{V}}_{\HAC}(d)$. In practice, $h(d)$ is unknown and we can examine the robustness of the results by varying $\tilde d$. 
In the appendix, we show that the regression estimator combined with the spatial HAC variance estimator with the uniform kernel provides asymptotically valid inference for the AME under an extra assumption:

\begin{assn}{(Homophily in treatment effects)}\label{assn:homo} $\frac{1}{N}\sum_{i=1}^N (\tau_i(d;\eta) - 	\textnormal{AME}(d;\eta)) \sum_{j \in \{i\}\cup \mathcal{B}(i;d)}(\tau_j(d;\eta) -	\textnormal{AME}(d;\eta)) \geq 0$ for each value $d\geq 0$.\footnote{Remember that $\tau_{i}(d;\eta) = \E_{\Z_{-i}}[\mu_i(\Y(1, \mathbf{Z}_{-i}); d)] - \E_{\Z_{-i}}[\mu_i(\Y(0, \mathbf{Z}_{-i}); d)]$.} 
\end{assn}

The assumption is that the expected treatment effect generated by node $i$ at distance $d$ is positively correlated with that generated by its neighbors in $\mathcal{B}(i;d)$ and itself. In other words, there is homophily in treatment effects in space: nodes that generate larger-than-average effects reside close to each other. We consider that a positive spatial correlation assumption reasonable in many applied settings so we recommend the use of HAC variance estimator in practice. 

When researchers have concern over C\ref{assn:homo}, one can use an alternative variance-bound estimator proposed in \citet{savje2021average}:
\begin{equation}\label{varianceSAH}
	\widehat{\textnormal{V}}_{\SAH}(d)=\frac{1}{N^2}\sum_{i=1}^N\frac{Z_i c_i(d)\widehat{e}_i^2}{p^2} +     \frac{1}{N^2}\sum_{i=1}^N \frac{(1-Z_i) c_i(d)\widehat{e}_i^2}{(1-p)^2}.
\end{equation}
The validity of this variance estimator does not depend on C\ref{assn:homo}, although the estimates tend to be overly conservative in realistic dataset. We provide a comparison using the simulation exercises below.

We summarize all the inferential results in the following proposition. For some $\alpha<1$, let $z_{\frac{\alpha}{2}}$ and $z_{1-\frac{\alpha}{2}}$ be the $\frac{\alpha}{2}$th and $\left(1-\frac{\alpha}{2}\right)$th quantiles of the standard normal distribution, respectively.
\begin{prop}\label{prop:inference}
Let $K(\cdot)$ in (\ref{HAC_formula}) be the uniform kernel.\footnote{$K(x)=I\left(x\in [0,1]\right)$.} Under C\ref{assn:bern-des}-C\ref{assn:spacing} and if $N\times \AVar(\widehat{\tau}_{\HA}(d))$ is uniformly bounded below for all large $N$, we have, for each $\alpha <1$,
\begin{enumerate}[label=(\roman*)]
	\item $\lim_{N\to\infty}\mathbf{Prob}\left(z_{\frac{\alpha}{2}}\leq \frac{\widehat{\tau}_{\HA}(d)-\textnormal{AME}(d;\eta) }{\sqrt{	\widehat{\textnormal{V}}_{\SAH}(d)}} \leq z_{1-\frac{\alpha}{2}}\right)\geq 1-\alpha$;
	\item additionally under C\ref{assn:homo}, $\lim_{N\to\infty}\mathbf{Prob}\left(z_{\frac{\alpha}{2}}\leq \frac{\widehat{\tau}_{\HA}(d)-\textnormal{AME}(d;\eta) }{\sqrt{	\widehat{\textnormal{V}}_{\HAC}(d)}}\leq z_{1-\frac{\alpha}{2}}\right)\geq 1-\alpha$.
\end{enumerate}
\end{prop}

	One possible concern for the HAC variance estimator is the possibility of negative estimates. This occurs in our simulation when estimating AMEs with large d values. For regular grids like $\mathbb{Z}^2$ and Euclidean distance metric, one can design positive semidefinite HAC variance estimators \citep{neweywest1987,conley99_spatial}. However, for irregular grids and arbitrary distance metric, we are not aware of a general method for creating an \textit{exact} positive semidefinite HAC variance estimator. Nonetheless, it is possible to design a biased-upward positive definite estimator.\footnote{This strategy has been employed by \cite{gao2023causal} and \cite{chang2023design}. For a discussion of positive-definite HAC estimator with Euclidean distance, see \cite{kelejian2007hac}.} The procedure is outlined below.

Note that the HAC variance estimator can be expressed as
\begin{equation}
 \widehat{\Sigma}_{\HAC}(d)= (\X' \X)^{-1}\left(  \X' \left(\left(\widehat{e}\widehat{e}'\right) \circ \mathcal{K} \right) \X \right)  (\X' \X)^{-1},
\end{equation}
where  $\widehat{e}=\left(\widehat{e}_1, ...., \widehat{e}_n  \right)$, $\mathcal{K}$ is a $n$-by-$n$ symmetric matrix with $\mathcal{K}_{ij}=K(\frac{d_{ij}}{\tilde{d}})$, and $\circ$ denotes the pointwise (Hadamard) matrix product. Denote the eigenvalue decomposition of the matrix $\mathcal{K}$ as $\mathcal{K}=\sum_{i=1}^n \lambda_i v_iv_i'$, where $\{\lambda_i\}_{i=1}^n$ are the eigenvalues and $\{v_i\}_{i=1}^n$ are the eigenvectors. We define $\mathcal{K}^{PD}=\sum_{i=1}^n \max\{\lambda_i,0\} v_iv_i' $, and the corresponding positive semidefinite variance estimator
\begin{equation}\label{formula:HAC_PD}
	\widehat{\Sigma}_{\HAC}^{\PD}(d)= (\X' \X)^{-1}\left(  \X' \left(\left(\widehat{e}\widehat{e}'\right) \circ\mathcal{K}^{\PD} \right) \X \right)  (\X' \X)^{-1},
\end{equation}
We note that $\left(\widehat{e}\widehat{e}'\right) \circ \mathcal{K}^{\PD} $ is a positive semidefinite matrix, and hence so is the estimator $	\widehat{\Sigma}_{\HAC}^{\PD}(d)$.\footnote{This follows from Theorem 7.5.3 in \cite{horn2012matrix}, and the fact that both $\widehat{e}\widehat{e}'$ and  $\mathcal{K}^{\PD}$ are positive semidefinite.} In addition, $	\widehat{\Sigma}_{\HAC}^{\PD}(d)\geq \widehat{\Sigma}_{\HAC}(d)$.\footnote{This follows because $\left(\widehat{e}\widehat{e}'\right)\circ \mathcal{K}^{\PD} - \left(\widehat{e}\widehat{e}'\right)\circ \mathcal{K}=\left(\widehat{e}\widehat{e}'\right)\circ\left( \mathcal{K}^{\PD} - \mathcal{K}\right) $, and  $\mathcal{K}^{\PD} - \mathcal{K}$ is a positive semidefinite matrix. } We shall refer to this positive semidefinite variance estimator as \textit{the HAC-PD estimator}. We investigate its performance in the simulation section. In general, the HAC-PD estimator returns nonnegative variance estimates with little loss in efficiency.

In addition to various variance estimators, we include a discussion of empirical degree of freedom (edof) adjustment in Appendix \ref{app:edof}. We find in our simulations that the edof adjustment is important for improving the finite sample performance of the confidence intervals. Since the derivation is fairly standard \citep{imbens-kolesar2012-robust-small,bell_mccaffrey2002_variance,young-2015-improved-inference}, we leave the derivation in the appendix.

\section{Extensions}

\subsection{Structural Interpretation of the AME}\label{struc-int}
Recall that the AME can be interpreted as the average effect of switching an intervention node from control to treatment, given ambient interference emanating from other intervention nodes.  The degree of such ambient interference is dictated by the experimental design and in particular the level of treatment saturation ($p$).  Generally speaking, the AME is not invariant with respect to the experimental design.  Here we show that the AME can have a structural interpretation (i.e., invariant over designs) if spatial effects are additive.  This particular case aligns with standard model-based spatial analyses \citep{darmofal2015-spatial-book}.

Suppose that for each outcome node $x$, its potential outcome value is generated additively: 
\begin{equation}
Y_x(\Z) = \sum_{i=1}^N Z_i g_i(x) + f(x),
\end{equation}
where $f(x)$ captures spatial trends in the absence of any intervention, and then $g_i(x)$ captures effects that emanate, perhaps idiosyncratically, from each of the intervention nodes. This model covers a wide variety of more restrictive models of homogeneous spatial effects.

Under this restriction on the potential outcomes, we have that the effect of assigning treatment to an intervention node $i$ shifts outcomes at point $x$ by $g_i(x)$:
\begin{equation}
	\begin{split}
\tau_{ix}(\eta) = & \E_{\mathbf{Z_{-i}}}\left[Y_x(1, \Z_{-i})\right] - \E_{\mathbf{Z_{-i}}}\left[Y_x(0,\Z_{-i}) \right] \\
= & \E_{\mathbf{Z_{-i}}}\left[g_i(x) + \sum_{j \ne i}^N Z_j g_j(x) + f(x) \right] - \E_{\mathbf{Z_{-i}}}\left[\sum_{j \ne i}^N Z_j g_j(x) + f(x) \right] \\
= & g_i(x).
	\end{split}
\end{equation}
The circle-average at distance $d$ from intervention node $i$ would be equal to the added effect that emanates from node $i$:
\begin{equation}
\tau_i(d;\eta) = \frac{\int_{x :d_i(x) \in \Omega_d } g_i(x) \text{d}\zeta}{\int_{x:d_i(x) \in \Omega_d} \text{d}\zeta}.
\end{equation}
Unlike the general case, this does not depend on the distribution of treatments over intervention nodes other than $i$.  The AME for distance $d$ is then the average of the ways that each intervention point individually affects outcomes at distance $d$, regardless of the treatment assignment.  Thus, we can interpret the AME as a structural quantity with such an additive potential outcome model. 

\subsection{Smoothing}

The AME at a particular distance $d$ with $\Omega_{d}=\{d\}$, as defined in (\ref{def:circleaverage}), may be a noisy quantity to estimate well in practice. When the AME curve as a function of $d$ is considered to be smooth, it may be a good idea to estimate an alternative quantity, the  \textit{smoothed} AME at d, defined as:
\begin{equation}
	\textnormal{sAME}(d;\eta)= \int_{\mathbb{R}} \textnormal{AME}(t;\eta)K_h\left(\frac{d-t}{h}\right)dt,
\end{equation}
where $K:\mathbb{R}\to \mathbb{R}^+$ is a nonnegative kernel and $h$ is a  user-specified bandwidth.\footnote{ $h$ can be tuned  on auxiliary data so that they are fixed with respect to $\Z$.  } The integral can be similarly defined for measures with discrete supports.  This quantity is a smoothed version of the AME function, defined with respect to a chosen kernel function and the bandwidth. The donut and disk AMEs defined in Section \ref{ame} are special cases with (properly-normalized) uniform kernels. 

We can define a similar smoothing operation on observed  circle means: 
\begin{align}\label{SAME:regression}
	\mu_i^{\SM}(\Y;d)=  \int_{\mathbb{R}} \mu_i(\Y;t)K_h\left(\frac{d-t}{h}\right)dt.
\end{align}
With this definition, we can estimate the smoothed-AME using the same regression approach studied in Section \ref{est-inf}. It should be clear that the statistical results developed in Section \ref{est-inf} remain valid provided that the new quantities satisfies the assumptions with proper parameters (e.g., with a larger interference neighborhood).\footnote{An alternative approach is to estimate a kernel-weighted regression: \begin{equation}\label{SAME:kernelregression}
		\left(\widehat{\mu}_0^K(d),  \widehat{\tau}^{K}_{\HA}(d)\right)= \arg \min_{(\mu,\tau)}   \sum_{i=1}^N \sum_{d' \in \mathcal{D}} \left(\mu_i(\Y(\z);d') - \mu- \tau Z_i \right)^2 K_h\left(\frac{d-d'}{h}\right),
	\end{equation} 
	and/or consider a more refined estimation strategy: 	\begin{align}\label{SAME:llregression}
	(\widehat{\mu}, & \widehat{\tau}, \widehat{\beta}, \widehat{\delta}) =\arg \min_{(\mu,\tau,\beta,\delta)} \sum_{i=1}^N \sum_{d' \in \mathcal{D}} \left(\mu_i(\Y(\z);d') - \mu - \tau Z_i -\beta (d'-d) - \delta Z_{i} (d'-d)\right)^2 K_h\left(\frac{ d-d'}{h}\right).
	\end{align}	 	
	The estimation theory for (\ref{SAME:kernelregression}), based on our setup, is similar to that of (\ref{SAME:regression}) and we omit here. We leave the development of the estimation theory of  (\ref{SAME:llregression}) for a future study.
 }

\subsection{Randomization Tests}\label{Section:permutation_test}
In our analysis above, we discussed how to construct pointwise confidence intervals for the AME values at different distance values. An alternative for testing is the randomization test, albeit under a stronger sharp null hypothesis. In addition to pointwise tests, the randomization test can be flexibly adapted to test other type of hypothese, for example, researchers may  be interested in whether effects are statistically significant on a particular interval rather than at some point.  For this purpose, one can use a randomization test with test statistics $\max_{d\in[d_1,d_2]}\widehat{\tau}_{\HA}(d)$.

Under the sharp null hypothesis that $Y_x(\mathbf{z}) = Y_x(\mathbf{0})$ for any $\mathbf{z}$, we know the full distribution of potential outcomes. Denote the statistic of interest as $T(\mathbf{Y}, \mathbf{Z})$. Examples include estimates of the AMEs at each distance value or the average/maximum of such estimates on an interval $[d_1, d_2]$. As all the potential outcomes are known under the sharp null, we can resample the assignment $\mathbf{z}$ for $P$ times and calculate the corresponding $T(\mathbf{Y}, \mathbf{Z}_p)$, $p=1,...P$. The resampling distribution of $\{T(\mathbf{Y}, \mathbf{Z}_p)\}_{p=1,...P}$ will approximate the distribution of $T(\mathbf{Y}, \mathbf{Z})$ under the sharp null.  As a result, rejecting the null if $\frac{1}{P}\sum_p \mathbf{1}\{|T(\mathbf{Y}, \mathbf{Z}_p)| \geq T(\mathbf{Y}, \mathbf{Z})\} \leq \alpha$ for some fixed large enough $P$ gives an $\alpha$-level test of the sharp null hypothesis. We include simulation results for a class of randomization tests in Section \ref{Section:permutation_test}.


\subsection{Observational studies}\label{Section:ObservationalStudies}

Our framework can be  generalized to observational studies. We first comment on this extension. In many observational studies, potential outcomes, treatment status, and confounders are assumed to be drawn from a superpopulation and uncertainties are assessed with respect to all three components. In our generalization, we treat the potential outcome and confounders as fixed and only study the uncertainty arising from the treatment assignments. In this way, we focus on estimating a \textit{sample} AME instead of a population AME. Estimating a population AME in our setting will require more assumptions on the data generating process of the potential outcomes.

Recall our notations for the evaluation set $\mathcal{X}$ and potential outcomes $Y_x(\z)$ in Section \ref{Section:setting}.  Additionally, we denote the collection of confounders for intervention node $i$ as $\C$, and make the following independent assignment assumption.
\begin{assn}{(Probablistic Assignment)}\label{assn:assignment_obsstudy}
For all sample sizes $N$,
\begin{enumerate}[label=(\roman*)]
	\item The  random assignment variables $\{Z_i\}_{i=1}^N$ are jointly independent.
	\item There exists a treatment probability model  $p( \cdot): \mathcal{O} \to [0,1]$ such that
	for each $i\in\node_N$,
	\begin{equation}
	\mathbf{Prob}(Z_i=1)= p(\C)
	\end{equation}
	\item For an $\epsilon\in (0,\frac{1}{2})$ and for all $i\in\node_N$,  $\mathbf{Prob}(Z_i=1)\in (\epsilon,1-\epsilon)$ .
\end{enumerate}
\end{assn}
C\ref{assn:assignment_obsstudy}-(i) assumes that the treatments are assigned independently. C\ref{assn:assignment_obsstudy}-(ii) states that the $i$th node's treatment probability can be fully described by the confounders $\C$, independent of the potential outcomes. The implication that the assignments are independent of the potential outcomes conditioning on the confounders is the same as what the unconfoundness assumption establishes for observational studies under a super-population assumption. C\ref{assn:assignment_obsstudy}-(ii)  implies that we can model the treatment probability solely as a function of $\C$ and that the probability of being treated at node $i$ is not predicted by the potential outcomes at any evaluation point $x$.
Under C\ref{assn:assignment_obsstudy}, one can show that the Horvitz-Thompson estimator with known propensity scores, defined as\footnote{We include the proof in Appendix \ref{Section:HTobsproof}.}
 $$
 \widehat{\tau}_{\HT}^{\obs}(d) = \frac{1}{N}\sum_{i=1}^N \frac{Z_i}{p(\C) } \mu_i(\Y; d) - \frac{1}{N}\sum_{i=1}^N \frac{1-Z_i}{p(\C)} \mu_i(\Y; d).
 $$
 is unbiased for $\textnormal{AME}(d;\eta)$.\footnote{Note that the AME is defined by marginalizing over the independently but not identically distributed assignment variables generated by the assignment mechanism. }



In practice, $p(\C) $ is unknown to researchers. In such cases, parametric methods, such as logistic regression, and nonparametric methods, such as the sieve estimator in \cite{hirano2003efficient}, can be used to estimate the propensity score. In Appendix \ref{appendix:observationalstudy}, we develop a complete inferential theory for the inverse probability weighted (IPW) estimator where propensity scores are modeled using a logistic model. Results there include: i) additional assumptions; 2) asymptotic linear expansion and asymptotic distribution characterization; 3) variance estimation and inference.


\subsection{Weaker assumptions on the extent of interference}

The limited interference assumption C\ref{assn:local-inf} can be relaxed to accommodate the cases where the potential outcome of one intervention node is affected by all intervention nodes but the effect decreases as the distance between the nodes increases. In Section \ref{Section:weak_interference}, we extend inferential results on the Hajek estimator by relaxing the limited interference assumption C\ref{assn:local-inf}. We follow the literature on spatial near-epoch dependence \citep{jenish2012spatial} and provide results on root-N consistency, asymptotic normality and HAC variance estimation.\footnote{We avoid appealing to the results in the spatial mixing literature, for example, as in \cite{jenish_prucha2009_clt_spatial_interaction}. The mixing condition may be too stringent to be satisfied in design-based causal inference settings, in which outcomes are modeled as a function of Bernoulli random variables. For example, see \cite{andrews1984non} and \cite{doukhan2002rates}.   }  


\section{Simulation}\label{Section:Simulation}
In this section, we use simulated datasets to illustrate propositions introduced in the previous sections and examine the performance of inferential methods based on our analytical results. In the main text, we present simulation results of a point intervention with two different effect functions.\footnote{ In Appendix \ref{Section:Simulation_Appendix}, we present simulation results for a polygon intervention. In the same appendix, one can find additional information on data generating process for the the simulation datasets.} For each simulation design, we run simulations with three sample sizes 64, 100 and 144. 

For the first simulation scenario, the effect function is non-monotonic and additive. Let $Y_x(0)$ be the control outcome at an outcome point $x$, the outcomes are generated as:
\begin{equation}\label{additive_effect}
	Y_x(\Z) = Y_x(0)+ \sum_{i=1}^n f_x(d_{ix})Z_i
\end{equation}
where $n$ denotes the number of intervention nodes, $Z_i$ is the treatment status of $i$th intervention node, and $d_{ix}$ is the distance from the outcome point $x$ to the intervention node $i$. $f_x(\cdot)$ is an effect function and is constructed by mixing the density of two gamma-distributions.

For the second simulation scenario, the effect is interactive. The outcome at the outcome point $x$ is generated by
\begin{equation}\label{interactive_effect}
	Y_x(\Z) = Y_x(0) + \sum_{i=1}^{n} f_x(d_{ix})Z_i + \sum_{i=1}^{n} g_x(d_{ix})Z_i Z_{\mathcal{N}(i)},
\end{equation}
where $\mathcal{N}(i)$ denotes the intervention node that is closest to the intervention node $i$ and $g_x(\cdot)$ is an additional effect function. This design reflects the story that the treatment effect may be stronger when two nearby nodes are treated. 

The AME curves for both cases are shown in Figure \ref{fig:AMR_effect}. The left figure displays the additive-effect case and the right figure the interactive-effect case. When effects are additive, the effect curve (i.e., $f_x(\cdot)$ in (\ref{additive_effect})) and the AME curve are the same as expected.  This follows from our analysis of the structural interpretation of the AME above. Th interactive effect function emanating from a treated intervention node has the same shape as the additive one only when its nearest neighbor is not treated. Otherwise, it is monotonic. Therefore, the AME curve looks like the average of two effect functions.

\begin{figure}[H]\centering
\includegraphics[width=.4\textwidth]{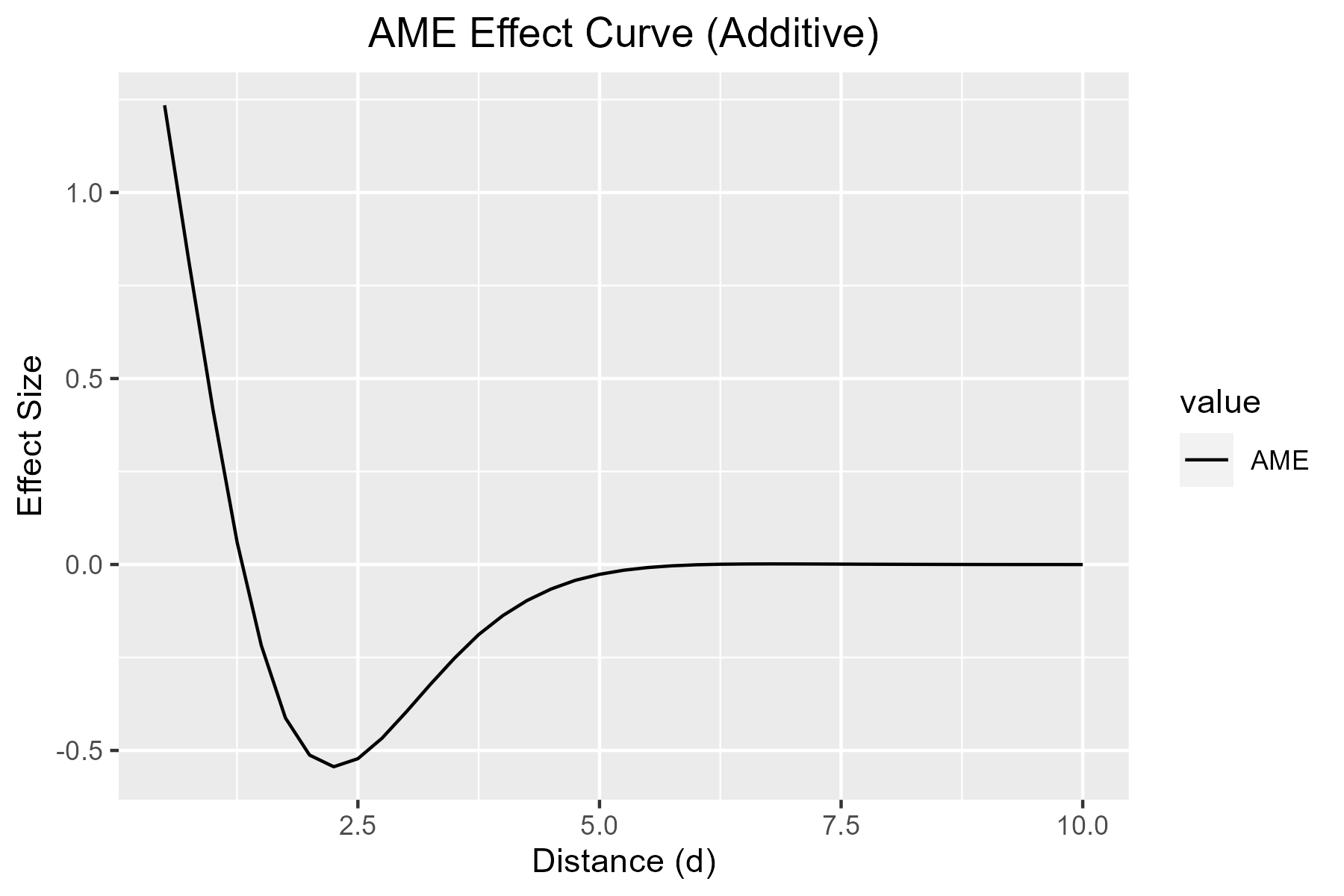}
\includegraphics[width=.47\textwidth]{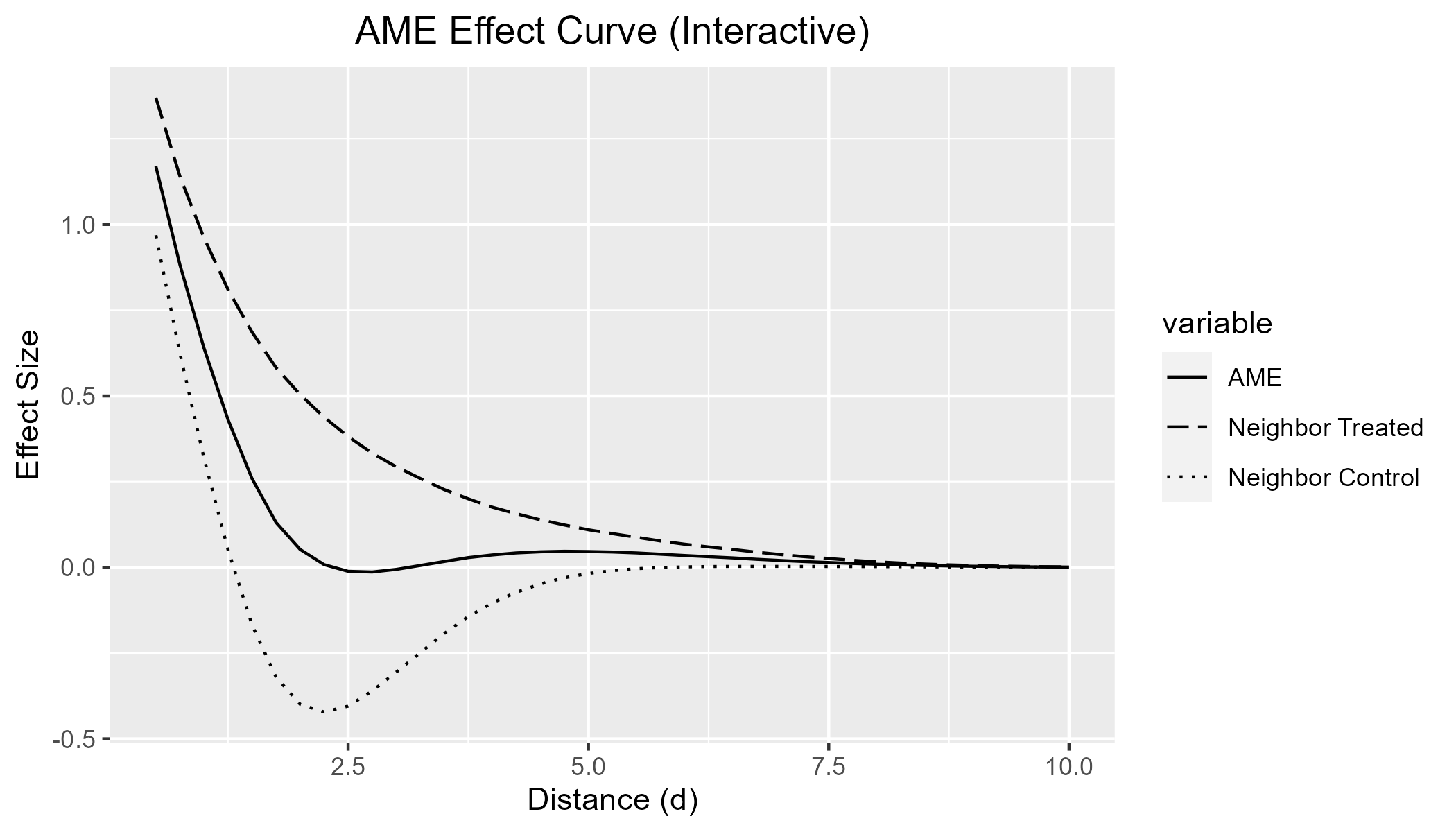}
\caption{This figure displays the AME curves in solid lines for both the additive case (\ref{additive_effect}) and the interactive case (\ref{interactive_effect}). For the interactive case, the dashed lines are the marginalized effect curves when the nearest neighbor is treated and when it is not. }
\label{fig:AMR_effect}
\end{figure}

Figure \ref{fig:MSE_Hajek_point} shows the Mean Squared Errors of the Hajek estimator for both additive-effect (\ref{additive_effect}) and interactive-effect cases (\ref{interactive_effect}),  with intervention-node sample sizes of 64, 100, and 144. In both cases the MSEs decrease as sample sizes increase, as predicted by our theory.

Figure \ref{fig:AME_hajek_additive_point} and Figure \ref{fig:AME_hajek_interactive_point} report coverage rates and median half-lengthes for the Hajek estimator with different confidence interval constructions in the additive-effect case  and the interactive-effect case,  respectively. For brevity, we only display the case with a sample size of 144. The results are illustrated for AMEs at different distance values. 

We highlight two observations from the these figures. Firstly, for small distance values, all confidence intervals have proper coverage rates. However, for large distance values, the empirical degree of correction is important to improve finite-sample performance. This happens because for large distance values, the effective sample size becomes small and some finite sample adjustment is necessary to better reflect the randomness with a small sample. Secondly, the confidence interval procedure based on the SAH variance estimator tends to be overly conservative. The confidence interval generated by the  positive semidefinite HAC variance estimator is not significantly longer than the one with the HAC variance estimator. Based on these observations, we recommend researchers to use the positive semidefinite HAC variance estimator with the empirical degree of adjustment when applying our method.

\begin{figure}[H]
	\centering
	\begin{subfigure}[b]{0.5\textwidth}
		\centering
		\includegraphics[scale=0.5]{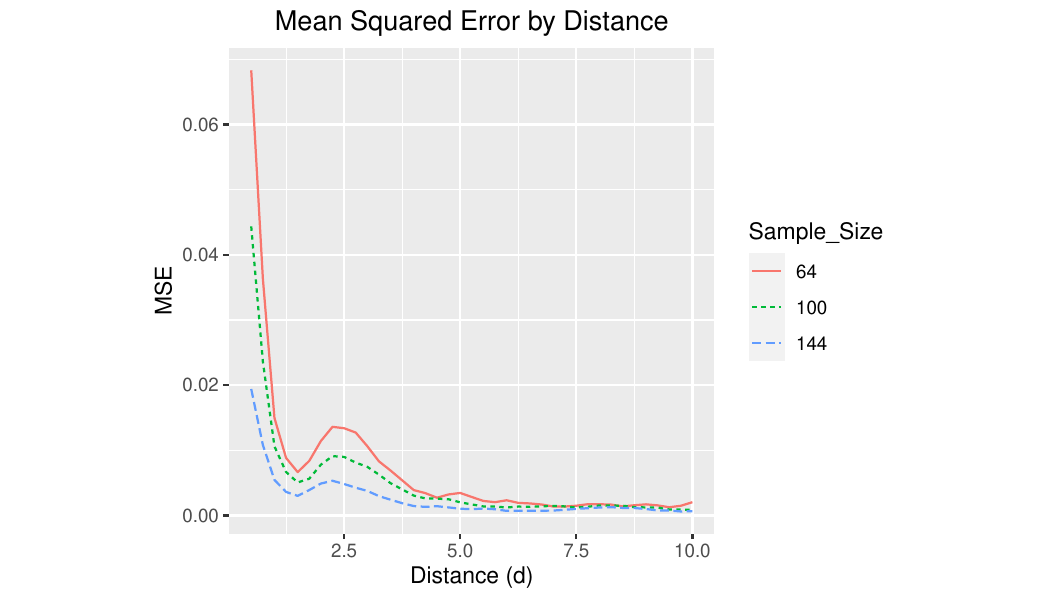}
		\caption{MSE for the additive case (\ref{additive_effect})}
	\end{subfigure}%
	~ 
	\begin{subfigure}[b]{0.5\textwidth}
		\centering
		\includegraphics[scale=0.5]{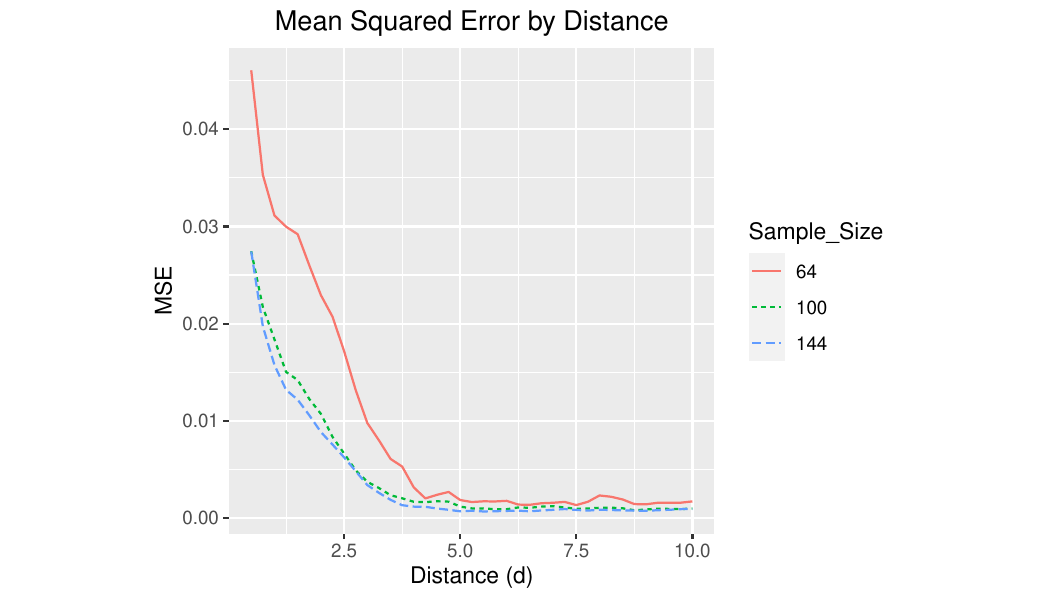}
		\caption{MSE for the interactive case (\ref{interactive_effect})}
	\end{subfigure}
	\caption{The left and right figures report the Mean Squared Errors of the Hajek estimator in the additive-effect case and the interactive-effect case, respectively.  }
	\label{fig:MSE_Hajek_point}
\end{figure}

\begin{figure}[H]
    \centering
\includegraphics[scale=0.45]{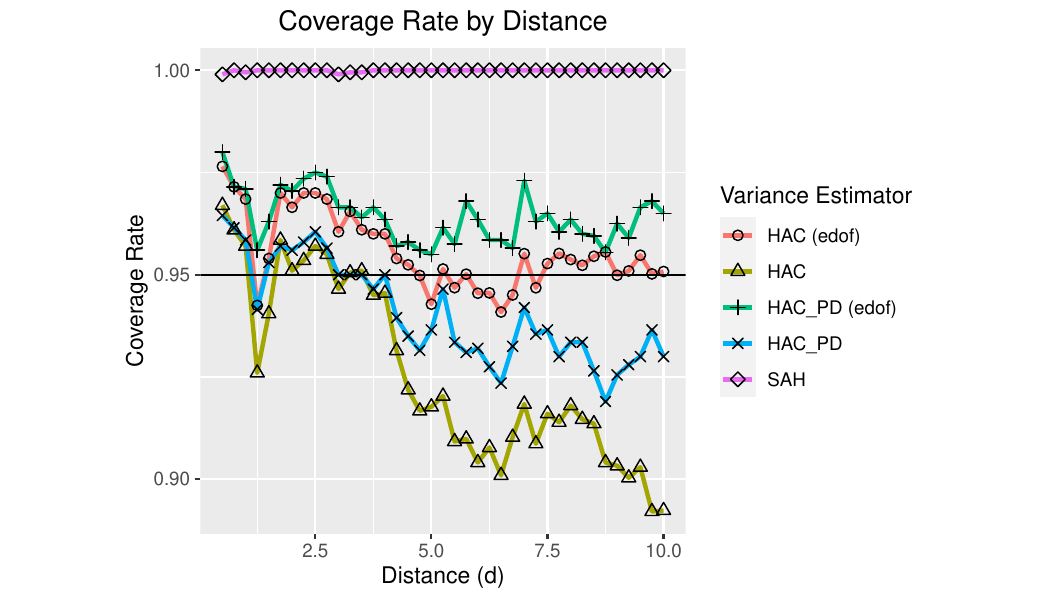}
\includegraphics[scale=0.45]{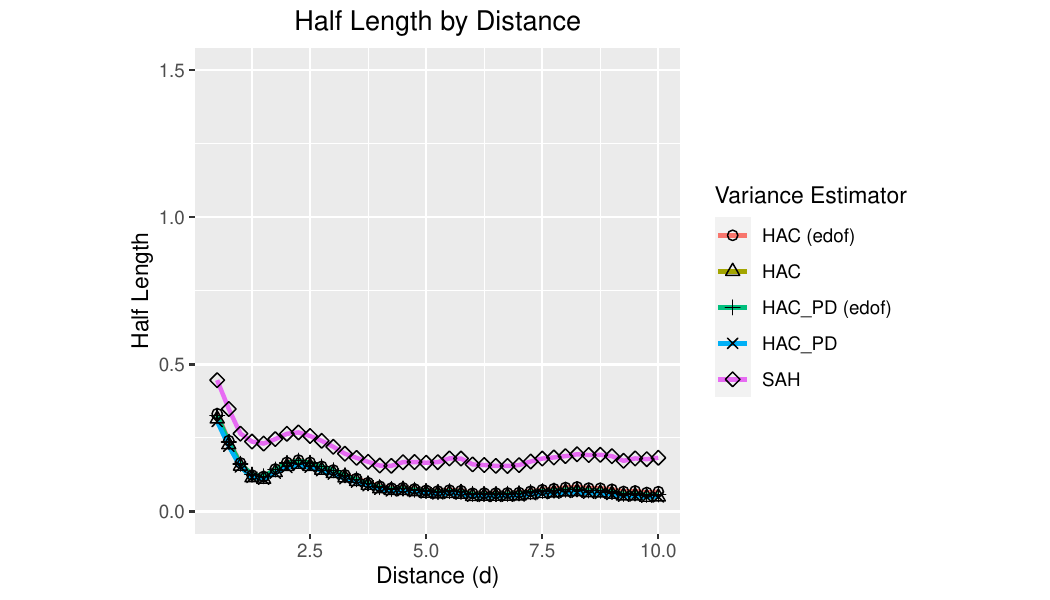}
\caption{ Point-intervention simulation results on the coverage rates and half lengthes of two-sided 95\% confidence intervals with the Hajek estimator and different variance estimators in the additive effect case (\ref{additive_effect}). The sample size is 144. HAC refers to the CI with the HAC variance estimator in (\ref{HAC_formula}) and a normal critical value. The length and coverage of the HAC CI is assesed with respect to the cases where HAC estimator returns a nonnegative value. HAC\_PD refers to the CI with positive-semidefinite HAC variance estimator in (\ref{formula:HAC_PD}) and a normal critical value.  HAC (edof) refers to the CI with HAC variance estimator in (\ref{HAC_formula}) and empirical degree of freedom adjustment. HAC\_PD (edof) refers to the CI with HAC variance estimator in (\ref{formula:HAC_PD})  and empirical degree of freedom adjustment. SAH refers to the CI with  SAH variance estimator (\ref{varianceSAH}).  }
\label{fig:AME_hajek_additive_point}
\end{figure}

\begin{figure}[H]
	\centering
	\includegraphics[scale=0.45]{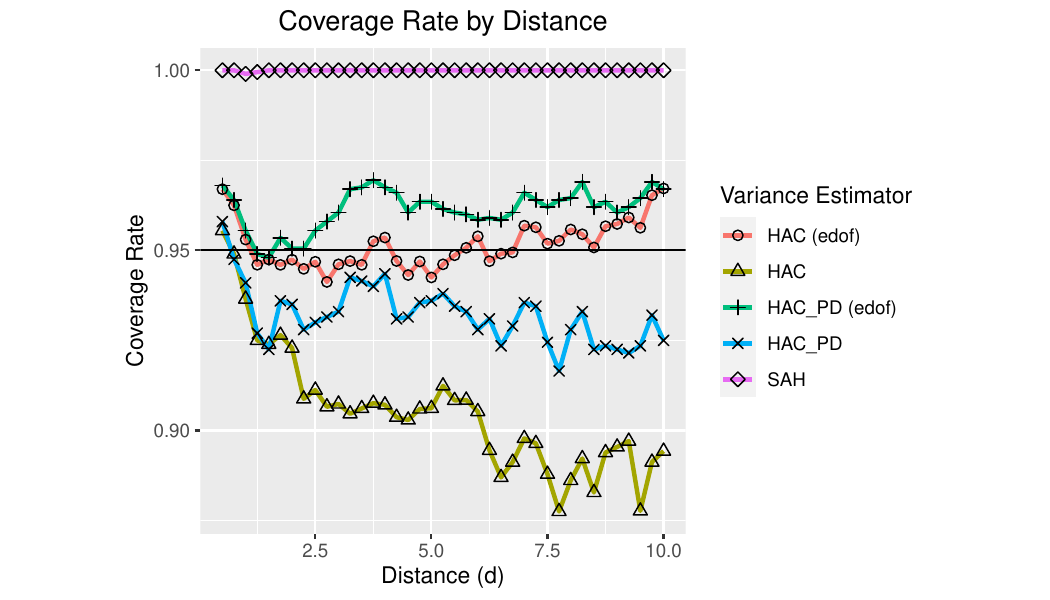}
	\includegraphics[scale=0.45]{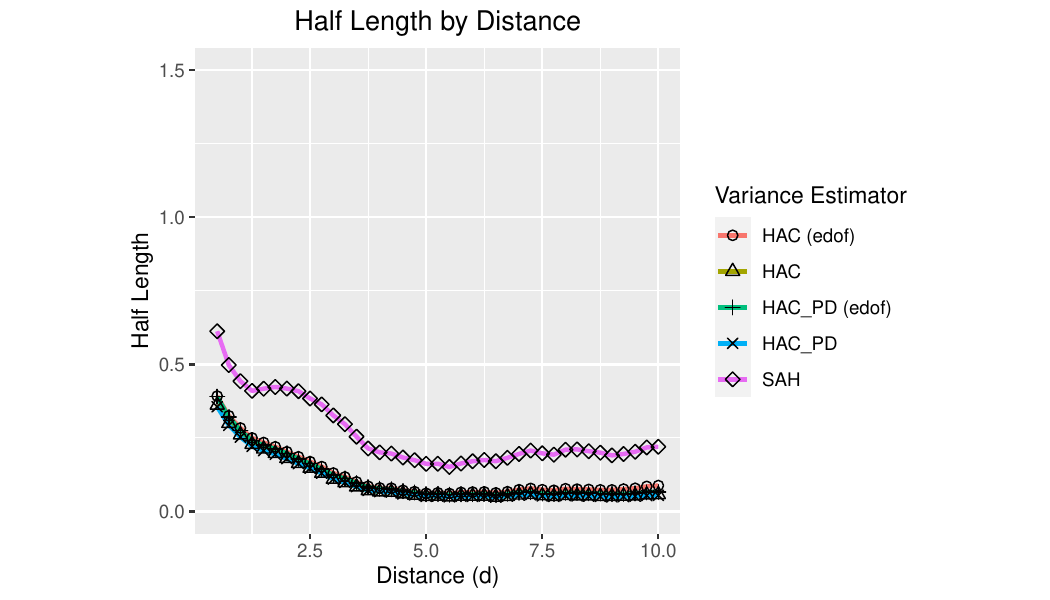}
\caption{ Point-intervention simulation results on the coverage rates and half lengthes of two-sided 95\% confidence intervals with the Hajek estimator and different variance estimators in the interactive effect case (\ref{interactive_effect}). The sample size is 144. HAC refers to the CI with the HAC variance estimator in (\ref{HAC_formula}) and a normal critical value. The length and coverage of the HAC CI is assesed with respect to the cases where HAC estimator returns a nonnegative value. HAC\_PD refers to the CI with positive-semidefinite HAC variance estimator in (\ref{formula:HAC_PD}) and a normal critical value.  HAC (edof) refers to the CI with HAC variance estimator in (\ref{HAC_formula}) and empirical degree of freedom adjustment. HAC\_PD (edof) refers to the CI with HAC variance estimator in (\ref{formula:HAC_PD})  and empirical degree of freedom adjustment. SAH refers to the CI with  SAH variance estimator (\ref{varianceSAH}).  }
	\label{fig:AME_hajek_interactive_point}
\end{figure}

We also evaluate the performance of the randomization tests for testing the sharp null hypothesis, as discussed in Section \ref{Section:permutation_test}. We use the AME estimator at each distance value $d$ as the test statistics and the size of our test is 5\%. Figure \ref{fig:randomization} reports the rejection probability of the tests (pointwise-in-d) in the null-effect scenario,\footnote{That is $Y_i(\Z)=Y_i(0)$ for all $\Z$.}, additive-effect scenario (\ref{additive_effect}), and interactive-effect scenario (\ref{interactive_effect}). It can be seen that under the null-effect case, the rejection probability is around 5\%. In the additive-effect and interactive-effects, the rejection probabilities are high at some locations (i.e., locations with anon-null effect) and approaches 1 as the sample size gets larger for these locations. We also evaluate the performance of the randomization tests based on the statistics $\sup_{d\in\mathcal{D}} |\widehat{\tau}_{\HA}(d)|$.\footnote{$\mathcal{D}$ is a grid from 0.5 to 10 with a step size of 0.25.} With a sample size of 64, the rejection probability is 0.056 in the null-effect case, and 1 in the additive-effect and interactive-effect case. The rejection probabilities remain similar for larger sample sizes 100 and 144.

\begin{figure}[H]
	\centering
	\includegraphics[scale=0.42]{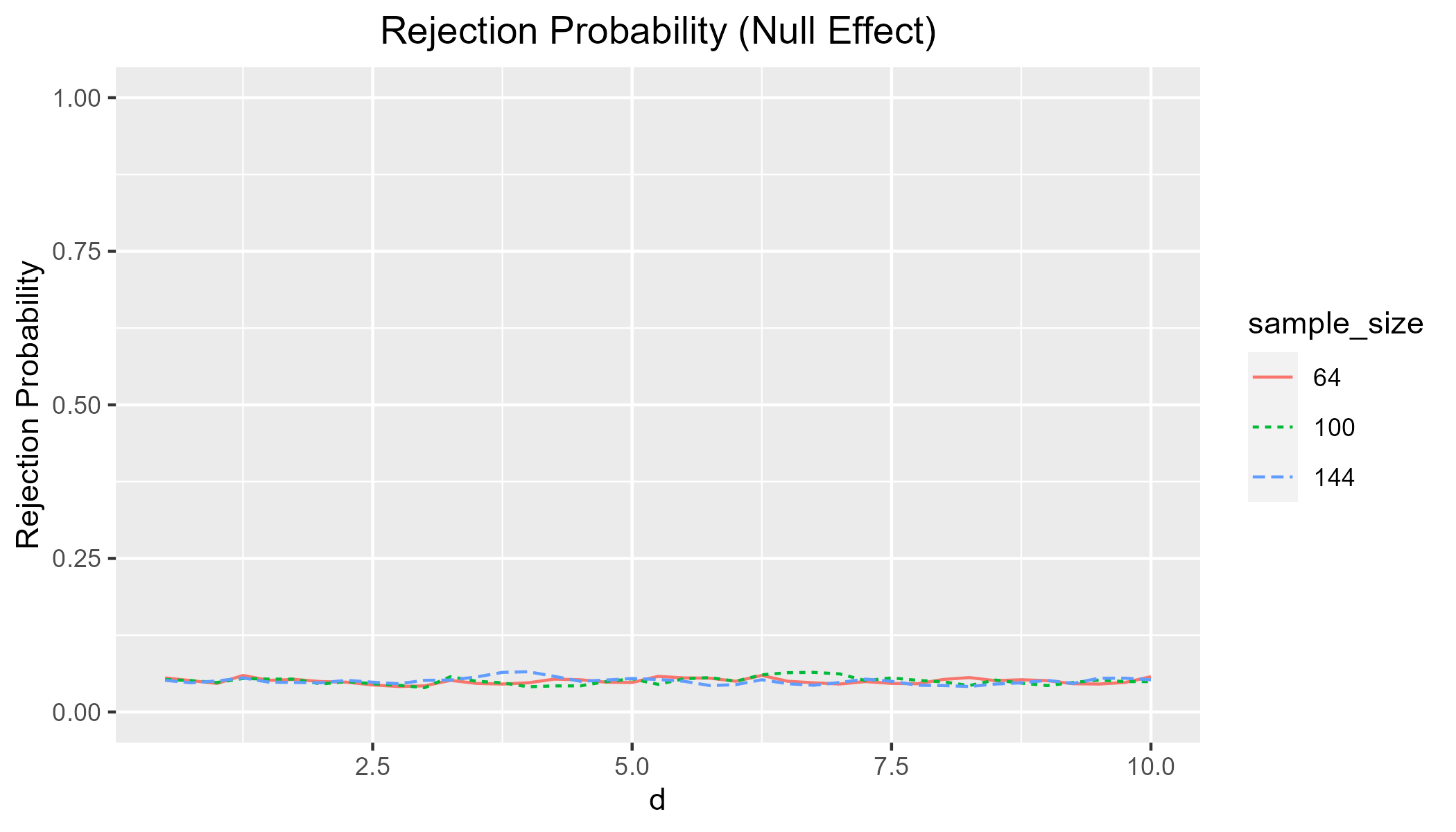}
	\includegraphics[scale=0.42]{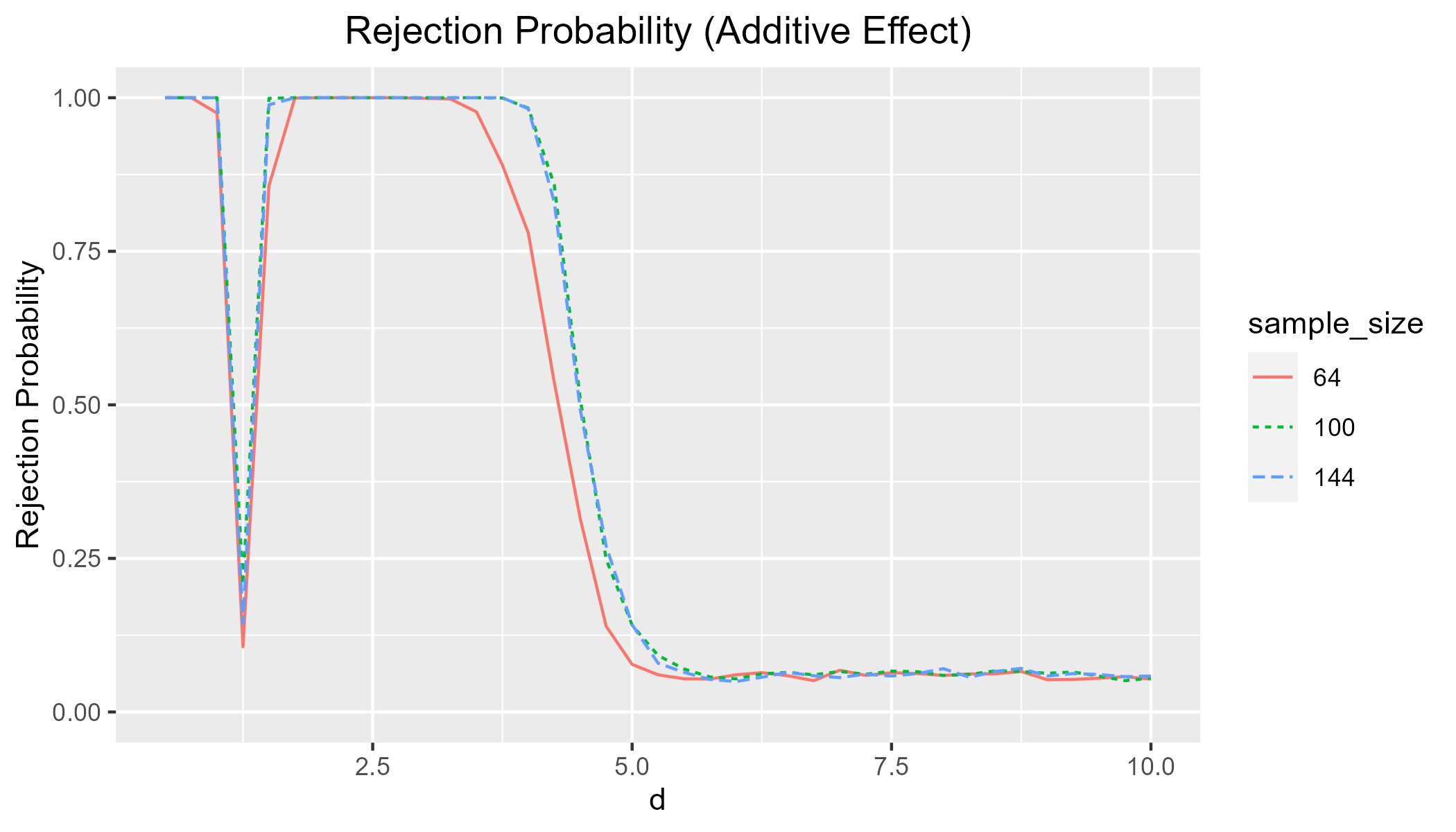}\\
	\includegraphics[scale=0.42]{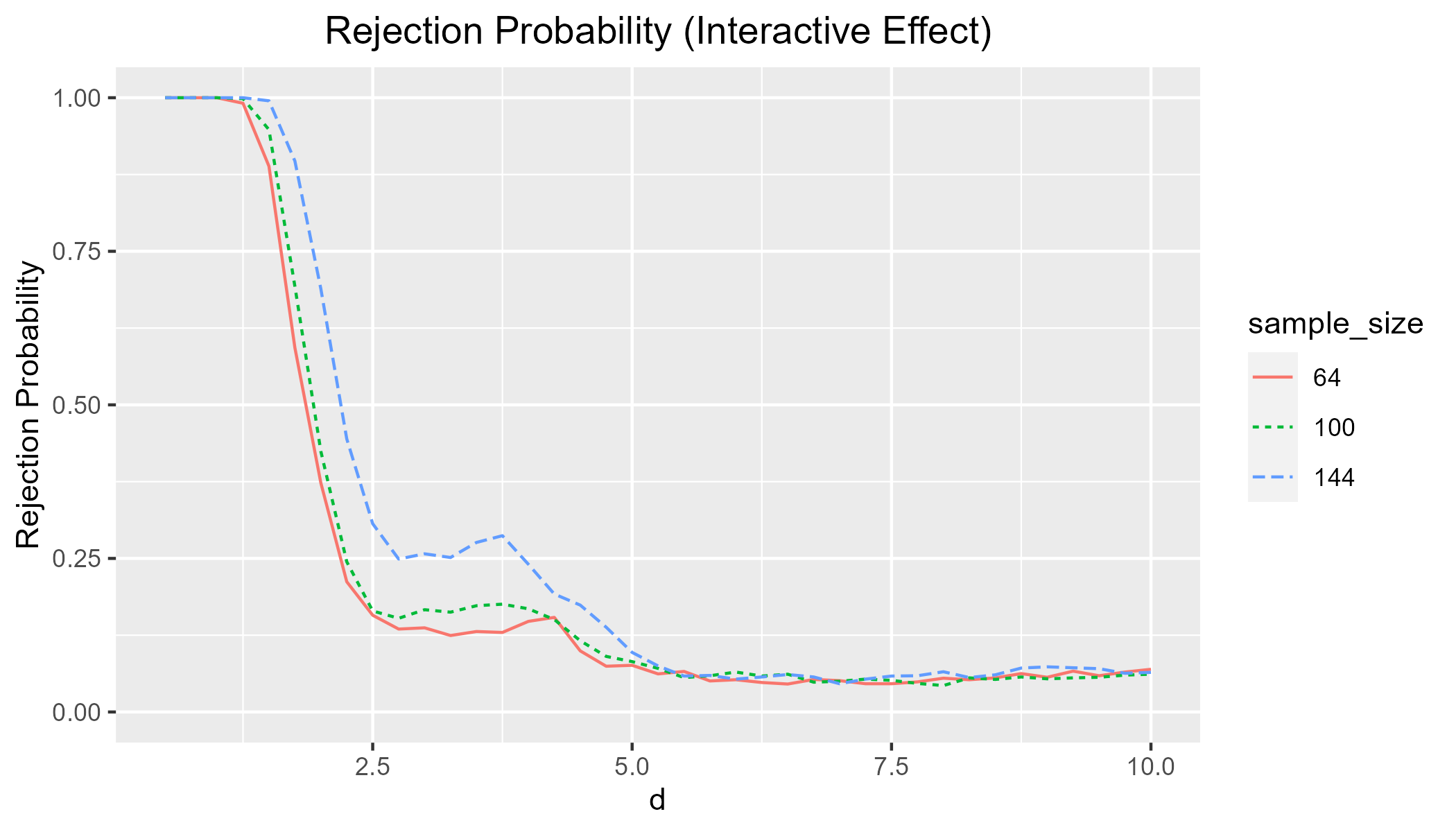}
	\caption{The figures display rejection probabilities of the pointwise randomization tests of the sharp null hypothesis for three simulation scenarios: null effect, additive effect (\ref{additive_effect}), and interactive effect (\ref{interactive_effect}).  }
	\label{fig:randomization}
\end{figure}

\section{Application}

\subsection{A forest conservation experiment \citep{Jayachandran267}}

We now return to the forest conservation experiment from \cite{Jayachandran267}. 
The authors evaluate the effects of a ``payments for ecosystems services'' (PES) program based on 121 villages in Hoima and northern Kibaale districts of Uganda. 60 villages were randomly assigned to the treatment group. Private forest owners in these villages were paid to reduce deforestation on their own land over a course of two years, from 2011 to 2013. Figure~\ref{fig:J-2017-buffers} shows the location of each village in the experiment and its treatment status. 

A primary concern in both academic and policy discussions about PES programs is what conservation scientists refer to as ``leakage,'' which in forest conservation contexts refers to negative spillover effects such that interventions reduce deforestation in targeted locations only to lead to its increase in others \citep{wunder2008we, alix2012forest, samii2014effects}. 
In the area of Uganda in which \cite{Jayachandran267} were researching, private forest owners cleared forest for either agricultural land or for timber sales into local markets. 
As such, the concern would be that forest conservation in targeted areas would cause the private forest owners to shift to clearing in other nearby forests.

\cite{Jayachandran267} originally estimated effects assuming no interference between villages, and only measure the outcome variable (forest cover) within the sampled village boundaries.  
They assessed the potential for leakage by studying whether the {\it beneficial} effects were larger in areas near forest reserves, with the idea being that these would be areas in which farmers would more easily shift forest clearing from their farmland to forest reserve land.
They did not find such a pattern. Jayachandran et al. also examined whether deforestation was higher in control villages that were near treated villages, and found no such pattern. 
This is essentially an ``exposure mapping'' approach, and its validity depends on proper specification of indirect exposure.

We use the methods above to conduct another analysis that also accounts for possible leakage into areas outside the sampled village boundaries. 
To do this, we construct a deforestation outcome variable using forest cover data from \citet{hansen2013high} for years 2012 and 2013. We code a pixel as deforested if a pixel goes from forest coverage rate greater than 25\% in 2012 to one that is below 25\% in 2013.  In Figure~\ref{fig:J-2017-buffers}, the dark spots indicate where deforestation happened. To construct the circle averages, we generate buffers around each of the village polygons.  The distance range for estimating the AME is set to run from 0 km to 15 km.

\begin{figure}[H]
\centering
\includegraphics[width=.42\textwidth]{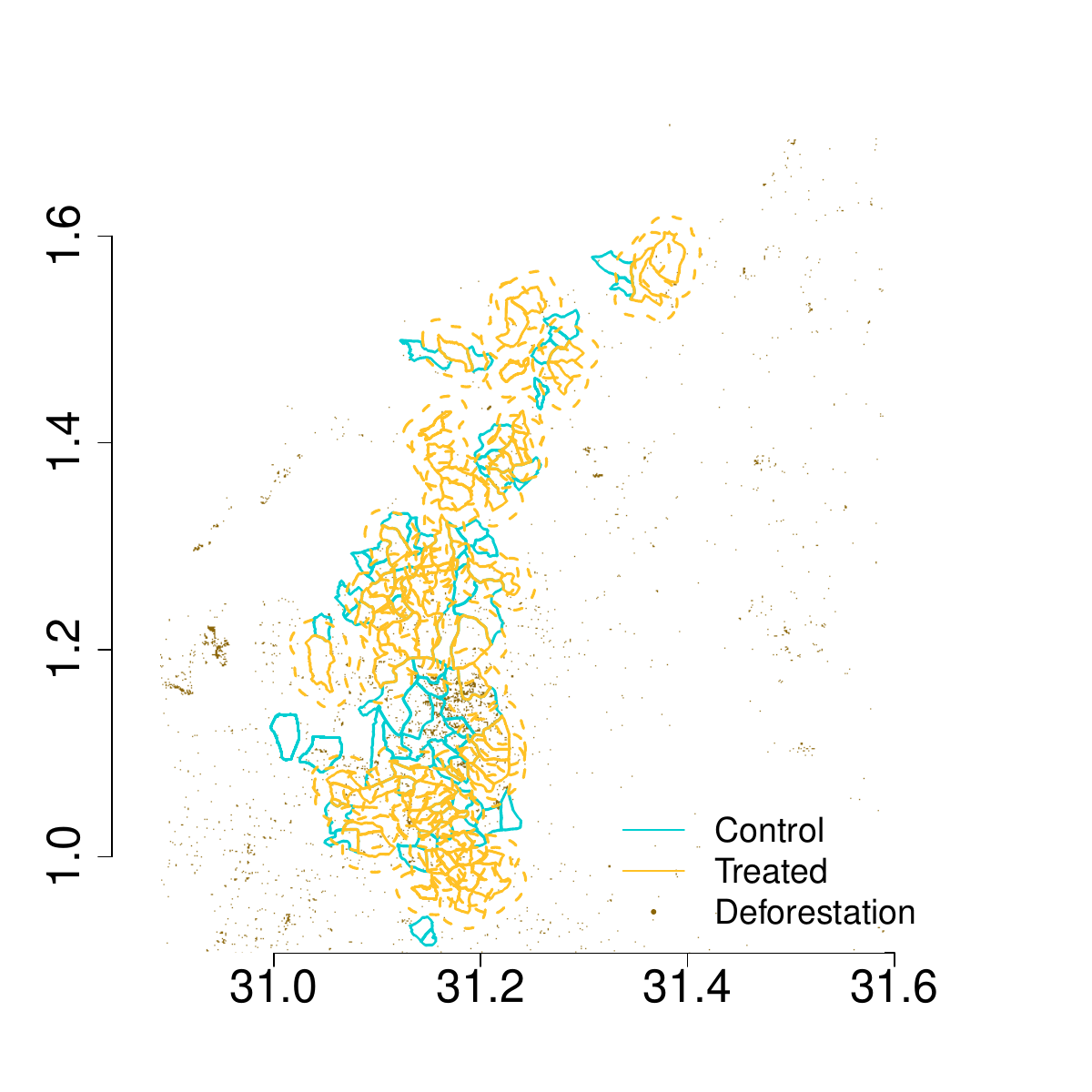}
\includegraphics[width=.42\textwidth]{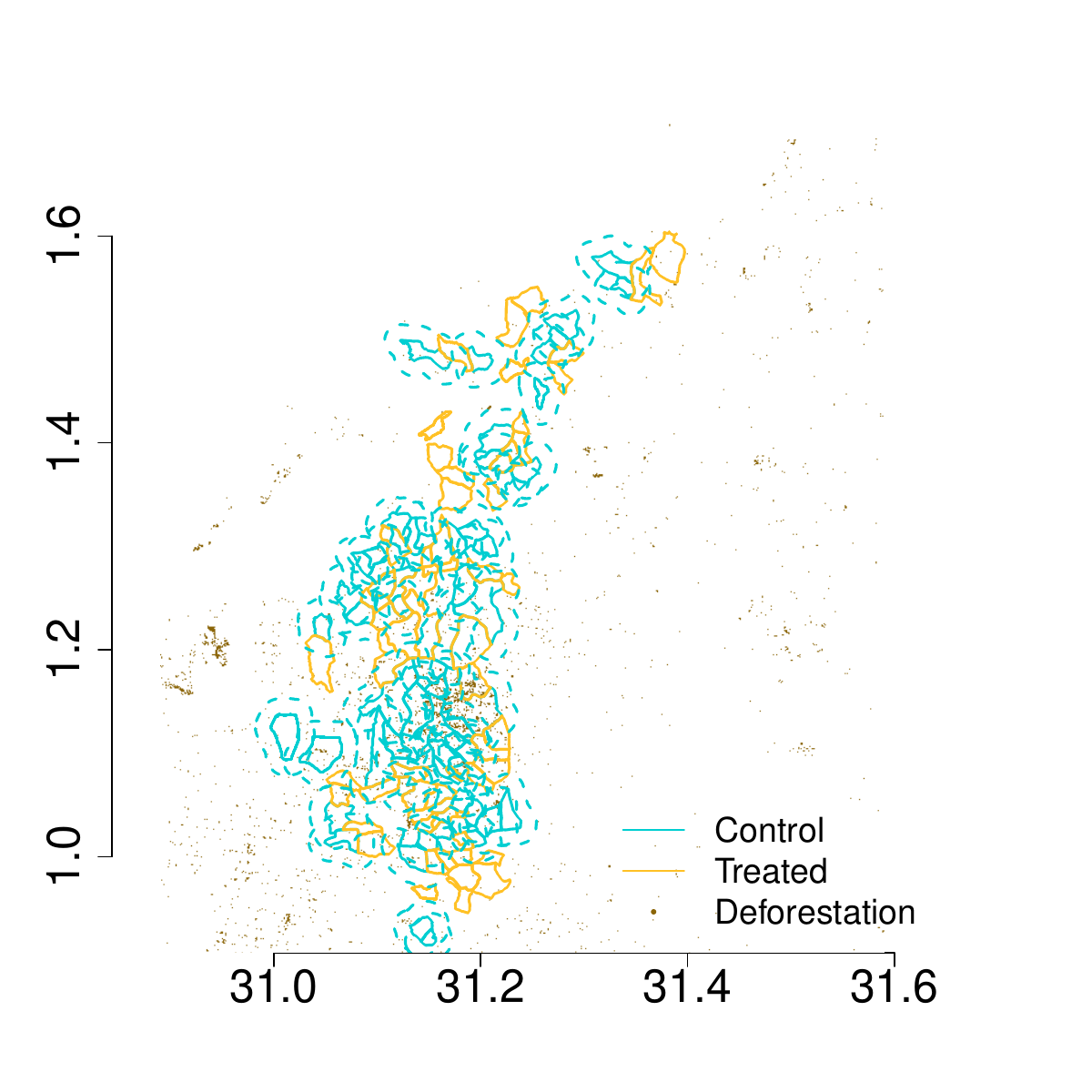}
\caption{Both plots show the boundary of the 121 villages in \citet{Jayachandran267}. Villages with a golden boundary are treated and those with a turquoise boundary are under control. Dark spots on the map represent deforestation during the experiment. The left plot shows buffers around each treated village and the right one shows buffers around each untreated village.}
\label{fig:J-2017-buffers}
\end{figure}

Figure \ref{fig:J-2017-estimates} displays results from the Hajek estimator and the smoothed Hajek estimator.\footnote{We use a triangular kernel with a bandwidth of 5km to construct the smoothed Hajek estimator.} 
The point estimates show decreased deforestation within the treated village boundaries (Distance = 0km), similar to the authors' original analysis.
Then, the estimates for distance values greater than 0 km capture the spatial spillovers. 
The issue that we seek to address is whether there is any indication of leakage---i.e., {\it increases} in deforestation within the vicinity of treated villages.
Given that we do not have a strong substantive basis to select a cutoff value for the spatial HAC variance estimator, we assess the robustness of inferences by considering a range of values (2km, 5km and 10km).\footnote{To be precise, for AME at distance value d, we set the bandwidth at 2$\tilde{d}$+2d, as discussed in Section \ref{Section:inference}.  We set $\tilde{d}=6$. }
We display two-sided 90\% confidence intervals.  
The results do not offer any indication of substantial leakage and the point estimates actually suggest some beneficial spillovers.

Alternatively, we  use a randomization test to evaluate the cumulative  effects between 0 km and 5 km.
The test statistic is the sum of the point Hajek estimates between 0 and 5 km. 
The estimated value is -0.034 (a 3.4 percentage point decline in deforested area). For the randomization test, we use 10000 random draws and the p-value is around 0.065. Thus we would reject with 90\% confidence of the sharp null hypothesis. These results indicate that the net gains from the intervention were indeed beneficial, with no indication of substantial leakage.

\begin{figure}[H]\centering
\includegraphics[width=0.4\textwidth]{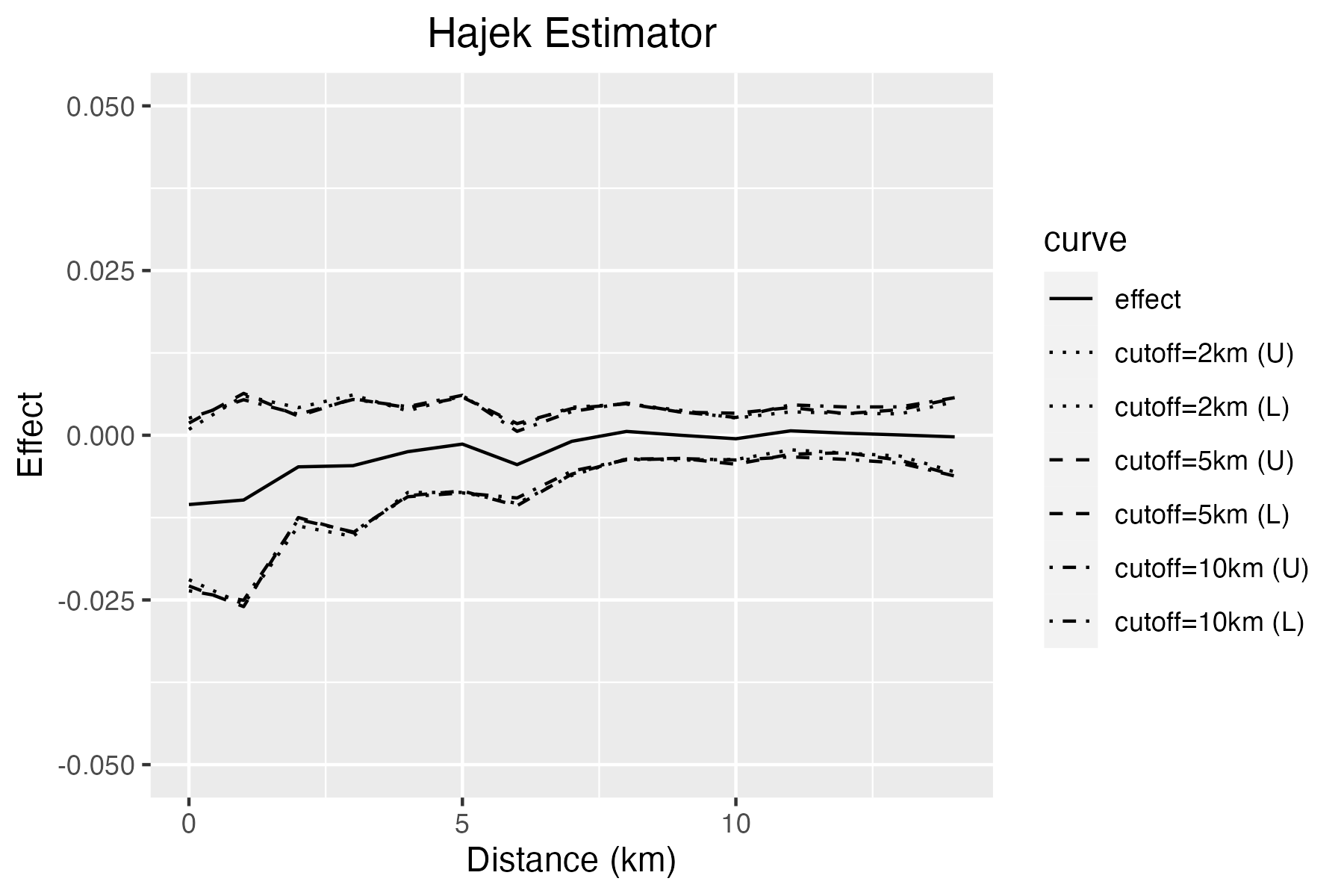}
\includegraphics[width=0.4\textwidth]{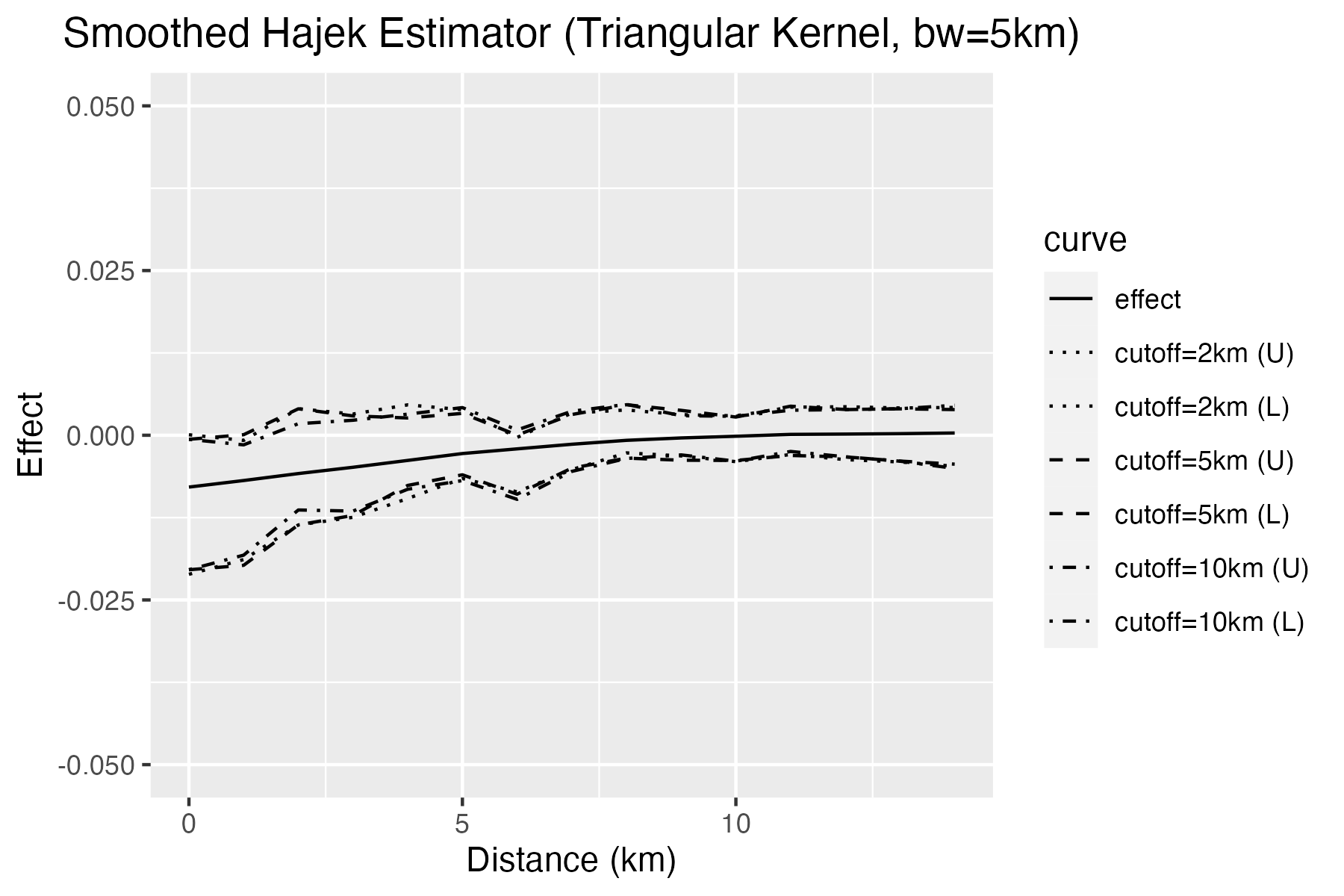}
\caption{Two plots present results based on \citep{Jayachandran267}. The top plot presents results for the Hajek estimator. The bottom plot presents results for the smoothed Hajek estimator with a triangular kernel and a bandwidth of 5km.  We plot pointwise two-sided 90\% intervals based on the positive-semidefinite HAC variance estimators and empirical degree of adjustments with three different cutoffs (2km, 5km, and 10km). (U) indicates the upper end of the interval and (L) indicates the lower end of the interval.
}
\label{fig:J-2017-estimates}
\end{figure}

\subsection{A forest conservation observational study \citep{ferraro2011conditions}}
In this section, we detail results based on \cite{ferraro2011conditions}, an observational study that scrutinized the efficacy of Costa Rica's protected areas in forest conservation. The outcome is measured at the level of parcels with a size of 3 hectares. Each parcel has a value of $1$ if deforestation occurred inside it before 1980. The parcels are incorporated in a raster object for analysis. Originally, the study compared deforestation in parcels within protected areas against those external, after matching them on observable characteristics. This approach neglects potential spillover effects on outer parcels proximate to protected areas. Our method helps to alleviate this concern. 

We construct the intervention nodes by converting the map of Costa Rica into a separate raster object. Tiles in this raster represent potential intervention nodes and are larger in size compared to parcels in which we measure the outcome. We first find all the tiles within the protected areas and those that are no more than 5 km away from the boundary of these areas. The former set of nodes are defined as treated and the latter as untreated. Then, we remove all the treated nodes that are not adjacent to untreated ones to ensure that the two sets are more comparable. We end up with 233 treated nodes and 522 untreated ones. The geographic distribution of outcomes and intervention node placement is visualized in the left panel of Figure \ref{fig:F-2015-map}. We estimate the probability for each node to be treated by running a logistic regression model of the treatment indicator on three covariates aggregated to the level of intervention nodes: soil quality, distance to the nearest road, and distance to the nearest large city. The distribution of the the propensity score estimates is presented in the right panel of Figure \ref{fig:F-2015-map}.

In the analysis, we generate buffers around each intervention node and construct circle averages accordingly. The distance range is set to run from 0 km and 20 km. The cutoff value is set at 2km, 5km or 10km. We present results from the Hajek estimator and the smoothed Hajek estimator in Figure \ref{fig:F-2015-estimates}. Similar to what we observe from the replication of \cite{Jayachandran267}, deforestation activities diminished within the protected areas and nearby areas. Notably, spillover effects wane with increasing distance from the boundaries of the protected areas, ceasing to be significant past 5 km. 

\begin{figure}[H]
\centering
\includegraphics[width=.4\textwidth]{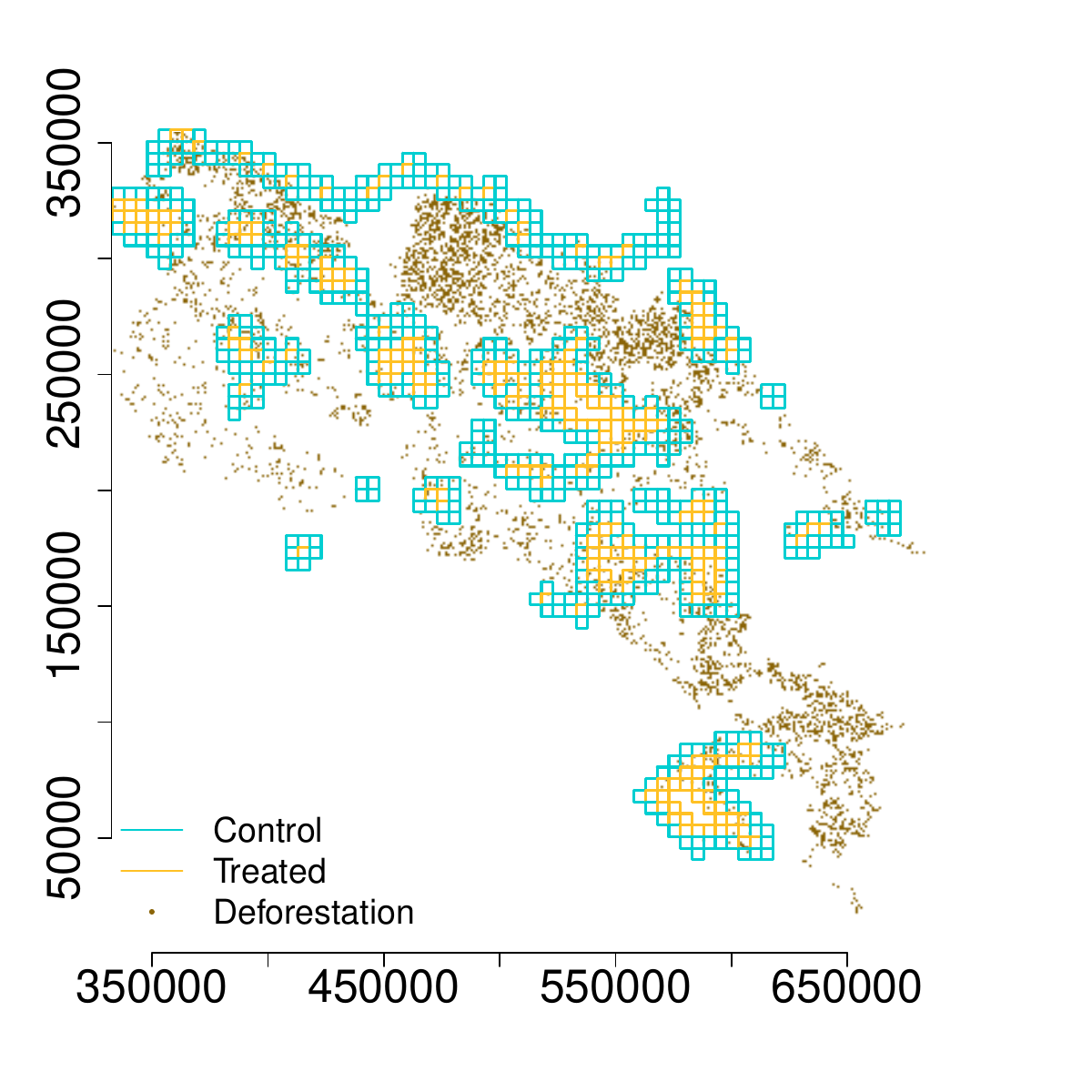}
\includegraphics[width=.4\textwidth]{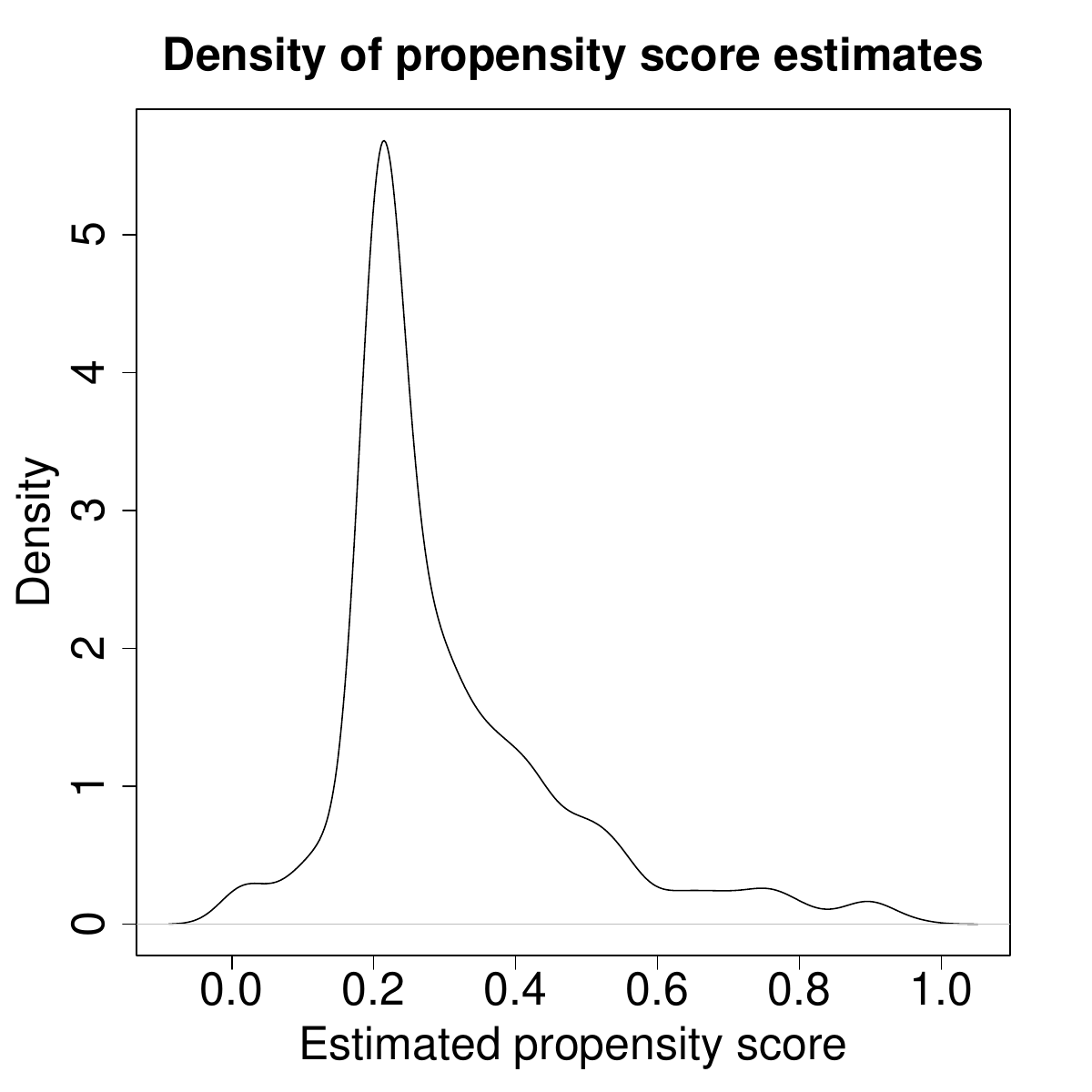}
\caption{The left plot shows the boundary of all the intervention nodes in our replication of \citet{ferraro2011conditions}. Tiles with a golden boundary are treated and those with a turquoise boundary are under control. Dark spots on the map represent the occurrence of deforestation in the outcome parcels. The right plot shows the distribution of the propensity score estimates based on logistic regression.}
\label{fig:F-2015-map}
\end{figure}

\begin{figure}[H]\centering
\includegraphics[width=0.4\textwidth]{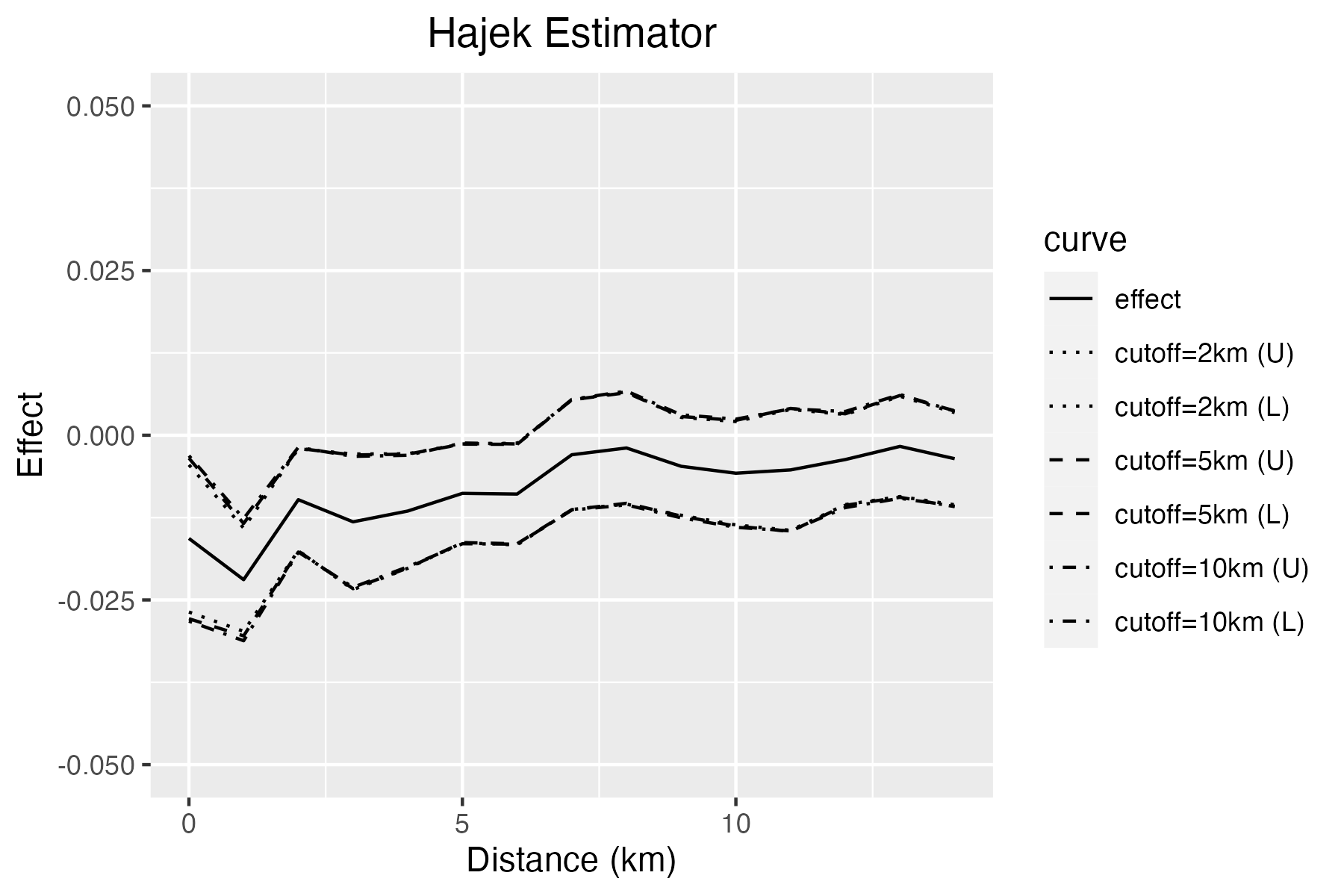}
\includegraphics[width=0.4\textwidth]{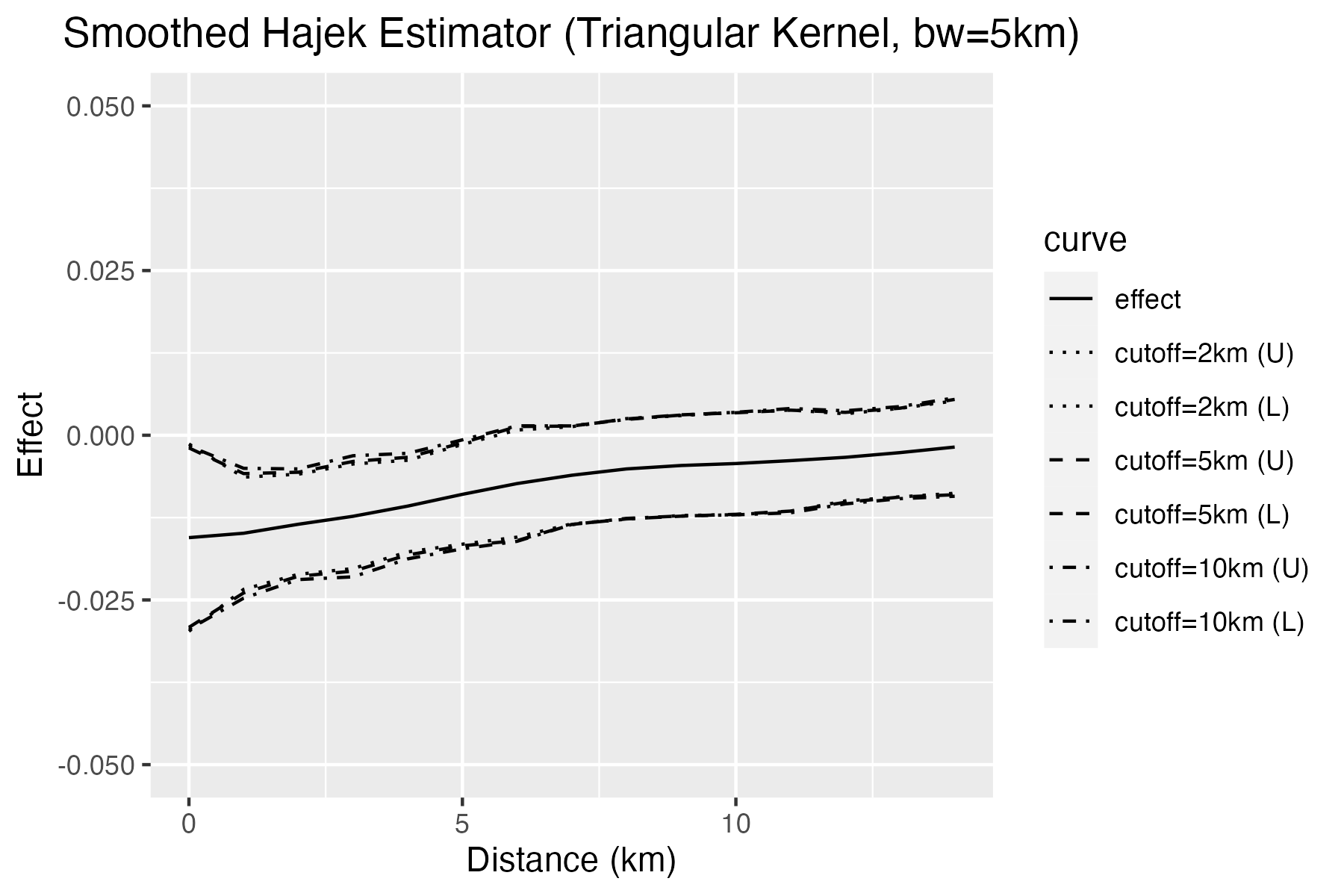}
\caption{Two plots present results based on \citep{ferraro2011conditions}. The top plot presents results from the Hajek estimator. The bottom plot presents results from the smoothed Hajek estimator with a triangular kernel and a bandwidth of 5km.  We plot pointwise two-sided 90\% intervals based on the positive semidefinite HAC variance estimators and empirical degree of adjustments with three different cutoffs (2km, 5km, and 10km). (U) indicates the upper end of the interval and (L) indicates the lower end of the interval.
}
\label{fig:F-2015-estimates}
\end{figure}

\section{Conclusion}
When treatments are applied at locations in space, the effects may bleed out and feed back in complex ways. As a result, outcomes at any point can depend on the distribution of treatments over the space, rather than on the treatment status of, e.g., the nearest intervention site.  
Such effects have important implications for policy.
For example, in approaches to forest conservation such ``payments for ecosystems services'' (PES) interventions, a primary concern is that gains in targeted areas are negated by ``leakage'' of negative spillover effects into non-targeted areas. 
To capture such effects, one needs to account for spatial interference. 
Standard approaches, which ignore such interference, yield conclusions about average policy impacts that may be unwarranted.

This paper explains how one can account for such interference in a randomized spatial experiment in which available information or knowledge is limited, and so we cannot confidently specify a parametric outcome model or non-parametric ``exposure mapping,'' nor can we be confident that interference is neatly contained within discrete geographical regions.
We show that even in this situation of ``unknown interference,'' we can still estimate a meaningful spatial effect---what we call the ``average marginalized effect'' (AME).  
The AME tells us what would happen, on average, if we switch an intervention node at a given distance into treatment, averaging over ambient effects emanating from other intervention nodes.  We can construct AME estimates for different distances, yielding a spatial effect curve.  The AME is identified under random assignment as a simple contrast.  Under restrictions on the spatial extent of interference, we can estimate the AME consistently and perform accurate inference using simple difference-in-means estimators and readily-available spatial standard error estimators.

We also develop extensions.  This includes specifying conditions under which the AME can be interpreted as a structural quantity that does not depend on the experimental design.  
We offer an approach for smoothing over distance, and explain how to test hypotheses on joint effects using Fisher-style randomization test under the sharp null.  

We illustrate our approach using simulation and applications to two real-world studies on forest conservation.  The examples show the soundness of our proposed methods but also point to areas for further research.  These include introducing methods to increase precision, for example, through covariate adjustment, variance estimation under less restrictive conditions, and inference for joint hypotheses. \\
\vspace{2em}
\noindent {\it Disclosure: The authors report there are no competing interests to declare.}

\bibliography{all}

\clearpage
\appendix
{\LARGE Appendix to {\it Design-Based Inference for Spatial Experiments under Unknown Interference}}

\DoToC

\section{Analytical Results}

This section contains technical results discussed in the paper. It includes the proof of Propositions \ref{prop:identification}, \ref{prop:ht-asymptotics},  \ref{prop:hajek-asymptotics}, and \ref{prop:inference}. The proof of Proposition \ref{prop:identification} is in Section \ref{section:prop1}, proofs of \ref{prop:ht-asymptotics} and \ref{prop:hajek-asymptotics} are in Section \ref{Section:Normality}, and results on HAC variance estimator are included in Section \ref{Section:HAC}.   

We first prove the unbiasedness of the Horvitz-Thompson estimator for the AME. We then characterize the variance of the Horvitz-Thompson estimator and the asymptotic variance of the Hajek estimator. Asymptotic normality follows from Lemma 1 and Lemma 2 in \cite{ogburn2017causal}. Finally, we show that the spatial HAC standard errors estimator estimates a quantity that is probably larger than the true asymptomatic variance under C\ref{assn:homo}, enabling conservative Wald-inference.

To reduce the complexity of notations, we make the following simplifications that will be used throughout the appendices:
	\begin{itemize}
		\item Without noted otherwise, the expectation is always taken over random assignments $\Z$.
		\item We write $\E[\mu_i\left( \Y\left(\Z\right);d\right)|Z_i=1]=\E[\muit]$ and $\E[\mu_i\left( \Y\left(\Z\right);d\right)|Z_i=0]=\E[\muic]$.
		\item  We write $\E[\mu_i\left( \Y\left(\Z\right);d\right)|Z_i=a,Z_j=b]=\E[\muiab]$ for $a,b\in\{0,1\}$.
		\item We write $\sum_{i=1}^N\sum_{j\in \mathcal{B}(i;d)}$ as $\sum_{i;j\in\mathcal{B}(i;d)}$.
\end{itemize}

\subsection{Results on Horvitz-Thompson Estimator}

The following lemma is useful.
\begin{lemma}\label{lemma:1}
	For any function $f: \{0,1\}^N\to \mathbb{R}$ and assuming  $C\ref{assn:bern-des}$, we have:
	\begin{description}
		\item[C\ref{assn:bern-des}.1] $\E\left[Z_i^k f(\mathbf{Z}) \right] = p\E\left[f(1, \mathbf{Z}_{-i}) \right]$ for any positive integer $k$.
		\item[C\ref{assn:bern-des}.2] $\E\left[Z_i^k Z_j^l f(\mathbf{Z}) \right] = p^2\E\left[f(1, 1, \mathbf{Z}_{-(i,j)}) \right]$ for any postive integers $k$ and $l$.
	\end{description}
\end{lemma}
\begin{proof}
	By Law of Iterated Expectations.
\end{proof}

\subsubsection{Proof for Proposition \ref{prop:identification}}\label{section:prop1}
\begin{proof}The Horvitz-Thompson estimator is defined in (\ref{eq:ht}) and the AME is defined in (\ref{def:AME}). We have:
\begin{align*}
&	\textnormal{AME}(d;\eta)   = \frac{1}{N}\sum_{i=1}^N \mu_i(1; d, \eta) - \frac{1}{N}\sum_{i=1}^N \mu_i(0; d, \eta) \\
& = \frac{1}{Np}\sum_{i=1}^N p \E[\mu_i(\Y; d) | Z_i=1 ]- \frac{1}{N(1-p)}\sum_{i=1}^N(1-p)  \E[\mu_i(\Y; d) | Z_i=0 ]\\
& = \frac{1}{Np}\sum_{i=1}^N \E[Z_i \mu_i(\Y; d) ] - \frac{1}{N(1-p)}\sum_{i=1}^N\E[(1-Z_i)\mu_i(\Y; d) ]\\
& = \E\left[ \frac{1}{Np}\sum_{i=1}^N Z_i \mu_i(\Y; d) - \frac{1}{N(1-p)}\sum_{i=1}^N (1-Z_i) \mu_i(\Y; d)\right]= \E[\widehat{\tau}_{\HT}(d) ],
\end{align*}
where the second equality uses the definition, and the third equality follows from Lemma \ref{lemma:1}.
\end{proof}

\subsubsection{Variance Characterization}
We now characterize the variance of the Horvitz-Thompson estimator. 
\begin{lemma}\label{lemma:ht-var}
	Under conditions C\ref{assn:bern-des}-C\ref{assn:local-inf}, the variance of estimator $\widehat{\tau}_{\HT}(d)$ is bounded as follows:
	\begin{align*}
		& \Var\left(\widehat{\tau}_{\HT}(d) \right) 
		\leq  \frac{1}{N^2p} \sum_{i=1}^N \E\left[\muit^2\right] + \frac{1}{N^2(1-p)} \sum_{i=1}^N \E\left[ \muic^2 \right] \\
		& + \frac{1}{N^2}\sumij\sum_{a=0}^1\sum_{b=0}^1 (-1)^{a+b}\Big\{ \E\left[\muiab\mujab  \right] - \E\left[\muia\right] \E\left[\mujb\right]\Big\},
	\end{align*}
	and, in addition under C\ref{assn:spacing}, we have that
	\begin{align*}
		\Var\left(\widehat{\tau}_{\HT}(d) \right) = O\left(\frac{1}{N}\right).
	\end{align*}
\end{lemma}

\begin{proof}
	Using the expression of the Horvitz-Thompson estimator, we have:
	\begin{align*}
		& \Var\left(\widehat{\tau}_{\HT}(d)  \right) =\frac{1}{N^2} \Var\left[\sum_{i=1}^N \left(\frac{Z_i}{p} - \frac{1-Z_i}{1-p}\right)\mu_i(\Y; d) \right] \\
		& =  \frac{1}{N^2} \sum_{i=1}^N \Var\left[ \left(\frac{Z_i}{p} - \frac{1-Z_i}{1-p}\right)\mu_i(\Y; d) \right] \\
		& + \frac{1}{N^2} \sum_{i=1}^N \sum_{j \neq i} \Cov\left[ \left(\frac{Z_i}{p} - \frac{1-Z_i}{1-p}\right)\mu_i(\Y; d), \left(\frac{Z_j}{p} - \frac{1-Z_j}{1-p}\right)\mu_j(\Y; d) \right] \\
		= & \frac{1}{N^2} \sum_{i=1}^N \E\left[\left( \left(\frac{Z_i}{p} - \frac{1-Z_i}{1-p}\right)\mu_i(\Y; d) \right)^2\right] - \frac{1}{N^2} \sum_{i=1}^N \left(\E\left[ \left(\frac{Z_i}{p} - \frac{1-Z_i}{1-p}\right)\mu_i(\Y; d) \right] \right)^2\\
		& + \frac{1}{N^2} \sum_{i=1}^N \sum_{j \neq i} \Cov\left[ \frac{Z_i}{p}\mu_i(\Y; d), \frac{Z_j}{p} \mu_j(\Y; d) \right] \\
		& - \frac{1}{N^2} \sum_{i=1}^N \sum_{j \neq i} \Cov\left[ \frac{Z_i}{p}\mu_i(\Y; d), \frac{1-Z_j}{1-p} \mu_j(\Y; d) \right] \\
		& - \frac{1}{N^2} \sum_{i=1}^N \sum_{j \neq i} \Cov\left[ \frac{1-Z_i}{1-p}\mu_i(\Y; d), \frac{Z_j}{p} \mu_j(\Y; d) \right] \\
		& + \frac{1}{N^2} \sum_{i=1}^N \sum_{j \neq i} \Cov\left[ \frac{1-Z_i}{1-p}\mu_i(\Y; d), \frac{1-Z_j}{1-p} \mu_j(\Y; d) \right].
	\end{align*}
	
	We further expand the first two terms in the above expression:
	\begin{align*}
		& \frac{1}{N^2} \sum_{i=1}^N \E\left[\left( \left(\frac{Z_i}{p} - \frac{1-Z_i}{1-p}\right)\mu_i(\Y; d) \right)^2\right] - \frac{1}{N^2} \sum_{i=1}^N \E\left[ \left(\frac{Z_i}{p} - \frac{1-Z_i}{1-p}\right)\mu_i(\Y; d) \right]^2\\
		= & \frac{1}{N^2} \sum_{i=1}^N \E\left[ \frac{Z_i^2}{p^2} \mu_i^2(\Y; d) \right] + \frac{1}{N^2} \sum_{i=1}^N \E\left[ \frac{(1-Z_i)^2}{(1-p)^2} \mu_i^2(\Y); d) \right] \\
		& - \frac{1}{N^2} \sum_{i=1}^N \E^2\left[ \left(\frac{Z_i}{p} - \frac{1-Z_i}{1-p}\right)\mu_i(\Y; d) \right] \\
		= & \frac{1}{N^2p} \sum_{i=1}^N \E\left[\muit^2\right] + \frac{1}{N^2(1-p)} \sum_{i=1}^N \E\left[ \muic^2\right] - \frac{1}{N^2} \sum_{i=1}^N \left(\E\left[\muit - \muic \right]\right)^2 \\
		\leq &  \frac{1}{N^2p} \sum_{i=1}^N \E\left[\muit^2 \right] + \frac{1}{N^2(1-p)} \sum_{i=1}^N \E\left[ \muic^2\right].
	\end{align*}
	We also have:
	\begin{align*}
		\frac{1}{Np} \sum_{i=1}^N \E\left[\muit^2\right] + \frac{1}{N^2(1-p)} \sum_{i=1}^N \E\left[ \muic^2\right] - \frac{1}{N^2} \sum_{i=1}^N \left(\E\left[\muit - \muic \right]\right)^2= O\left(1\right),
	\end{align*}
	since C\ref{assn:bounded-y} implies that all the moments are bounded. Next, we examine the first covariance term, which equals
	\begin{align*}
		& \frac{1}{N^2} \sum_{i=1}^N \sum_{j \neq i} \Cov\left[ \frac{Z_i}{p}\mu_i(\Y(\Z); d), \frac{Z_j}{p} \mu_j(\Y(\Z); d) \right] = \frac{1}{N^2}\sumij\Cov\left[ \frac{Z_i}{p}\mu_i(\Y(\Z); d), \frac{Z_j}{p} \mu_j(\Y(\Z); d) \right] \\
		= & \frac{1}{N^2} \sumij\E\left[\muitt\mujtt\right]- \frac{1}{N^2}\sumij\E\left[\muit \right]\E\left[\mujt \right].
	\end{align*}
	The first equality holds because of assumption C3 on local interference. Moreover,
	\begin{align*}
		\frac{1}{N^2} \sumij\E\left[\muitt\mujtt\right]- \frac{1}{N^2}\sumij\E\left[\muit \right]\E\left[\mujt \right]= O\left(1\right),
	\end{align*}
	since $|\mathcal{B}(i;d)|$ is bounded by (C4a). Other covariance terms have similar forms. We obtain the bound of the variance and its convergence rate by combining these terms together. 
\end{proof}
\subsubsection{Unbiasedness of the HT estimator in observaitional studies}\label{Section:HTobsproof}

\begin{proof}
	Note that in our setup the only randomness comes from the random assignment. Hence functions of the covariates are considered fixed.  In the scenario of observational studies and known propensity scores, we have
	\begin{align*}
		&	\textnormal{AME}(d;\eta)  = \frac{1}{N}\sum_{i=1}^N \mu_i(1; d, \eta) - \frac{1}{N}\sum_{i=1}^N \mu_i(0; d, \eta) \\
		& = \frac{1}{N}\sum_{i=1}^N \E\left[ \frac{p(\C) }{p(\C) }\mu_i(\Y(\mathbf{Z}); d)|Z_i=1\right] - \frac{1}{N}\sum_{i=1}^N \E\left[\frac{1-p(\C) }{1-p(\C) }\mu_i(\Y( \mathbf{Z}); d)|Z_i=0\right]\\
		& = \frac{1}{N}\sum_{i=1}^N\frac{p(\C) }{p(\C) } \E[Z_i\mu_i(\Y(\mathbf{Z}); d) | Z_i=1] -  \frac{1}{N}\sum_{i=1}^N\frac{1-p(\C) }{1-p(\C) } \E[(1-Z_i)\mu_i(\Y(\mathbf{Z}); d) | Z_i=0]   \\
		& = \E\left[ \frac{1}{N}\sum_{i=1}^N \frac{Z_i}{p(\C) } \mu_i(\Y; d) - \frac{1}{N}\sum_{i=1}^N \frac{1-Z_i}{p(\C)} \mu_i(\Y; d)\right]
	\end{align*}
	The fourth equality uses C\ref{assn:assignment_obsstudy}-(ii) and the law of total expectations.
\end{proof}


\subsection{Results on Hajek Estimator}


\subsubsection{Linearization and Asymptotic Variance Characterization}
We derive the limiting variance of the Hajek estimator using linearization technique.
\begin{lemma}\label{lemma:ha-var}
Consider the estimator $\widehat{\tau}_{\HA}(d)$ defined in (\ref{eq:hajek}). It has the following asymptotic linear expansion:
\begin{align*}\label{eq:ha-taylor}
	& \widehat{\tau}^{\Taylor}_{\HA}(d) =\textnormal{AME}(d;\eta) + \frac{1}{Np}\sum_{i=1}^N Z_i( \mu_i(\Y(\Z); d) - \bar{\mu}^1(d))  -  \frac{1}{Np}\sum_{i=1}^N (1-Z_i)(\mu_i(\Y(\Z);d)-\bar{\mu}^0(d)).
\end{align*}
Such an expansion satisfies $\sqrt{N}( \widehat{\tau}^{\Taylor}_{\HA}(d)-\widehat{\tau}_{\HA}(d))=o_p(1)$.
\end{lemma}

\begin{proof}
Denote $\widehat{\mu}^1(d)=\frac{1}{Np}\sum_{i=1}^N Z_i \mu_i(\Y(\Z); d)$, $\widehat{\mu}^0(d)=\frac{1}{N(1-p)}\sum_{i=1}^N (1-Z_i) \mu_i(\Y(\Z); d)$, $\widehat{N}_1=\frac{\sum_{i=1}^N Z_i}{Np}$, $\widehat{N}_0=\frac{\sum_{i=1}^N (1-Z_i)}{N(1-p)}$, and $\W = (\widehat{\mu}^1(d), \widehat{\mu}^0(d), \widehat{N}_1, \widehat{N}_0)$.  Further define $\bar{\mu}^1(d) = \frac{1}{N}\sum_{i=1}^N \E \left[\muit\right]$ and $\bar{\mu}^0(d) = \frac{1}{N}\sum_{i=1}^N \E \left[\muic\right]$.

We know that $\E \left[\widehat{\mu}^1(d) \right] = \bar{\mu}^1(d)$, $\E \left[\widehat{\mu}^0(d) \right] = \bar{\mu}^0(d)$, $\E[\widehat{N}_1] = \E[\widehat{N}_0] =1$. Thus, $\E[ \W] = (\bar{\mu}^1(d), \bar{\mu}^0(d)), 1, 1)$. Define $f(w)=f(a, b, c, d) = \frac{a}{c} - \frac{b}{d}$. Then the Hajek estimator can be written as $ f(\W)=f(\widehat{\mu}^1(d), \widehat{\mu}^0(d), \widehat{N}_1, \widehat{N}_0)$.

With probability approaching $1$, we have the following Taylor expansion of the Hajek estimator:\footnote{ $\nabla f(\E[\W])$ is the gradient of function $f$ evaluated at $\E[\W]$. $||\cdot||_2$ denotes the vector $l_2$ norm.}
\begin{align*}
\widehat{\tau}_{\HA}(d) = f(\W)= f(\E[ \W]) +\left(\nabla f(\E[\W])\right)^T(\W - \E[ \W]) + O_P(||\W - \E[ \W]||_2^2).
\end{align*}
Following the same argument as in Lemma \ref{lemma:ht-var}, we have that $N||\W - \E[ \W]||_2^2=O_p(1)$ and $O_P(\sqrt{N}||\W - \E[ \W]||_2^2)=o_P(1)$. It is easy to see that $ \nabla f(\E[\W])= (1, -1, -\bar{\mu}^1(d), \bar{\mu}^0(d))^{'}$. Some algebraic manipulations prove that the first two terms simplify to the expressions in $\widehat{\tau}_{\HA}^{\Taylor}(d)$. 
\end{proof}

\begin{lemma}\label{lemma:ha-var-homo}
The variance of the linearized Hajek estimator can be expressed as
\begin{align}
	& \Var(d) \left(\widehat{\tau}_{\HA}^{\Taylor}(d)\right)\\
	=&\frac{1}{N^2p}\sum_{i=1}^N \E\left[\left(\muit- \bar {\mu}^1(d)\right)^2\right] + \frac{1}{N^2(1-p)}\sum_{i=1}^N \E\left[\left(\muic- \bar {\mu}^0(d)\right)^2\right]  \label{LemmaA4:hajekvar1}\\
	-&  \frac{1}{N^2}\sum_{i=1}^N\E^2\left[ \muit - \muic - (\bar{\mu}^1(d) - \bar{\mu}^0(d) )\right]\label{LemmaA4:hajekvar3}\\
	+ & \frac{1}{N^2} \sumij\sum_{a=0}^1 \sum_{b=0}^1(-1)^{a+b} \E[\left(\muiab - \bar{\mu}^a(d)\right) \left(\mujab - \bar{\mu}^b(d)\right)] \\
	- & \frac{1}{N^2} \sumij\sum_{a=0}^1 \sum_{b=0}^1 (-1)^{a+b} \E[\muia - \bar{\mu}^a(d)] E[\mujb - \bar{\mu}^b(d)], \label{LemmaA4:hajekvar2}
\end{align}

Under conditions C\ref{assn:bern-des}-C\ref{assn:homo}, we have the following variance bound for $\Var \left(\widehat{\tau}_{\HA}^{\Taylor}(d)\right)$:
\begin{align}
\tilde{\textnormal{V}}_{\HA}= & \frac{1}{N^2p}\sum_{i=1}^N \E\left[\left(\muit- \bar {\mu}^1(d)\right)^2\right] + \frac{1}{N^2(1-p)}\sum_{i=1}^N \E\left[\left(\muic; d)- \bar {\mu}^0(d)\right)^2\right] \label{LemmaA4:hajekvarb1}\\
+ & \frac{1}{N^2} \sumij\sum_{a,b=0}^1 (-1)^{a+b} \E[\left(\muiab- \bar{\mu}^a(d)\right) \left(\mujab; d) - \bar{\mu}^b(d)\right)].\label{LemmaA4:hajekvarb2}
\end{align}
\end{lemma}
\begin{proof}
The characterization of the asymptotic variance is similar to that Lemma \ref{lemma:ht-var}. We omit the details. Note that for the terms in  lines (\ref{LemmaA4:hajekvar3}) and (\ref{LemmaA4:hajekvar2}), we have, up to a minus sign, that
\begin{align*}
& \frac{1}{N^2}\sum_{i=1}^N\E^2\left[ \muit -\muic- (\bar{\mu}^1(d) - \bar{\mu}^0(d) )\right] \\
& + \frac{1}{N^2}\sumij \sum_{a,b=0}^1 (-1)^{a+b}\E\left[\muia- \bar{\mu}^a(d)\right] \E\left[\mujb - \bar{\mu}^b(d) \right]  \\
= & \frac{1}{N^2}\sum_{i=1}^N  (\tau_i(d;\eta) -	\textnormal{AME}(d;\eta))^2 \\
 & + \frac{1}{N^2}\sumij \bigg\{ \E\left[
 \muit- \bar{\mu}^1(d)\right] \times \E\left[\mujt - \mujc - 	\textnormal{AME}(d;\eta))\right]\bigg\} \\
& - \frac{1}{N^2}\sumij \bigg\{ \E\left[\muic- \bar{\mu}^0(d)\right]  \times \E\left[\mujt-\mujc -	\textnormal{AME}(d;\eta)\right] \bigg\} \\
= & \frac{1}{N^2}\sum_{i=1}^N  (\tau_i(d;\eta) - 	\textnormal{AME}(d;\eta)) \sum_{j\in  \mathcal{B}(i;d)} (\tau_j(d;\eta) - 	\textnormal{AME}(d;\eta)) \geq 0. 
\end{align*}
by C\ref{assn:homo}. Hence the term combining expressions in (\ref{LemmaA4:hajekvar1}) and (\ref{LemmaA4:hajekvar2})  is non-positive in the limit under C\ref{assn:homo}, which proves the lemma.
\end{proof}

%
%


\subsection{Results on Asymptotic Distribution}

The consistency of the proposed estimator follows from the fact that both $\Var(\widehat{\tau}_{\HT}(d))$ and $\Var(\widehat{\tau}_{\HA}(d))$ converge to zero as $N \rightarrow \infty$ under conditions C\ref{assn:bern-des}-C\ref{assn:spacing}. The asymptotic normality of the Horvitz-Thompson estimator can be derived using classic central limit theorems for finitely dependent random variables based on the Stein's method \citep{chen2010normal, ross2011fundamentals, ogburn2017causal}. The Hajek estimator's asymptotic distribution can be then obtained via the linear expansion in Proposition \ref{lemma:ha-var-homo} and a similar application of the CLT result. For this purpose, we adpat the results in  \citet{ogburn2017causal} using the terms defined in our paper.

\begin{lemma}(Ogburn et al. (2020), Lemma 1 and 2)\label{lemma:ogburn-2020-1}
Consider a set of $N$ units. Let $U_1, \dots, U_N$ be bounded mean-zero random variables with finite fourth
moments and dependency neighborhoods $\mathcal{B}(i; d)$. If $c_i(d) \leq \tilde{c}$ for all $i$ and $\tilde{c}^2/N \rightarrow 0$, then
$$
\frac{\sum_{i=1}^N U_i}{ \sqrt{\Var(\sum_{i=1}^N U_i)}} \rightarrow N(0, 1).
$$
\end{lemma}

\subsubsection{Proofs for Propositions \ref{prop:ht-asymptotics} and \ref{prop:hajek-asymptotics} }\label{Section:Normality}
\begin{proof}

			We first prove the case for the HT estimator.  Define $U_i$ as
			\begin{equation*}
U_i=\frac{1}{\sqrt{\Var(\widehat{\tau}_{\HT}(d)}}\left(\frac{Z_i \muit}{Np} - \frac{(1-Z_i) \mu_i(\muic}{N(1-p)} - \frac{\E\left[ \muit\right] - \E\left[\muic\right]}{N}\right).
			\end{equation*} 

			We have $\frac{\widehat{\tau}_{\HT}(d)-\textnormal{AME}(d;\eta)}{\sqrt{\Var(\widehat{\tau}_{\HT}(d))}}=\sum_{i=1}^N U_i$. Obviously $\E[U_i]=0$ and $\Var( \sum_{i=1}^N U_i)=1$. Under the premises of Proposition \ref{prop:ht-asymptotics} and by C\ref{assn:bern-des} and C\ref{assn:bounded-y}, we have that the fourth moment of $U_i$ is bounded for all $i$.
			 By condition C 4a, $c_{i}(d) \leq c_{N}(d)$ in our case and $c^2_{N}(d)/N \rightarrow 0$. From Lemma \ref{lemma:ogburn-2020-1}, we know that $\frac{\widehat{\tau}_{\HT}(d)-\textnormal{AME}(d;\eta)}{\sqrt{\Var(\widehat{\tau}_{\HT}(d))}}\rightarrow N(0, 1)$. Similarly by Lemma \ref{lemma:ha-var}, we have $	\sqrt{N}(\widehat{\tau}_{\HA}(d) - \textnormal{AME}(d;\eta)) =\sqrt{N}( \widehat{\tau}^{\Taylor}_{\HA}(d)  -\textnormal{AME}(d;\eta)) +o_p(1)$. Under the premises in Proposition \ref{prop:hajek-asymptotics}, a similar argument  proves the normality for the Hajek estimator. 
\end{proof}

\subsection{Variance Estimation}\label{Section:HAC}

\subsubsection{HAC Variance Estimator}\label{subsection:HAC}

In Section \ref{est-inf}, we showed that the Hajek estimator can be interpreted as an OLS estimator that regresses circle averages on a constant term and the treatment indicator. We further suggested the use of a spatial HAC estimator for quantifying uncertanties. This section studies the behavior of the spatial HAC estimator. Importantly, we show that the HAC variance estimator is consistent for the variance bound defined in Lemma  \ref{lemma:ha-var} and lines (\ref{LemmaA4:hajekvarb1})-(\ref{LemmaA4:hajekvarb2}). This result suggests that the standard Wald type inference is valid for AMEs.

With a uniform kernel and $\tilde{d}=h(d)$, expanding the expression in (\ref{HAC_formula}) we have:
\begin{align*}
&	\widehat{\Sigma}_{\HAC} \begin{pmatrix}
		\hat{\mu}_0(d)  \\
		\widehat{\tau}_{\HA}(d)
	\end{pmatrix} =\begin{pmatrix}
		N, N_1 \\
		N_1, N_1
	\end{pmatrix}^{-1}\left(\sum_{i=1}^N \sum_{j=1}^N \mathbf{X}_{i} \mathbf{X}_{j}^{'}\hat{e}_i(d) \hat{e}_j(d) \mathbf{1}\{j \in \mathcal{B}(i;d) \}\right)\begin{pmatrix}
		N, N_1 \\
		N_1, N_1
	\end{pmatrix}^{-1} \\
	= & \frac{1}{N_1^2 N_0^2}\begin{pmatrix}
		N_1, -N_1 \\
		-N_1, N
	\end{pmatrix}\left(\sum_{i=1}^N \sum_{j=1}^N \begin{pmatrix}
		1, Z_j \\
		Z_i, Z_i Z_j
	\end{pmatrix}\hat{e}_i(d) \hat{e}_j(d) \mathbf{1}\{j \in \mathcal{B}(i;d) \}\right)\begin{pmatrix}
		N_1, -N_1 \\
		-N_1, N
	\end{pmatrix},
\end{align*}
Note that the $(2,2)$ entry of $\begin{pmatrix}
	N_1, -N_1 \\
	-N_1, N
\end{pmatrix} \begin{pmatrix}
	1, Z_j \\
	Z_i, Z_i Z_j
\end{pmatrix} \begin{pmatrix}
	N_1, -N_1 \\
	-N_1, N
\end{pmatrix}$ equals to $N_1^2 - N_1 N Z_i -  N_1 N Z_j + N^2 Z_i Z_j$.

Reindex the sample such that treated observations lie before observations under control and plug in the expression of $\hat{e}_i(d)$ , we can see that:
\begin{align}
	& \widehat{\V}_{\HAC}(\widehat{\tau}_{\HA}(d)) \\
	= & \frac{1}{N_1^2} \sum_{i=1}^{N_1} \hat{e}_i^2(d) + \frac{1}{N_0^2} \sum_{i=N_1+1}^{N} \hat{e}_i^2(d) + \frac{1}{N_1^2} \sum_{i=1}^{N_1} \sum_{j \in \mathcal{B}(i; d), j\not=i, Z_j = 1}\hat{e}_i(d)\hat{e}_j(d)\\
	 &- \frac{1}{N_1 N_0} \sum_{i=1}^{N_1} \sum_{j \in \mathcal{B}(i; d), j\not=i, Z_j = 0}\hat{e}_i(d)\hat{e}_j(d)  - \frac{1}{N_1 N_0} \sum_{i=N_1+1}^{N} \sum_{j \in \mathcal{B}(i; d), j\not=i, Z_j = 1}\hat{e}_i(d)\hat{e}_j(d) \\
	 & + \frac{1}{N_0^2} \sum_{i=N_1+1}^{N} \sum_{j \in \mathcal{B}(i; d), j\not=i, Z_j = 0}\hat{e}_i(d)\hat{e}_j(d) \\
	= &\frac{1}{N_1^2} \sum_{i=1}^{N_1} \left(\mu_i (d) - \widehat{\bar{\mu}}^1(d)\right)^2 + \frac{1}{N_0^2} \sum_{i=N_1+1}^{N} \left(\mu_i (d) - \widehat{\bar{\mu}}^0(d)\right)^2  \label{HAC:term1}\\
	& +  \frac{1}{N_1^2} \sum_{i=1}^{N_1} \sum_{j \in \mathcal{B}(i; d), j\not=i,Z_j = 1} \left(\mu_i (d) - \widehat{\bar{\mu}}^1(d)\right) \left(\mu_j (d) - \widehat{\bar{\mu}}^1(d)\right) \label{HAC:term2}\\
	& - \frac{1}{N_1 N_0} \sum_{i=1}^{N_1} \sum_{j \in \mathcal{B}(i; d), j\not=i,Z_j = 0} \left(\mu_i (d) - \widehat{\bar{\mu}}^1(d)\right) \left(\mu_j (d) - \widehat{\bar{\mu}}^0(d)\right)  \label{HAC:term3}\\
	& -\frac{1}{N_1 N_0} \sum_{i=N_1+1}^{N} \sum_{j \in \mathcal{B}(i; d), j\not=i,Z_j = 1} \left(\mu_i (d) - \widehat{\bar{\mu}}^0(d)\right) \left(\mu_j (d) - \widehat{\bar{\mu}}^1(d)\right)  \label{HAC:term4} \\
	& + \frac{1}{N_0^2} \sum_{i=N_1+1}^{N} \sum_{j \in \mathcal{B}(i; d), j\not=i,Z_j = 0} \left(\mu_i (d) - \widehat{\bar{\mu}}^0(d)\right) \left(\mu_j (d) - \widehat{\bar{\mu}}^0(d)\right) \label{HAC:term5},
\end{align}
where $\widehat{\bar{\mu}}^1(d)=\frac{\sum_{i=1}^N Z_i \mu_i(\Y(\Z); d)}{\sum_{i=1}Z_i}$ and $\widehat{\bar{\mu}}^0(d)=\frac{\sum_{i=1}^N \left(1-Z_i\right) \mu_i(\Y(\Z); d)}{\sum_{i=1}\left(1-Z_i\right)}$.

We now show that the variance estimate $\widehat{\V}_{\HAC}(\widehat{\tau}_{\HA}(d))$ is consistent for the rescaled variance bound defined in Lemma \ref{lemma:ha-var-homo}, $\tilde{\textnormal{V}}_{\HA}(d)$. Note that $\tilde{\textnormal{V}}_{\HA}(d)$ is provably larger than the asymptotic variance of the Hajek estimator. The result below thus suggests that the normal confidence interval with the HAC variance provides conservative coverage for the Hajek estimator asymptotically.

\begin{prop}\label{prop:HACvarconsistency}
If $N\times \tilde{\textnormal{V}}_{\HA}(d)$ is uniformly bounded below for large $N$, then we have
\begin{equation*}
\frac{\widehat{\V}_{\HAC}(\widehat{\tau}_{\HA}(d))  - \tilde{\textnormal{V}}_{\HA}(d) }{ \tilde{\textnormal{V}}_{\HA}(d) }\overset{p}{\to}0
\end{equation*}
\end{prop}
\begin{proof}
	We show below that $N\times (\widehat{\V}_{\HAC}(\widehat{\tau}_{\HA}(d))  - \tilde{\textnormal{V}}_{\HA}(d) ) \overset{p}{\to}0$. This, together with the premise in the lemma, leads to the claim that $\frac{\widehat{\V}_\HAC(\widehat{\tau}_{\HA}(d))  - \tilde{\textnormal{V}}_{\HA}  }{  \tilde{\textnormal{V}}_{\HA}  }\overset{p}{\to}0$.
	We first study terms in (\ref{HAC:term1}). For the treated group, we have
	\begin{align*}
		& N\left(\frac{1}{N_1^2} \sum_{i=1}^{N_1} \left(\mu_i(\Y(\Z); d) - \widehat{\bar{\mu}}^1(d)\right)^2- \frac{1}{p}\frac{1}{N^2}\sum_{i=1}^N E[(\muit-\bar{\mu}^1(d))^2]\right) \\
		=&N \left(\frac{1}{N_1^2} \sum_{i=1}^{N_1} \mu_i^2(\Y(\Z); d) - \frac{1}{N_1}\left(\widehat{\bar{\mu}}^1(d)\right)^2- \frac{1}{p} \frac{1}{N^2}\sum_{i=1}^N E[\muit^2] + \frac{1}{pN}\left(\bar{\mu}^1(d)\right)^2\right)\\
		=& \left(p\frac{N^2}{N_1^2} \frac{1}{Np}\sum_{i=1}^{N_1} \mu_i^2(\Y(\Z); d)- \frac{1}{p} \frac{1}{N}\sum_{i=1}^N E[\muit^2] \right)-\left(\frac{N}{N_1}\left(\widehat{\bar{\mu}}^1(d)\right)^2- \frac{1}{p}\left(\bar{\mu}^1(d)\right)^2\right)\\
	\overset{p}{\to}0 
	\end{align*}
	The convergence in probability is justified by noting $ \frac{1}{Np}\sum_{i=1}^{N_1} \mu_i^2(\Y(\Z); d)$, $\frac{1}{Np}\sum_{i=1}^{N_1} \mu_i(\Y(\Z); d)$ and $\frac{N_1}{N}$ are all Horvitz-Thompson estimators, and under C\ref{assn:bern-des}-C\ref{assn:spacing} they converge to their mean in probability. For the control group, we can similarly show
	\begin{align*}
	 N\left(\frac{1}{N_0^2} \sum_{i=1}^{N_0} \left(\mu_i(\Y(\Z); d) -\widehat{\bar{\mu}}^0(d)\right)^2- \frac{1}{1-p}\frac{1}{N^2}\sum_{i=1}^N E[(\muic-\bar{\mu}^0(d))^2]\right) \overset{p}{\to} 0
\end{align*}	
	\\
	Now we consider terms (\ref{HAC:term2})-(\ref{HAC:term5}). All terms are similar in stuctural so for simplicity we only include the calculation for (\ref{HAC:term2}). We first have the algebraic identity: 
	
		\begin{align}
			&\frac{N}{N_1^2} \sum_{i=1}^{N_1} \sum_{j \in \mathcal{B}(i; d), j\not=i,Z_j = 1} \left(\mu_i(\Y(\Z); d)  - \widehat{\bar{\mu}}^1(d)\right) \left(\mu_j(\Y(\Z); d)  -\widehat{\bar{\mu}}^1(d)\right)\\
			= & \frac{N}{N_1^2} \sum_{i=1}^{N} \sum_{j \in \mathcal{B}(i; d),j\not=i,}Z_iZ_j(\mu_i(\Y(\Z); d)-\bar{\mu}^1(d)) (\mu_j(\Y(\Z); d)-\bar{\mu}^1(d))  \\
			& +  \frac{N}{N_1^2}\sum_{i=1}^{N} \sum_{j \in \mathcal{B}(i; d),j\not=i}Z_iZ_j(\bar{\mu}^1(d)-\widehat{\bar{\mu}}^1(d)) (\mu_j(\Y(\Z); d)-\widehat{\bar{\mu}}^1(d))\\
			& +  \frac{N}{N_1^2}\sum_{i=1}^{N} \sum_{j \in \mathcal{B}(i; d),j\not=i}Z_iZ_j(\mu_i(\Y(\Z;d)-\bar{\mu}^1(d)) (\bar{\mu}^1(d)-\widehat{\bar{\mu}}^1(d))\\
			= & \frac{N^2}{N_1^2} \frac{1}{N}\sum_{i=1}^{N} \sum_{j \in \mathcal{B}(i; d),j\not=i}Z_iZ_j(\mu_i(\Y(\Z); d)-\bar{\mu}^1(d)) (\mu_j(\Y(\Z); d)-\bar{\mu}^1(d)) \label{HACcov:term1} \\
			& +  \frac{N^2}{N_1^2} \frac{1}{N}\sum_{i=1}^{N} \sum_{j \in \mathcal{B}(i; d),j\not=i}Z_iZ_j(\bar{\mu}^1(d)-\widehat{\bar{\mu}}^1(d)) (\mu_j(\Y(\Z); d)-\widehat{\bar{\mu}}^1(d))  \label{HACcov:term2}  \\
			& + \frac{N^2}{N_1^2}\frac{1}{N}  \sum_{j \in \mathcal{B}(i; d),j\not=i}Z_iZ_j(\mu_i(\Y(\Z;d)-\bar{\mu}^1(d)) (\bar{\mu}^1(d)-\widehat{\bar{\mu}}^1(d)) \label{HACcov:term3}  
		\end{align}
	Notice
		\begin{enumerate}
			\item Terms in (\ref{HACcov:term2}) and (\ref{HACcov:term3})  are of order $o_p(1)$. For (\ref{HACcov:term2}) , for example, by C\ref{assn:bern-des}-C\ref{assn:spacing},
			\begin{align*}
				& |\frac{N^2}{N_1^2}(\bar{\mu}^1(d)-\widehat{\bar{\mu}}^1(d)) \frac{1}{N} \sum_{i=1}^{N} \sum_{j \in \mathcal{B}(i; d)}Z_iZ_j(\mu_i(\Y(\Z)-\widehat{\bar{\mu}}^1(d))| \\
				& \leq \frac{N^2}{N_1^2}|\bar{\mu}^1(d)-\widehat{\bar{\mu}}^1(d))|\times \frac{1}{N} \sum_{i=1}^{N}c_i(d)|\mu_i(\Y(\Z)-\widehat{\bar{\mu}}^1(d)|\\
				& =o_p(1)\times O_p(1)=o_p(1)
			\end{align*}
			where $c_i(d)$ is defined in Section \ref{Section:inference}.\footnote{We define $\widehat{\bar{\mu}}^1(d)=0$ if $\sum_i Z_i=0$. } The argument is the same for the third term. 
			\item For  first term, we have, under C\ref{assn:bern-des}-C\ref{assn:spacing},
			\begin{equation*}
				 (\ref{HACcov:term1})- \frac{1}{N}\sumij E\left[  \left(\muitt-\bar{\mu}^1(d)) \right) \left(\mujtt-\bar{\mu}^1(d)) \right) \right]=o_p(1).
			\end{equation*}
          The proof is similar as in Proposition 6.2 in \cite{aronow_samii2017_interference}.
		\end{enumerate}
		Collecting all terms, we have proved that $N\times (\widehat{\V}_{\HAC}(\widehat{\tau}_{\HA}(d))  - \tilde{\textnormal{V}}_{\HA}(d)  ) \overset{p}{\to}0$.
\end{proof}


\subsubsection{SAH Variance Estimator}\label{Section:AFH}
In the previous section, we discussed the inference procedure for the Hajek estimator under C\ref{assn:homo}. We now provide an alternative approach, based on a proposal in \cite{savje2021average}. We have the following lemma:
\begin{prop}\label{prop:AFH}
Under C\ref{assn:bern-des}-C\ref{assn:d-bound}, we have,\footnote{The quantity $c_i(d)$ is defined in Section \ref{Section:inference}.}
\begin{align*}
& \Var(\hat{\tau}^{\Taylor}_{\HA}(d))  \leq \bar{\textnormal{V}}_{\HA}(d) \\
& =\frac{1}{N^2}\sum_{i=1}^N c_i(d) \frac{\E[\left(\muit-\bar{\mu}^1(d)\right)^2]}{p} +  \frac{1}{N^2}\sum_{i=1}^N c_i(d) \frac{\E[\left(\muit-\bar{\mu}^0(d)\right)^2]^2}{1-p}.
\end{align*}
Define an estimator $\widehat{{\V}}_{\SAH}(d)$:
	\begin{equation}\label{AFHbound}
	\widehat{\V}_{\SAH}(d)=\frac{1}{N^2}\sum_{i=1}^N Z_i c_i(d) \frac{(\mu_i(\Z,d)-\hat{\mu}^1(d))^2}{p^2} +     \frac{1}{N^2}\sum_{i=1}^N (1-Z_i) c_i(d) \frac{(\mu_i(\Z,d)-\hat{\mu}^0(d))^2}{(1-p)^2}.
\end{equation}
Provided that $N\times \bar{\textnormal{V}}_{\HA}(d)$ is uniformly bounded below for large $N$, 	$	\widehat{\V}_{\SAH}(d)$ is consistent for $\bar{\V}_{\HA}(d)$:
\begin{equation*}
	\frac{	\widehat{\V}_{\SAH}(d)-\bar{\textnormal{V}}_{\HA}(d)}{\bar{\textnormal{V}}_{\HA}(d)}\overset{p}{\to} 0.
\end{equation*}
\end{prop}

\begin{proof}
We first prove the upper bound. We have the following identity:
\begin{align*}
	& \Var\left(\frac{Z_i(\mu_i(\Y(\Z),d)-\bar{\mu}^1(d))}{p} -\frac{(1-Z_i)(\mu_i(\Y(\Z),d)-\bar{\mu}^0(d))}{1-p}\right)\\
	= &  \Var\left(\frac{Z_i(\mu_i(\Y(\Z),d)-\bar{\mu}^1(d))}{p})+\Var(\frac{(1-Z_i)(\mu_i(\Y(\Z),d)-\bar{\mu}^0(d))}{1-p})\right) \\
	& - 2\Cov{\left(\frac{Z_i(\mu_i(\Y(\Z),d)-\bar{\mu}^1(d))}{p},\frac{(1-Z_i)(\mu_i(\Y(\Z),d)-\bar{\mu}^0(d))}{1-p}\right)}\\
	= & \frac{\E[(\muit-\bar{\mu}^1(d))^2]}{p} - (\E[(\muit-\bar{\mu}^1(d))^2])^2 \\
	& +  \frac{\E[(\muic-\bar{\mu}^0(d))^2]}{1-p} - (\E[(\muic-\bar{\mu}^0(d))^2])^2 \\
	& + 2\E[\muit-\bar{\mu}^1(d)]\times \E[\muic-\bar{\mu}^0(d)]\\
	= &  \frac{\E[(\muit-\mu^1(d))^2]}{p} +  \frac{\E[(\muic-\bar{\mu}^0(d))^2]}{1-p}\\
	& - \left(\E[\muit-\bar{\mu}^1(d)]- \E[\muic-\bar{\mu}^0(d)] \right)^2\\
		\leq & \frac{\E[(\muit-\bar{\mu}^1(d))^2]}{p} +  \frac{\E[(\muic-\bar{\mu}^0(d))^2]}{1-p}   
\end{align*}

Let's define $A_i=\frac{Z_i(\mu_i(\Y(\Z),d)-\bar{\mu}^1(d))}{p}$ and $B_i = \frac{(1-Z_i)(\mu_i(\Y(\Z),d)-\bar{\mu}^0(d))}{1-p}$, then
\begin{align*}
&	\Var\left[\hat{\tau}^{\Taylor}_{\HA}(d)\right] =   \frac{1}{N^2}\sum_{i=1}^N \Var\left[A_i - B_i\right] + \frac{1}{N^2} \sum_{i=1}^N\sum_{j\neq i} \Cov\left[A_i - B_i, A_j - B_j\right] \\
	= &  \frac{1}{N^2}\sum_{i=1}^N \Var\left[A_i - B_i\right] + \frac{1}{N^2} \sum_{i=1}^N\sum_{j\in\mathcal{B}(i;d),j\not=i}\Cov\left[A_i - B_i, A_j - B_j\right]  \\
  \leq &  \frac{1}{N^2}\sum_{i=1}^N \Var\left[A_i - B_i\right] + \frac{1}{N^2} \sum_{i=1}^N\sum_{j\in\mathcal{B}(i;d),j\not=i} \frac{\Var\left[A_i - B_i\right] + \Var\left[A_j - B_j\right]}{2}\\
	= &\frac{1}{N^2}\sum_{i=1}^N \Var\left[A_i - B_i\right] + \frac{1}{N^2} \sum_{i=1}^N \left(c_i(d)-1\right) \Var\left[A_i - B_i\right]\\
	\leq & \frac{1}{N^2}\sum_{i=1}^N c_i(d)\left( \frac{\E[(\muit-\bar{\mu}^1(d))^2]}{p} +  \frac{\E[(\muic-\bar{\mu}^0(d))^2]}{1-p}\right).
\end{align*}
Note that we used the fact that, by definition, $j\in\mathcal{B}(i;d)$ if and only if $i\in\mathcal{B}(j;d)$.
This proves the upper bound for the variance. The consistency of the variance estimator is analogous to that in Proposition \ref{prop:HACvarconsistency} and we omit for brevity.

\end{proof}

\subsubsection{Proof of Proposition \ref{prop:inference} }

We use $\AVar(\widehat{\tau}_{\HA}(d))$ to denote $\Var(\tau^{\Taylor}_{\HA}(d))$.
By Proposition \ref{prop:hajek-asymptotics} and  Proposition \ref{prop:HACvarconsistency}, we have:

\begin{equation*}
	\frac{\widehat{\tau}_{\HA}(d)-\textnormal{AME}(d;\eta)}{\sqrt{\VhatHACd}} = 	\frac{\widehat{\tau}_{\HA}(d)-\textnormal{AME}(d;\eta)}{\sqrt{\AVar(\widehat{\tau}_{\HA}(d))}}\times \sqrt{\frac{\AVar(\widehat{\tau}_{\HA}(d))}{\tilde{\textnormal{V}}_{\HA}(d)}}  \times \sqrt{    \frac{\tilde{\textnormal{V}}_{\HA}(d)}{\VhatHACd} }.
\end{equation*}

Thus we have, for each $\alpha<\frac{1}{2}$,
\begin{align*}
	& \textbf{Prob}\left(  	z_{\frac{\alpha}{2}}\leq \frac{\widehat{\tau}_{\HA}(d)-\textnormal{AME}(d;\eta)}{\sqrt{\VhatHACd}}  	\leq z_{1-\frac{\alpha}{2}} \right)  \\
	=& 	\textbf{Prob}\left(   	z_{\frac{\alpha}{2}}\leq \frac{\widehat{\tau}_{\HA}(d)-\textnormal{AME}(d;\eta)}{\sqrt{\AVar(\widehat{\tau}_{HA}(d))}}\times \sqrt{\frac{\AVar(\widehat{\tau}_{\HA}(d))}{\tilde{\textnormal{V}}_{\HA}(d)}}  \times \sqrt{\frac{\tilde{\textnormal{V}}_{\HA}(d)}{\VhatHACd}}	\leq z_{1-\frac{\alpha}{2}} \right)\\
	= &	\textbf{Prob}\left(   	\sqrt{\frac{\tilde{\textnormal{V}}_{\HA}(d)}{\AVar(\widehat{\tau}_{\HA}(d))}}z_{\frac{\alpha}{2}}\leq \frac{\widehat{\tau}_{\HA}(d)-\textnormal{AME}(d;\eta)}{\sqrt{\AVar(\widehat{\tau}_{\HA}(d))}}\times  \sqrt{\frac{\tilde{\textnormal{V}}_{\HA}(d)}{\VhatHACd}}	\leq z_{1-\frac{\alpha}{2}}  \sqrt{\frac{\tilde{\textnormal{V}}_{\HA}(d)}  {\AVar(\widehat{\tau}_{\HA}(d))}}\right)\\
	\geq &  \textbf{Prob}\left(   	z_{\frac{\alpha}{2}}\leq \frac{\widehat{\tau}_{\HA}(d)-\textnormal{AME}(d;\eta)}{\sqrt{\AVar(\widehat{\tau}_{\HA}(d))}}\times  \sqrt{\frac{\tilde{\textnormal{V}}_{\HA}(d)}{\VhatHACd}}	\leq z_{1-\frac{\alpha}{2}} \right),
\end{align*}
where the last line follows because $  \sqrt{\frac{\tilde{\textnormal{V}}_{\HA}(d)}  {\AVar(\widehat{\tau}_{HA}(d))}}\geq 1$. Hence we have:
\begin{align*}
	& \lim_{N\to\infty} \textbf{Prob}\left(  	z_{\frac{\alpha}{2}}\leq \frac{\widehat{\tau}_{\HA}(d)-\textnormal{AME}(d)}{\sqrt{\VhatHACd}}  	\leq z_{1-\frac{\alpha}{2}} \right)\\
	&  \geq  \lim_{N\to\infty}  \textbf{Prob}\left(   	z_{\frac{\alpha}{2}}\leq \frac{\widehat{\tau}_{\HA}(d)-\textnormal{AME}(d)}{\sqrt{\AVar(\widehat{\tau}_{HA}(d))}}\times  \sqrt{\frac{ \tilde{\textnormal{V}}_{\HA}(d)}{\VhatHACd}}	\leq z_{1-\frac{\alpha}{2}} \right)=1-\alpha,
\end{align*}
because $\frac{\widehat{\tau}_{\HA}(d)-\textnormal{AME}(d;\eta)}{\sqrt{\AVar(\widehat{\tau}_{\HA}(d))}}\times  \sqrt{\frac{\tilde{\textnormal{V}}_{\HA}(d)}{\VhatHACd}}\overset{d}{\to} N(0,1)$. A similar calculation follows for the case using the variance estimator $\widehat{\V}_{\SAH}(d)$.




\subsection{Efficiency Comparison between the Hajek and HT estimators}\label{Section:HAHTComparison}
This section compares the estimation efficiency, in terms of asymptotic variances, between the Hajek and HT estimators. Although it is not true that the Hajek estimator is \textit{not always} more efficient than the HT estimator in our setting, the Hajek estimator has some efficiency properties that make it appealing in practice:
	\begin{enumerate}
		\item When there is no interactive effect,\footnote{We say that there is no interactive effect, when for all $i,j\in \node_N$,	\begin{equation}
				\left(\E[\muitt]-\E[\muict]\right)-\left( \E[\muit]-\E[\muic]\right)=0
			\end{equation}	
			Note, by C\ref{assn:bern-des}, this also implies that
			\begin{equation}
				\left(\E[\muitc]-\E[\muicc]\right)-\left( \E[\muit]-\E[\muic]\right)=0
			\end{equation} 
			If we assume that the potential outcome has the form $\mu_i(\Z)=\beta_{i}\Z_i + \sum_{j\in \mathcal{B}(i:d)}\beta_{i,j}\Z_j+ \sum_{k,l\in  \mathcal{B}(i:d) } \beta_{i,kl}\Z_k\Z_l$, setting all coefficients $\beta_{i,kl}=0$ will rule out the interactive effect.} the Hajek estimator is optimal among the estimators that use treatment-arm-specific intercepts.	
		\item When the interactive effect size is small, the Hajek estimator can be expected to be more efficient than the HT estimator.
		
		\item When an interactive effect quantity ($\Delta(1)$ below), average treated and control outcomes have the same sign, the Hajek estimator is more efficient than the HT estimator.
	\end{enumerate}
	For other cases, there are potential outcomes that make Hajek estimator more efficient than the HT estimator and vice versa.

		Consider the following class of estimators, which adjust the treated and control outcomes with a treatment-arm-specific intercepts.
		\begin{equation}
			\widehat{\tau}(\mu_1,\mu_0;d)= \mu_1-\mu_0 +\left( \frac{1}{Np}\sum_{i=1}^n \Z_i(\mu_i(\Y;d)-\mu_1)- \frac{1}{N(1-p)}\sum_{i=1}^n (1-\Z_i)(\mu_i(\Y;d)-\mu_0)\right)
		\end{equation}
		Both the HT estimator and the linearized version of the Hajek estimator have this form: $\widehat{\mu}(d)=	\widehat{\tau}(0,0)$ and $\widehat{\mu}_{\HA}(d)=	\widehat{\tau}(\mutd,\mucd)$. It is in this class of the estimators we discuss the efficiency property of the Hajek estimator. We define the interactive effect quantity mentioned above
\begin{equation*}
	\Delta(1)=  \frac{1}{N}\sumij\left(\E[\muitt]-\E[\muict]-\left(\E[\muit]-\E[\muic]\right)\right).\footnote{One can also define a similar interactive effect quantity
		\begin{equation*}
			\Delta(0)=  \frac{1}{N}\sumij\left(\E[\muicc]-\E[\muitc]-\left(\E[\muic]-\E[\muit]\right)\right).
		\end{equation*}
		These two quantities are related by the identity $p	\Delta(1)=(1-p)\Delta(0)$. This follows by a calculation
		\begin{align*}
			& N\times\left(p\Delta(1)-(1-p)\Delta(0)\right)=  \sumij \left(p\E[\muitt]+(1-p)\E[\muitc]\right) \\
			& - \sumij  \left(p\E[\muict]+(1-p)\E[\muicc]\right)- \left(\sumij\E[\muit] - \sumij\E[\muic]\right)\\
			& =0 .
	\end{align*}}
\end{equation*}
		\begin{prop}
			Under C\ref{assn:bern-des}-C\ref{assn:spacing}, we have, for any scalars $\mu_1$ and $\mu_0$,
				\begin{align}
				&\V\left(	\widehat{\tau}(\mu_1,\mu_0;d) \right)=
					\frac{1}{N^2p}\sum_{i=1}^N \E\left[ \left(\muit- \mu_1\right)^2\right] + \frac{1}{N^2(1-p)}\sum_{i=1}^N \E\left[\left(\muic- \mu_0\right)^2\right]  \\
					-&  \frac{1}{N^2}\sum_{i=1}^N\left(\E\left[\muit - \muic - (\mu_1 - \mu_0 )\right]\right)^2\\
					+ & \frac{1}{N^2} \sumij \sum_{a=0}^1 \sum_{b=0}^1  (-1)^{a+b}\E[\left(\muiab - \mu_a\right) \left(\mujab - \mu_b\right)] \\
					- & \frac{1}{N^2}\sumij \sum_{a=0}^1 \sum_{b=0}^1  (-1)^{a+b} \E[\muia -\mu_a] E[\mujb - \mu_b].
				\end{align}
In particular, we have the following algebraic identity:
\begin{equation}\label{formula:variance}
			\V\left(	\widehat{\tau}(\mu_1,\mu_0;d) \right)=  \frac{(1-p)p}{N} \left( \frac{\mu_1}{p}+\frac{\mu_0}{1-p}-\left(  \frac{\bar{\mu}_1(d)}{p} + \frac{\bar{\mu}_0(d)}{1-p} +\frac{p\Delta(1)}{p(1-p)} \right)\right)^2+C,
\end{equation}
				where $C$ is a constant independent of $\mu_1$ and $\mu_0$.
		\end{prop}
\begin{proof}
The characterization of the variance is similar to that in Proposition \ref{lemma:ha-var-homo}. Expanding the variance expression we have the following identity:
\begin{align*}
	&\V\left(	\widehat{\tau}(\mu_1,\mu_0;d) \right)=\left(	\frac{1}{Np}-\frac{1}{N}\right)\mu^2_1 + 
	\left(\frac{1}{N(1-p)}-\frac{1}{N}\right)\mu_0^2+\frac{2}{N}\mu_1\mu_0\\
	& +\left( \left(-\frac{2}{Np} +\frac{2}{N}\right)\bar{\mu}_1(d)-\frac{2}{N}\bar{\mu}_0(d) - \frac{2}{N}\Delta(1)\right)\mu_1\\
	& +\left( \left(-\frac{2}{N(1-p)} +\frac{2}{N}\right)\bar{\mu}_0(d)-\frac{2}{N}\bar{\mu}_1(d)  - \frac{2}{N}\Delta(0)\right)\mu_0\\
	& + \V(\widehat{\tau}(0,0;d))\\
	& = \frac{1}{N}\left( \frac{1-p}{p}\mu_1^2 + \frac{p}{1-p}\mu_0^2 + 2\mu_1\mu_0 \right)- \frac{2}{N}\left(   \frac{1-p}{p}\bar{\mu}_1(d) +\bar{\mu}_0(d) +\Delta(1) \right)\mu_1\\
	& - \frac{2}{N}\left(   \frac{p}{1-p}\bar{\mu}_0(d) + \bar{\mu}_1(d) +\Delta(0) \right)\mu_0+ \V(\widehat{\tau}(0,0;d))\\		 
	& = \frac{(1-p)p}{N}\left( \frac{\mu_1^2}{p^2} +\frac{\mu_0^2}{(1-p)^2} + 2\frac{\mu_1\mu_0}{p(1-p)} \right)-  \frac{2(1-p)p}{N}\left(  \frac{\bar{\mu}_1(d)}{p} + \frac{\bar{\mu}_0(d)}{1-p} +\frac{p\Delta(1)}{p(1-p)} \right)\frac{\mu_1}{p}\\
	& - \frac{2(1-p)p}{N}\left(   \frac{\bar{\mu}_0(d)}{1-p} + \frac{\bar{\mu}_1(d)}{p} +\frac{(1-p)\Delta(0)}{p(1-p)} \right)\frac{\mu_0}{1-p}+ \V(\widehat{\tau}(0,0))\\		 		
	& = \frac{(1-p)p}{N} \left(\left(  \frac{\mu_1}{p}+\frac{\mu_0}{1-p}\right)^2- 2\left(  \frac{\bar{\mu}_1(d)}{p} + \frac{\bar{\mu}_0(d)}{1-p} +\frac{p\Delta(1)}{p(1-p)} \right)\left(  \frac{\mu_1}{p}+\frac{\mu_0}{1-p}\right)\right) + \V(\widehat{\tau}(0,0;d))\\
	& = \frac{(1-p)p}{N} \left( \frac{\mu_1}{p}+\frac{\mu_0}{1-p}-\left(  \frac{\bar{\mu}_1(d)}{p} + \frac{\bar{\mu}_0(d)}{1-p} +\frac{p\Delta(1)}{p(1-p)} \right)\right)^2+C
\end{align*}
where $C$ is independent of $\mu_1$ and $\mu_0$. In the fourth equality, we used the fact that $p\Delta(1)=(1-p)\Delta(0)$.
\end{proof}

Expression (\ref{formula:variance}) immediately leads to the conclusions at the beginning of this section. When the interactive effect is small and hence $\Delta(1)$ is very close to zero, we have
\begin{equation}
\V(\widehat{\tau}_{HA}^{\Taylor}(d))=	\V\left(	\widehat{\tau}(\bar{\mu}_1(d),\bar{\mu}_1(d)) \right)\approx 0+C,
\end{equation}
which is close to minimizing the variance within the class of estimators. We can further compare the variance of the HT and HA estimator:
\begin{align}
	&\V(\widehat{\tau}_{\HT}(d))- \V(\widehat{\tau}_{HA}^{\Taylor}(d)) \propto\left(  \frac{\bar{\mu}_1(d)}{p} + \frac{\bar{\mu}_0(d)}{1-p} +\frac{p\Delta(1)}{p(1-p)} \right)^2 -\left(\frac{p\Delta(1)}{p(1-p)}\right)^2\\
	& = \left(  \frac{\bar{\mu}_1(d)}{p} + \frac{\bar{\mu}_0(d)}{1-p} \right) \left(  \frac{\bar{\mu}_1(d)}{p} + \frac{\bar{\mu}_0(d)}{1-p} +\frac{2p\Delta(1)}{p(1-p)} \right).
\end{align}
Thus when the average treated and control outcomes, and the interactive effect is of the same sign, the quantity is nonnegative and the Hajek estimator is weakly more efficient than the HT estimator.

\subsection{Theoretical Results on Observational Studies}\label{appendix:observationalstudy}
This section contains results for inference on the AME in observational setting. In particular, we consider the case where assignment probabilities are modeled by a logistic model and derive asymptotic linearization and variance estimator. 

Given the setup in Section \ref{Section:ObservationalStudies}. We make the further parametric assumptions on the assignment model. 
\begin{assn}\label{assn:MLE}
The following properties hold for all sample size $N$ and for all $i\in\node_N$:
\begin{enumerate}[label=(\roman*)]
	\item (Fixed Dimensionality)The confounder $\C=(o_{1i},...,o_{ki})'$ is of dimension $k$.
	\item (Bounded Confounders) Confouders are uniformly bounded: there exists a constant $U$ such that $\max_{i\in\node_N}||\C||_{\infty}\leq U$.\footnote{$||\cdot||_{\infty}$ is defined as $||\C||_{\infty}=\max_{s=1,...,k}|o_{si}|$.} 
	\item (Logistic Model) The assignment probabilites follow a logistic model: for all $i\in\node_N$ there exists a $\theta_0\in\mathbb{R}^k$ such that
\begin{equation*}
	P(Z_i=1)=p(\C|\theta) = \frac{\exp(\C'\theta_0)}{1+ \exp(\C'\theta_0)}
\end{equation*}	
  \item Compact Parameter Space: $\theta_0\in\Theta_0$ where $\Theta_0\subset \mathbb{R}^k$ is a compact set.
	\item (Nonsingularity) The smallest eigenvalue of $\frac{1}{N}\sum_{i=1}^N \C \C'$ is uniformly bounded below.
	\item The MLE estimator is used to estimate the coefficient vector $\theta$:
	\begin{equation*}
		\widehat{\theta}_{\textnormal{MLE}} = \arg\max_{\theta\in\Theta_0} \sum_{i=1}^N \left(Z_i \C'\theta - \log(1+\exp(\C'\theta))\right).
	\end{equation*}
\end{enumerate}
\end{assn}
\begin{lemma}\label{lemma:logistic}
	Under C\ref{assn:assignment_obsstudy} and  C\ref{assn:MLE}, $\widehat{\theta}_{MLE}-\theta_0 = o_p(1)$ and 
	\begin{equation*}
		\widehat{\theta}_{\textnormal{MLE}}-\theta_0 = \left(  \frac{1}{N}\sum_{i=1}^N p(\C|\theta_0)(1-p(\C|\theta_0))\C\C'  \right)^{-1} \frac{1}{N} \sum_{i=1}^N \left( \left(Z_i-p(\C|\theta_0)\right)\C\right)+ o_p(N^{-\frac{1}{2}}).\footnote{We note that the randomness is reasoned over random assignments $Z_i$'s, holding potential outcomes and confoundners as fixed.}
			\end{equation*}
In particular, 	$\widehat{\theta}_{\textnormal{MLE}}-\theta_0=O_p(N^{-\frac{1}{2}})$.
\end{lemma}
\begin{proof}
	Identification is established by the standard KL-divergence argument and C\ref{assn:MLE}-(v). The rest of the proof is standard by Taylor expansions. See for example \cite{newey1994large} and \cite{chang2023design}.
\end{proof}

We defined the following IPW estimator:
\begin{equation*}
	\widehat{\tau}_{\IPW}(d) = \frac{1}{N}\sum_{i=1}^N \frac{Z_i}{p(\C|\widehat{\theta}_{\textnormal{MLE}}) } \mu_i(\Y; d) - \frac{1}{N}\sum_{i=1}^N \frac{1-Z_i}{1-p(\C|\widehat{\theta}_{\textnormal{MLE}})} \mu_i(\Y; d).
\end{equation*}	
For brevity, we denote $p\left(\C|\theta_0\right)=p_0(\C)$. We define the following coefficients:
\begin{equation*}
	\beta_{1,N} = \left(\frac{1}{N}\sum_{i=1}^N p_0(\C)\left(1-p_0(\C)\right)\C\C'\right)^{-1} \left(\frac{1}{N}\sum_{i=1}^N \E[\mu_i(1;d)p_0(\C)\left(1-p_0(\C)\right)\C] \right),
\end{equation*}	
and,
\begin{equation*}
	\beta_{0,N} = \left(\frac{1}{N}\sum_{i=1}^N p_0(\C)\left(1-p_0(\C)\right)\C\C'\right)^{-1} \left(\frac{1}{N}\sum_{i=1}^N \E[\mu_i(0;d)p_0(\C)\left(1-p_0(\C)\right)\C] \right).\footnote{Interested readers can check $\beta_{1_N}=\bar{\mu}_1(d)$ and $\beta_{0_N}=\bar{\mu}_0(d)$ when $\C$ only contains an intercept.}
\end{equation*}
A similar homophily condition is needed for the HAC variance estimation.
\begin{assn}\label{assn:obs_homophily}
For all sample size $N$, 
\begin{equation*}
\frac{1}{N}\sum_{i=1}^N \left(\tau_{i}(d;\eta)- \left(\C'\beta_{1,N}-\C'\beta_{0,N}\right)\right) \sum_{j\in  \mathcal{B}(i;d)} \left(\tau_{j}(d;\eta) -\left(\C'\beta_{1,N}-\C'\beta_{0,N}\right)\right)\geq 0.
\end{equation*}
\end{assn}

We can similarly define a HAC type estimator as in (\ref{assn:assignment_obsstudy}). Let $\X  = \begin{pmatrix}
		1  & 1 &  \hdots, 1 \\
		Z_1 & Z_2 & \hdots Z_N\\
	\end{pmatrix}'\in \mathbb{R}^{N\times 2}$.
	\begin{align*}
		\widehat{ \V}_{\HAC}^{\obs} (d)= (\X' \X)^{-1}\left(\sum_{i=1}^N \sum_{j=1}^N \X_i \X_j^{'}\hat{\epsilon}_i\hat{\epsilon}_j \mathbf{1}\{j \in \mathcal{B}(i;d) \}  \right)(\X' \X)^{-1},
	\end{align*}
	where $\widehat{\epsilon}_i =\frac{N_1}{N}\frac{Z_i}{p(\Cj|\widehat{\theta}_{\textnormal{MLE}})}\left(\mu_i(\Y;d)-  \C'\widehat{\beta}_{1,N} \right)+ \frac{N_0}{N}\frac{1-Z_i}{1-p(\Cj|\widehat{\theta}_{\textnormal{MLE}})}\left(\mu_i(\Y;d)-  \C'\widehat{\beta}_{0,N} \right) $.\footnote{$\widehat{\beta}_{1,N}= \left(\frac{1}{N}\sum_{i=1}^N p(\C|\widehat{\theta}_{\textnormal{MLE}})\left(1-p(\C|\widehat{\theta}_{\textnormal{MLE}})\right)\C\C'\right)^{-1}  \left(\frac{1}{N}\sum_{i=1}^N \frac{Z_i\mu_i(\Y;d)}{p(\C|\widehat{\theta}_{\textnormal{MLE}})} p(\C|\widehat{\theta}_{\textnormal{MLE}})(1-p(\C|\widehat{\theta}_{\textnormal{MLE}}))\C\right)$  and 
		$\widehat{\beta}_{0,N}= \left(\frac{1}{N}\sum_{i=1}^N p(\C|\widehat{\theta}_{\textnormal{MLE}})\left(1-p(\C|\widehat{\theta}_{\textnormal{MLE}})\right)\C\C'\right)^{-1}  \left(\frac{1}{N}\sum_{i=1}^N \frac{\left(1-Z_i\right)\mu_i(\Y;d)}{1-p(\C|\widehat{\theta}_{\textnormal{MLE}})} p(\C|\widehat{\theta}_{\textnormal{MLE}})(1-p(\C|\widehat{\theta}_{\textnormal{MLE}}))\C\right)$.  
	} The formula is similar to (\ref{assn:assignment_obsstudy}), except now that the residuals are weighted with the inverse propensity scores. Note that the covariates $\C$ do not appear in the matrix $\X$. Similarly, we can similarly define a \cite{savje-etal2018-unknown-interference} type variance estimator as in (\ref{varianceSAH}):
	\begin{equation*}
		\widehat{V}^{\obs}_{\SAH}(d)=\frac{1}{N^2}\sum_{i=1}^N \frac{Z_i c_i(d) \widehat{\epsilon}^2_i }{p^2(\C|\widehat{\theta}_{\textnormal{MLE}} )}+     \frac{1}{N^2}\sum_{i=1}^N \frac{ \left(1-Z_i\right) c_i(d) \widehat{\epsilon}^2_i }{\left(1-p(\C|\widehat{\theta}_{\textnormal{MLE}} )\right)^2},
	\end{equation*}

We have the following linearization results.
\begin{prop}
	Define:
\begin{align*}
& \widehat{\tau}^{\Taylor}_{\IPW}(d)=\frac{1}{N}\sum_{i=1}^N \C'\beta_{1,N}- \frac{1}{N}\sum_{i=1}^N \C'\beta_{0,N}\\
& + \frac{1}{N}\sum_{i=1}^n \frac{Z_i\left( \mu_i(\Y;d)-\C'\beta_{1,N} \right)}{p(\C|\theta_0)}-  \frac{1}{N}\sum_{i=1}^n \frac{\left(1-Z_i\right)\left( \mu_i(\Y;d)-\C'\beta_{0,N} \right)}{1-p(\C|\theta_0)}.
\end{align*}
	Under C\ref{assn:bounded-y}, C\ref{assn:local-inf}, C\ref{assn:spacing}, C\ref{assn:assignment_obsstudy}, and C\ref{assn:MLE}, 
\begin{equation}\label{eqn:obs_equivalence}
\widehat{\tau}^{\Taylor}_{\IPW}(d)-\widehat{\tau}_{\IPW}(d)=o_p(N^{-\frac{1}{2}}).
\end{equation}
Define
\begin{align}
	& \V\left(\widehat{\tau}^{\Taylor}_{\IPW}(d)\right) =
		\frac{1}{N^2}\sum_{i=1}^N\frac{\E\left[ \left(\muit- \C'\beta_{1,N}\right)^2\right]}{p_0(\C)} + \frac{1}{N^2}\sum_{i=1}^N \frac{\E\left[\left(\muic- \C'\beta_{0,N}\right)^2\right]}{1-p_0(\C)}  \\
		-&  \frac{1}{N^2}\sum_{i=1}^N\left(\E\left[\muit - \muic - (\C'\beta_{1,N} - \C'\beta_{0,N} )\right]\right)^2\\
		+ & \frac{1}{N^2} \sumij \sum_{a=0}^1 \sum_{b=0}^1 \E[\left(\muiab - \C'\beta_{a,N}\right) \left(\mujab -  \C'\beta_{b,N}\right)] \\
		- & \frac{1}{N^2}\sumij \sum_{a=0}^1 \sum_{b=0}^1  (-1)^{a+b} \E[\muia - \C'\beta_{a,N}] E[\mujb -  \C'\beta_{b,N}].
	\end{align}
In addition under C\ref{assn:obs_homophily},
\begin{align*}
	 & \V\left(\widehat{\tau}^{\Taylor}_{\IPW}(d)\right) \leq 	\frac{1}{N^2}\sum_{i=1}^N\frac{\E\left[ \left(\muit- \C'\beta_{1,N}\right)^2\right]}{p_0(\C)} + \frac{1}{N^2}\sum_{i=1}^N \frac{\E\left[\left(\muic- \C'\beta_{0,N}\right)^2\right]}{1-p_0(\C)}  \\
		+ & \frac{1}{N^2} \sumij \sum_{a=0}^1 \sum_{b=0}^1 \E[\left(\muiab - \C'\beta_{a,N}\right) \left(\mujab -  \C'\beta_{b,N}\right)] \\
\end{align*}
Provided that $N\times \V\left(\widehat{\tau}^{\Taylor}_{\IPW}(d)\right)$ is uniformly bounded below for all large $N$, $\widehat{\tau}_{\IPW}(d)$ follows the asymptotic distribution:
\begin{equation*}
	\frac{\widehat{\tau}_{\IPW}(d)-\textnormal{AME}(d;\eta)}{\sqrt{ \V\left(\widehat{\tau}^{\Taylor}_{\IPW}(d)\right)}}\overset{d}{\to} N(0,1).
\end{equation*}
Further more, for each $\alpha < \frac{1}{2}$,
	\begin{enumerate}[label=(\roman*)]
		\item $\lim_{N\to\infty}\mathbf{Prob}(z_{\frac{\alpha}{2}}\leq \left|\frac{\widehat{\tau}_{\HA}(d)-\textnormal{AME}(d;\eta) }{\sqrt{\widehat{\V}^{\obs}_{\SAH}(d)}}\right| \leq z_{1-\frac{\alpha}{2}})\geq 1-\alpha$;
		\item additionally under C\ref{assn:obs_homophily}, $\lim_{N\to\infty}\mathbf{Prob}(z_{\frac{\alpha}{2}}\leq \left|\frac{\widehat{\tau}_{\HA}(d)-\textnormal{AME}(d;\eta) }{\sqrt{\widehat{\V}^{\obs}_{\HAC}(d)}}\right| \leq z_{1-\frac{\alpha}{2}})\geq 1-\alpha$.
	\end{enumerate}
\end{prop}

\begin{proof}
We only show the calculation for the treated group. Calculation for the control group is similar. First note we have the following identitity: for any $\theta_1$ and $\theta_2$
\begin{equation}\label{eqn:lowerbound}
	\left|p(\C|\theta_1)-p(\C|\theta_2) \right| \leq ||\C||_2 \times ||\theta_1-\theta_2||_2,\footnote{$||\cdot||_2$ denote the vector $l_2$ norm.}
\end{equation}
by C\ref{assn:MLE}-(iii) and a Taylor expansion argument. Note this implies $p(\C|\theta_2)\in [p(\C|\theta_1)-||\C||_2 \times ||\theta_1-\theta_2||_2,p(\C|\theta_1)+||\C||_2 \times ||\theta_1-\theta_2||_2 ]$.

Now a Taylor expansion of $ \frac{1}{N}\sum_{i=1}^N \frac{Z_i}{p(\C|\widehat{\theta}_{\textnormal{MLE}}) } \mu_i(\Y; d)$ around $\theta_0$ yields
\begin{align}\label{eqn:linearization}
\begin{split}
 & \frac{1}{N}\sum_{i=1}^N \frac{Z_i \mu_i(\Y; d)}{p(\C|\widehat{\theta}_{\textnormal{MLE}}) }=  \frac{1}{N}\sum_{i=1}^N \frac{Z_i \mu_i(\Y; d)}{p(\C|\theta_0) } - \frac{1}{N}\sum_{i=1}^N \frac{Z_i \mu_i(\Y; d)}{p(\C|\theta_0)}(1-p(\C|\theta_0))\C' (\widehat{\theta}_{\textnormal{MLE}}-\theta_0)\\
& + (\widehat{\theta}_{\textnormal{MLE}}-\theta_0)' \left( \frac{1}{2N}\sum_{i=1}^N  \frac{Z_i\mu_i(\Y;d)}{p(\C|\tilde{\theta})} \left(1-p(\C|\tilde{\theta})\right) \C\C' \right)(\widehat{\theta}_{\textnormal{MLE}}-\theta_0),
\end{split}
\end{align}
where $\tilde{\theta}$ is a point between $\theta_0$ and $\widehat{\theta}_{\textnormal{MLE}}$. By  (\ref{eqn:lowerbound}) and $||\widehat{\theta}_{\textnormal{MLE}}-\theta_0 ||_2=O_p(N^{-\frac{1}{2}})$, we have $||\tilde{\theta}-\theta_0||=O_p(N^{-\frac{1}{2}})$, and $\inf_{i\in \node_N} p(\C|\tilde{\theta})$ is bounded away from 0 with probability one. Together with C\ref{assn:bounded-y}, C\ref{assn:MLE}-(ii) and Lemma \ref{lemma:logistic}, we have that the second order term is of order $O_p(N^{-1})$. 

Let $H=\frac{1}{N}\sum_{i=1}^N p(\C|\theta_0)(1-p(\C|\theta_0)\C\C'$. It's clear that by C\ref{assn:bounded-y}, C\ref{assn:assignment_obsstudy}-(iii), C\ref{assn:MLE}-(ii) and (v), and an argument similar to \ref{lemma:ht-var}   :
\begin{equation}
	H^{-1} \left(\frac{1}{N}\sum_{i=1}^N \frac{Z_i\mu_i(\Y;d)}{p(\C|\theta_0)} p(\C|\theta_0)(1-p(\C|\theta_0))\C\right) - \beta_{1,N} = o_p(1). 
\end{equation}
Hence (\ref{eqn:linearization}) becomes 
\begin{equation*}
 \frac{1}{N}\sum_{i=1}^N \frac{Z_i \mu_i(\Y; d)}{p(\C|\widehat{\theta}_{\textnormal{MLE}}) }=  \frac{1}{N}\sum_{i=1}^N \frac{Z_i \mu_i(\Y; d)}{p(\C|\theta_0) } - \frac{1}{N}\sum_{i=1}^N \frac{Z_i-p(\C|\theta_0)}{p(\C|\theta_0)}\C'\beta_{1,N} + o_p(N^{-\frac{1}{2}}).
\end{equation*}
In particular,
\begin{equation*}
 \frac{1}{N}\sum_{i=1}^N \frac{Z_i \mu_i(\Y; d)}{p(\C|\widehat{\theta}_{\textnormal{MLE}}) }-\bar{\mu}_1(d) = \frac{1}{N}\sum_{i=1}^N \C'\beta_{1,N}- \bar{\mu}_1(d) + \frac{1}{N}\sum_{i=1}^n \frac{Z_i\left( \mu_i(\Y;d)-\C'\beta_{1,N} \right)}{p(\C|\theta_0)}+ o_p(N^{-\frac{1}{2}})
\end{equation*}
This proves (\ref{eqn:obs_equivalence}). The calculation of the asymptotic variance, variance upper bound is similar to that in Lemma \ref{lemma:ha-var-homo}. The asymptotic distribution result follows similarly from Lemma \ref{lemma:ogburn-2020-1}.

For the variance estimation results, the proof is similar to that of Proposition \ref{prop:HACvarconsistency}. The only difference is to establish the equivalence between the variance estimator with estimated coefficient $\widehat{\theta}_{\textnormal{MLE}}$ and the variance estimator with true coefficient $\theta_0$. This again can be shown to hold using a Taylor expansion argument together with C\ref{assn:bounded-y}, C\ref{assn:local-inf}, C\ref{assn:spacing}
C\ref{assn:assignment_obsstudy}, C\ref{assn:MLE}-(ii) and (v). We omit the details here for brevity.
\end{proof}

\subsection{Effective Degree of Freedom Adjustment \label{app:edof}}
When the distance of the AME is significant relative to the dataset's total spatial coverage, the reliability of confidence intervals based on a normal approximation may diminish due to small effective sample sizes. This situation is similar to the cluster-robust inference settings with a small number of clusters. To improve the finite-sample coverage of the confidence intervals, we derive the effective degree of freedom adjustment as in \citep{imbens-kolesar2012-robust-small, bell_mccaffrey2002_variance, young-2015-improved-inference}.  

Recall the regression interpretation of the Hajek estimator discussed in (\ref{est-inf}) and the HAC variance estimator $\widehat{\Sigma}_{\HAC}(d)$ in (\ref{HAC_formula}). Define $\mathbf{w}=(0,1)'$. Suppose we want to test the null hypothesis $\textnormal{AME}(d) = \tau_0$, the t-statistic under the null can be written as
\begin{align*}
\frac{\widehat{\tau}(d) - \tau_0}{\sqrt{\mathbf{w}' \widehat{\Sigma}_{\HAC}\mathbf{w}}} = 
\frac{\frac{\widehat{\tau}_\HA(d) - \tau_0}{\sqrt{\mathbf{w}' (\mathbf{X}'\mathbf{X})^{-1}\mathbf{w}}}}{\sqrt{\frac{\mathbf{w}' \widehat{\Sigma}_{\HAC}\mathbf{w}}{\mathbf{w}' (\mathbf{X}'\mathbf{X})^{-1}\mathbf{w}}}}.
\end{align*}

Define $\lambda = \mathbf{w}' (\mathbf{X}'\mathbf{X})^{-1}\X \in\mathbb{R}^{1\times N}$,  $\M = \I - \X'(\mathbf{X}'\mathbf{X})^{-1}\X\in\mathbb{R}^{N\times N}$, and the weight matrix $\Omega\in\mathbb{R}^{\N\times N}$ where $\Omega_{ij}= \mathbf{1}\{j \in \mathcal{B}(i;d) \}$.  \cite{young-2015-improved-inference}'s effective degree of freedom adjustment in our setting is calculated as:
\begin{equation}
	\mu = \mathbf{Trace}\left( \frac{N_1 N_0}{N} \M \left(\Omega \circ  
\left(\lambda'\lambda\right) \right) \M\right),
\end{equation}
where $\circ$ denotes the pointwise (Hadamard) matrix product. The estimated variance is inflated by $\frac{1}{\mu}$, becoming $\frac{\mathbf{w}' \widehat{\Sigma}_{\HAC}\mathbf{w}}{\mu}$. We further define the quantity:
\begin{equation}
	\nu = 2*\mathbf{Trace}\left( \left(\frac{N_1 N_0}{N} \M \left(\Omega \circ  
	\left(\lambda'\lambda\right) \right) \M\right)\left(\frac{N_1 N_0}{N} \M \left(\Omega \circ  
	\left(\lambda'\lambda\right) \right) \M\right)\right).
\end{equation}	
We use the $\frac{\alpha}{2}$-quantile and $\left(1-\frac{\alpha}{2}\right)$-quantile of the t-distribution with $\frac{2\mu^2}{v}$ degree of freedom to construct the confidence interval. 

\newpage

\section{Simulation Results}\label{Section:Simulation_Appendix}

\subsection{Simulation Designs}
\textbf{Outcome Points and Intervention Nodes}\\
Let $S=\{80,100,120\}$. We first generate a raster with $S \times S$ tiles, each of which is an outcome point. The side length of each tile is $1$ generic unit. The untreated potential outcome for the outcome point $x$, $Y_x(0)$, is randomly drawn from the standard normal distribution.

For point interventions, we coarsen the raster into $\frac{S^2}{100}$ tiles and random sample half of the tiles. For each sampled tile, we add perturbations to the center of the tile twice to create a pair of intervention points. 

For polygon interventions, we subsample $\frac{S^2}{10}$ tiles, use the sampled tiles and the Voronoi tessellation to generate polygons. To construct the set of intervention nodes, we randomly sample $\frac{S^2}{200}$ polygons and, for each sampled polygons, we randomly sample an adjacent polygon. This gives a total of $\frac{S^2}{100}$ polygons as intervention nodes.

\noindent \textbf{Data Generating Process}\\
Let $\Gamma(d,a,b)$ denote the density of a gamma distribution at value  d with shape parameter a and scale parameter b.  
For an outcome point $x$, we define the effect function:
\begin{equation}
f_x(d) =  3\alpha_x\times\left(\Gamma(d; 1, 1) - \Gamma(d; 5, 0.5)\right) \times \max\left( \left(1-\frac{d^2}{36}\right),0\right),
\end{equation}
where $\alpha_x$ captures the heterogeneity in treatment effects among outcome points. $\alpha_x$ are generated from kriging interpolation of some randomly generated values on a coarsened raster. The term $ \max\left( \left(1-\frac{d^2}{36}\right),0\right)$ is used to guarantee that there is no treatment effect beyond 6 units of distance.

For the additive-effect case, the outcome at point $x$, $Y_x(\Z)$, is generated by the formula
\begin{equation}
Y_x(\Z) = Y_x(0) + \sum_{i=1}^{n} f_x(d_{ix})Z_i,
\end{equation}
where $n$ denotes the number of intervention nodes, and $d_{ix}$ is the distance from the outcome point $x$ to the intervention node $i$. 

For the interactive-effect case, we define another effect function:
\begin{equation}
	g_x(d) =  3\alpha_x\times\left(\Gamma(d; 5, 0.5)\right) \times \max\left( \left(1-\frac{d^2}{36}\right),0\right).
\end{equation}
The outcome at point $x$, $Y_x(\Z)$, is generated by the formula
\begin{equation}
	Y_x(\Z) = Y_x(0) + \sum_{i=1}^{n} f_x(d_{ix})Z_i + \sum_{i=1}^{n} g_x(d_{ix})Z_i Z_{\mathcal{N}(i)},
\end{equation}
where $\mathcal{N}(i)$ denotes the intervention node that is closest to the intervention node $i$.  When there are multiple polygons that are adjacent to the $i$th polygon, we pick the polygon with the smallest assigned index to be $\mathcal{N}(i)$.

\noindent \textbf{Experimental Designs}\\
For our simulation, we use a Bernoulli design where $Z_i=1$ with probability 0.5 for all intervention nodes. We run the simulation for 2000 times. 
Denote the $p$th assignment of $\Z$ as $\Z_p$ and the value of $Z_i$ under the assignment as $Z_{pi}$. To obtain AME, we first calculate
\begin{equation}
\begin{aligned}
	\tau_{ix}(\eta) \approx \frac{\sum_{p=1}^{2000}Z_{pi}Y_x(\Z_p)}{\sum_{p=1}^{2000}Z_{pi}} - \frac{\sum_{p=1}^{2000}(1-Z_{pi})Y_x(\Z_p)}{\sum_{p=1}^{2000}(1-Z_{pi})}.
\end{aligned}
\end{equation}
Then $\tau_i(d; \eta)$ and the AME can be constructed following their definitions.

\subsection{More Point-Intervention Simulation Results in Section \ref{Section:Simulation}  }

We report simulation results  for the Smoothed Hajek estimator with a triangular kernel and a bandwidth 1. Figure \ref{fig:MSE_sm_Hajek_point} reports results on MSE, Figure  \ref{fig:AME_sm_hajek_additive_point} on coverage and half length for the additive case, and Figure \ref{fig:AME_sm_hajek_interactive_point} on  coverage and half length for the interactive case

\begin{figure*}[h]
	\centering
	\begin{subfigure}[b]{0.5\textwidth}
		\centering
		\includegraphics[scale=0.45]{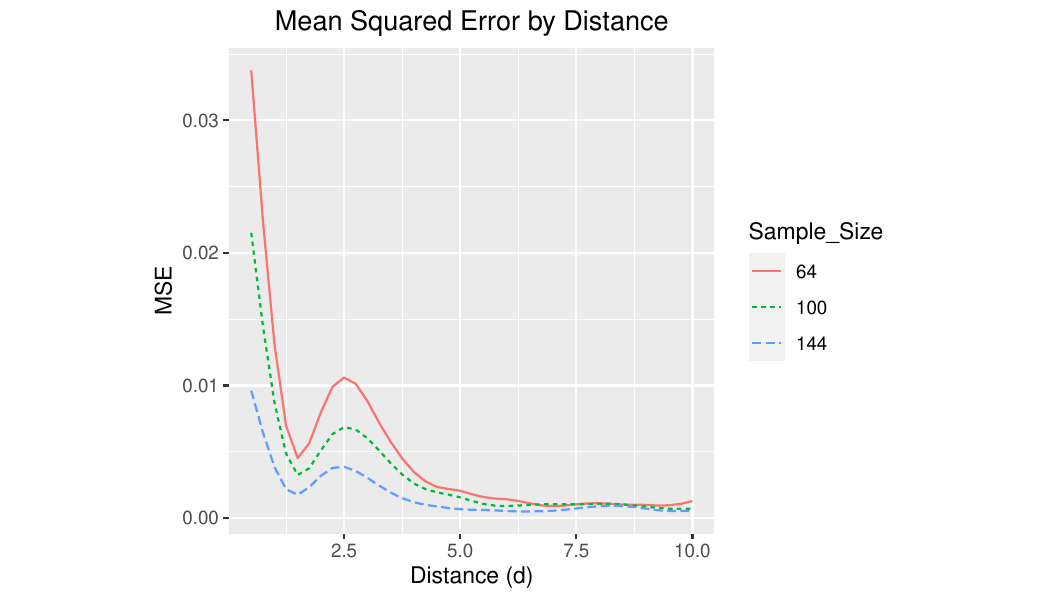}
		\caption{MSE for the additive case (\ref{additive_effect})}
	\end{subfigure}%
	~ 
	\begin{subfigure}[b]{0.5\textwidth}
		\centering
		\includegraphics[scale=0.45]{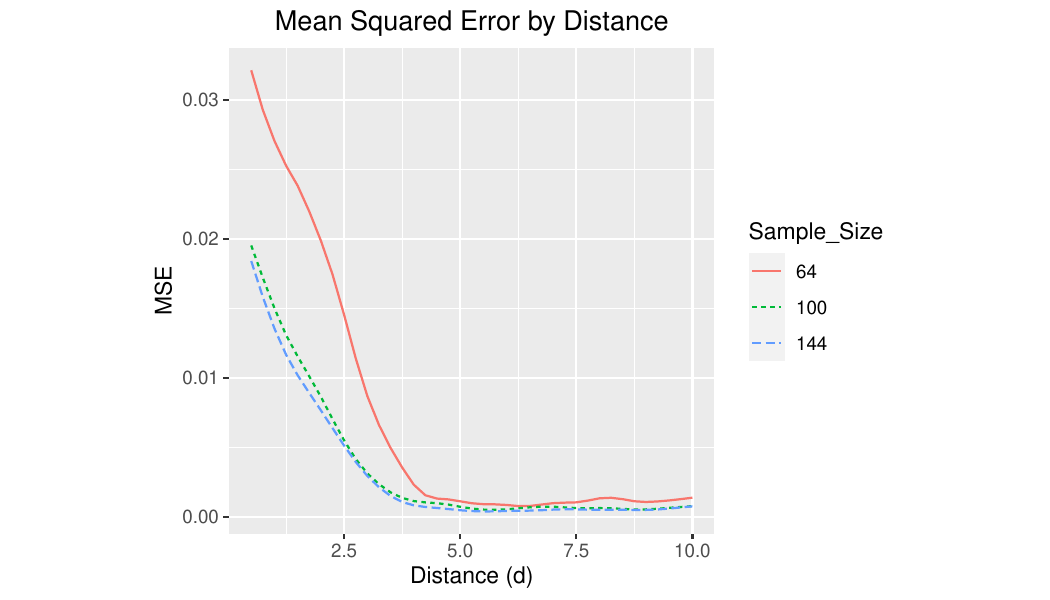}
		\caption{MSE for the interactive case (\ref{interactive_effect})}
	\end{subfigure}
	\caption{The left and right figures report the Mean Squared Errors of the Smoothed Hajek estimator in the additive-effect case and the interactive-effect case, respectively.  }
	\label{fig:MSE_sm_Hajek_point}
\end{figure*}

\begin{figure}[h]
	\centering
	\includegraphics[width=0.5\textwidth]{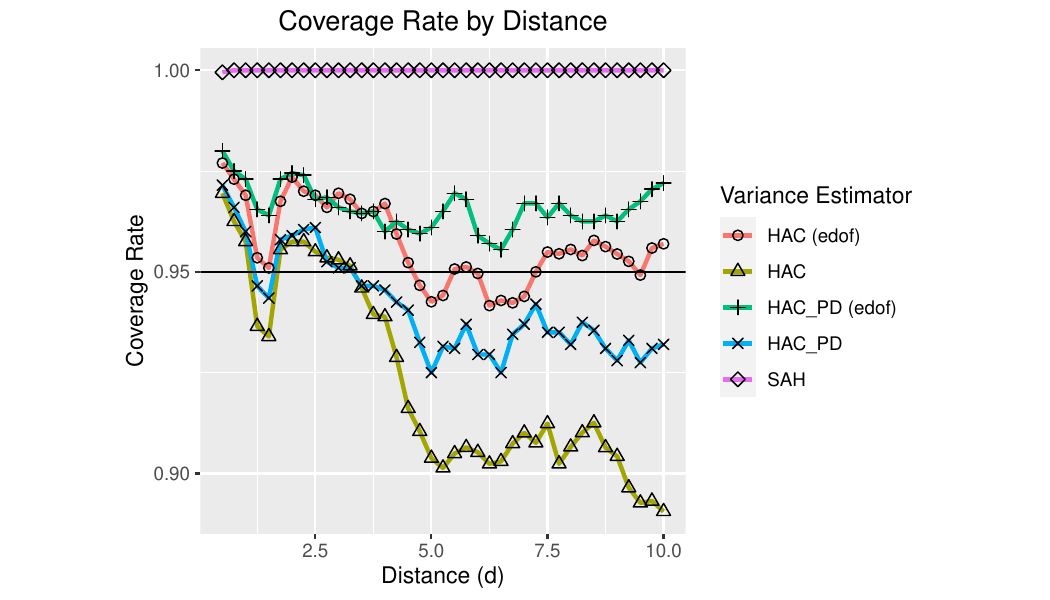}\includegraphics[width=0.5\textwidth]{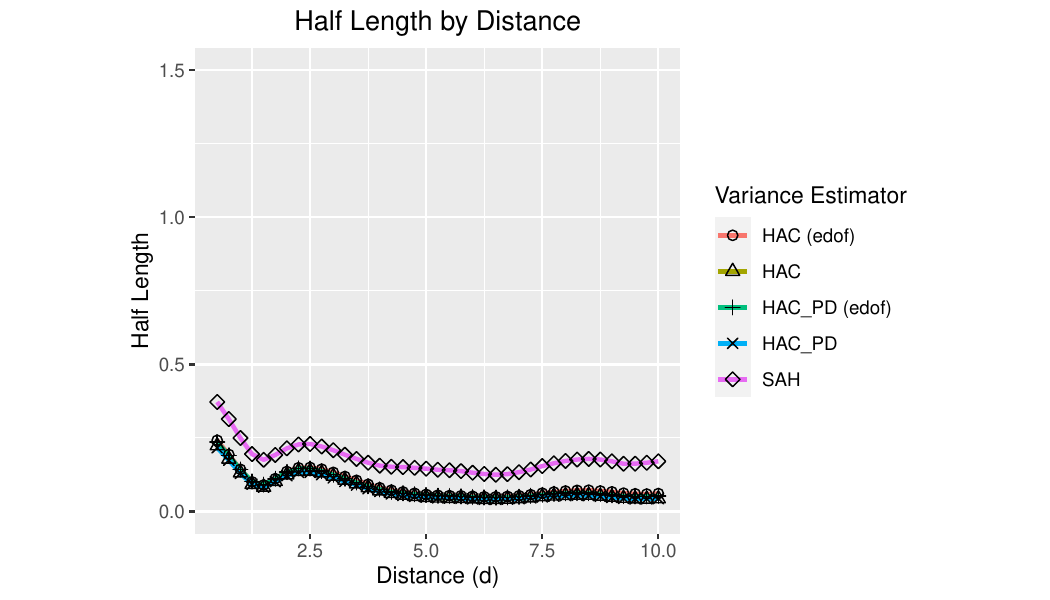}
	\caption{ Point-intervention simulation results on the coverage rates and half lengthes of two-sided 95\% confidence intervals with the Smoothed Hajek estimator and different variance estimators in the additive effect case (\ref{additive_effect}). The sample size is 144. HAC refers to the CI with the HAC variance estimator in (\ref{HAC_formula}) and a normal critical value. The length and coverage of the HAC CI is assesed with respect to the cases where HAC estimator returns a nonnegative value. HAC\_PD refers to the CI with positive-semidefinite HAC variance estimator in (\ref{formula:HAC_PD}) and a normal critical value.  HAC (edof) refers to the CI with HAC variance estimator in (\ref{HAC_formula}) and empirical degree of freedom adjustment. HAC\_PD (edof) refers to the CI with HAC variance estimator in (\ref{formula:HAC_PD})  and empirical degree of freedom adjustment. SAH refers to the CI with  SAH variance estimator (\ref{varianceSAH}).  }
	\label{fig:AME_sm_hajek_additive_point}
\end{figure}

\begin{figure}[h]
	\centering
	\includegraphics[width=0.5\textwidth]{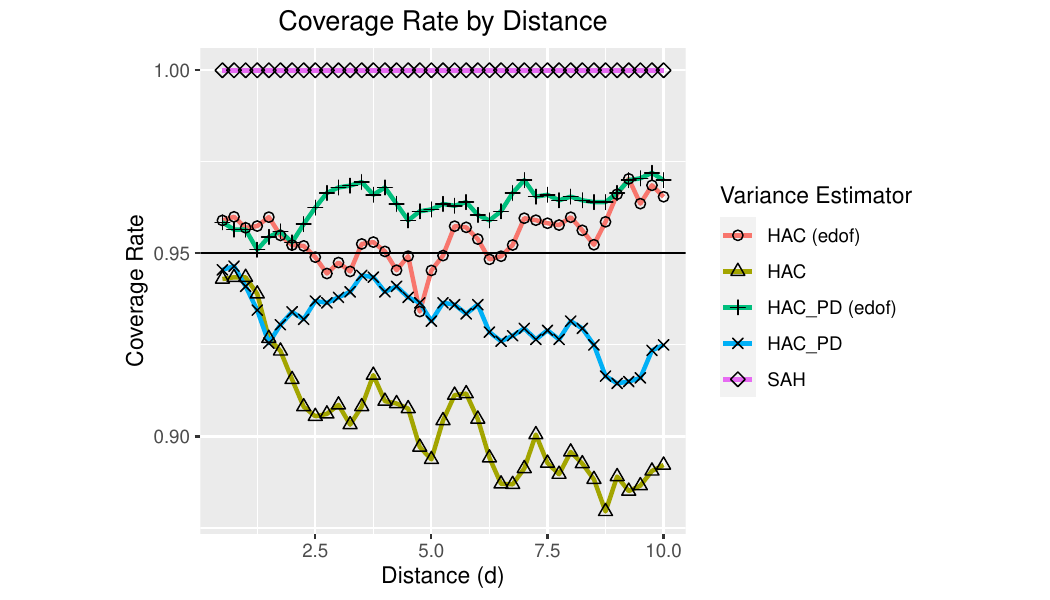}\includegraphics[width=0.5\textwidth]{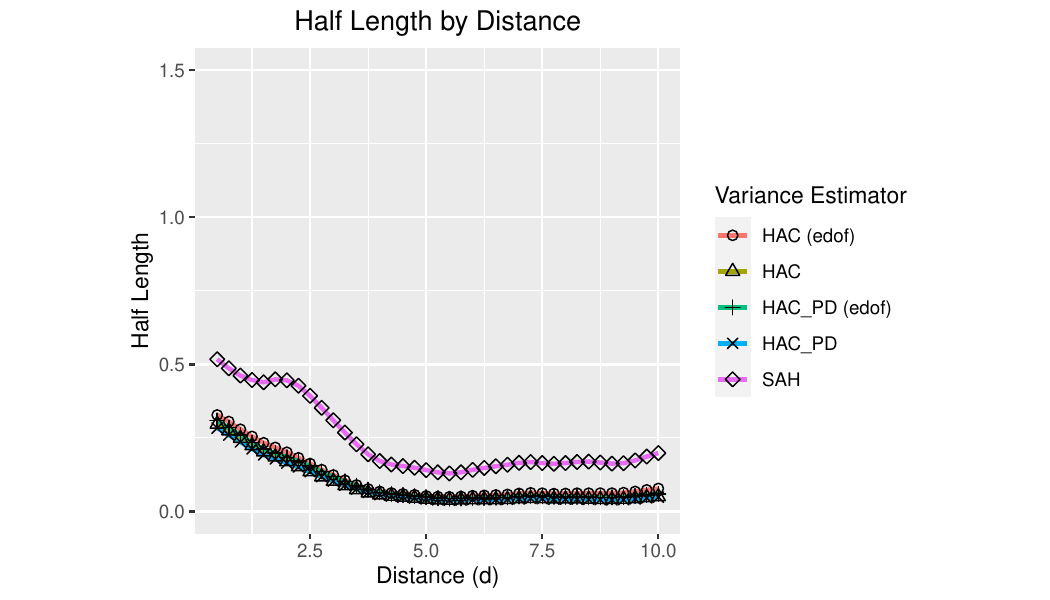}
	\caption{ Point-intervention simulation results on the coverage rates and half lengthes of two-sided 95\% confidence intervals with the Smoothed Hajek estimator and different variance estimators in the interactive effect case (\ref{interactive_effect}). The sample size is 144. HAC refers to the CI with the HAC variance estimator in (\ref{HAC_formula}) and a normal critical value. The length and coverage of the HAC CI is assesed with respect to the cases where HAC estimator returns a nonnegative value. HAC\_PD refers to the CI with positive-semidefinite HAC variance estimator in (\ref{formula:HAC_PD}) and a normal critical value.  HAC (edof) refers to the CI with HAC variance estimator in (\ref{HAC_formula}) and empirical degree of freedom adjustment. HAC\_PD (edof) refers to the CI with HAC variance estimator in (\ref{formula:HAC_PD})  and empirical degree of freedom adjustment. SAH refers to the CI with  SAH variance estimator (\ref{varianceSAH}).  }
	\label{fig:AME_sm_hajek_interactive_point}
\end{figure} 

\subsection{Results for a Polygon-Intervention Simulation}

We report simulation results for the Hajek estimator with a polygon intervention simulation. Figure \ref{fig:MSE_Hajek_polygon} reports results on MSE, Figure \ref{fig:AME_hajek_additive_polygon} on coverage and half length for the additive case, and Figure \ref{fig:AME_hajek_interactive_polygon} on coverage and half length for the interactive case.

\begin{figure*}[h]
	\centering
	\begin{subfigure}[b]{0.5\textwidth}
		\centering
		\includegraphics[scale=0.45]{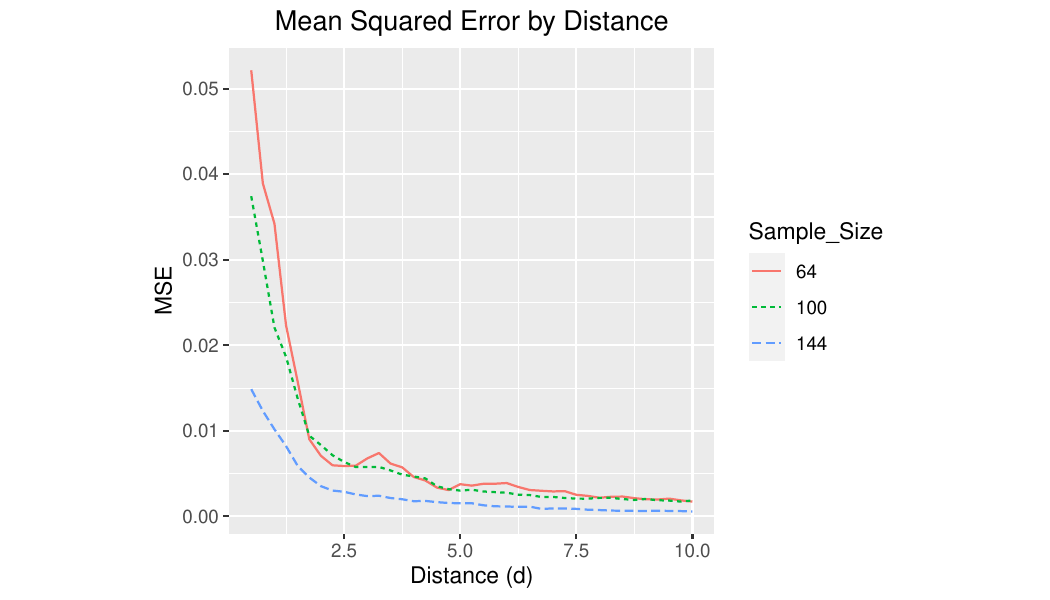}
		\caption{MSE for the additive case (\ref{additive_effect})}
	\end{subfigure}%
	~ 
	\begin{subfigure}[b]{0.5\textwidth}
		\centering
		\includegraphics[scale=0.45]{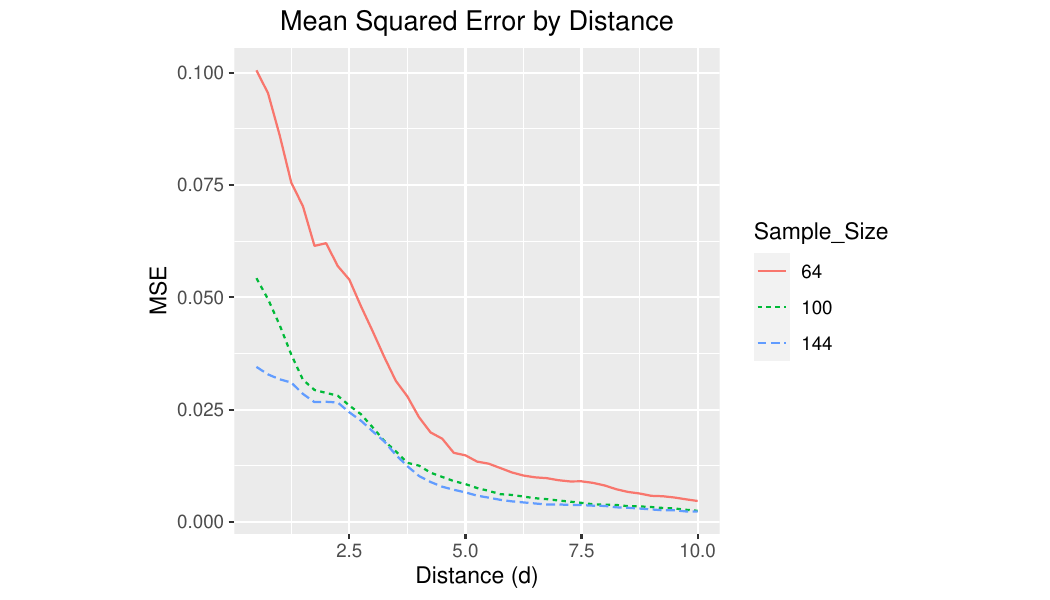}
		\caption{MSE for the interactive case (\ref{interactive_effect})}
	\end{subfigure}
	\caption{The left and right figures report the Mean Squared Errors of the Hajek estimator in the additive-effect case and the interactive-effect case, respectively.  }
	\label{fig:MSE_Hajek_polygon}
\end{figure*}

\begin{figure}[h]
	\centering
	\includegraphics[width=0.5\textwidth]{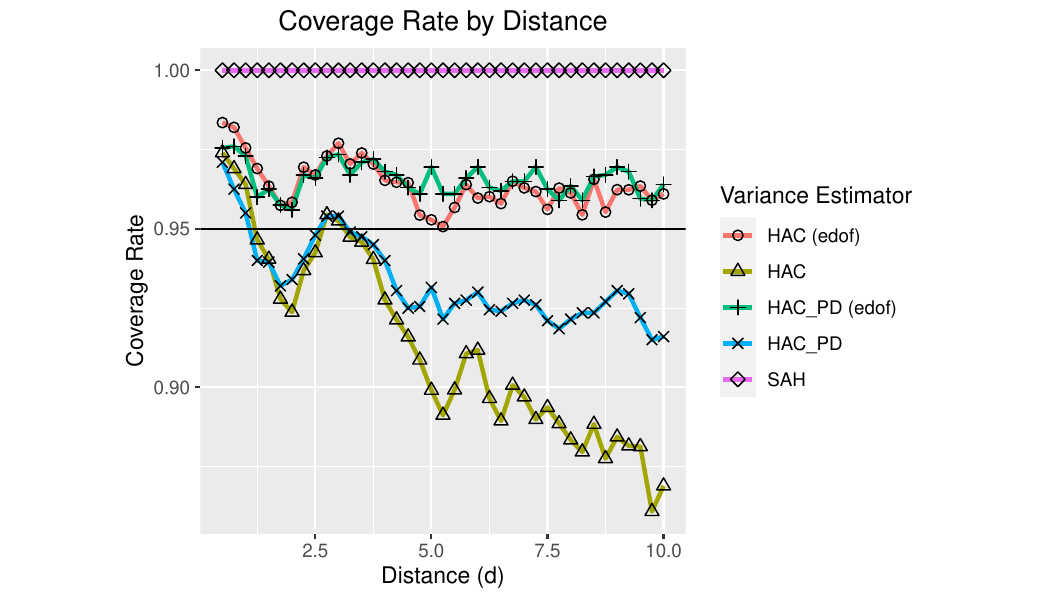}\includegraphics[width=0.5\textwidth]{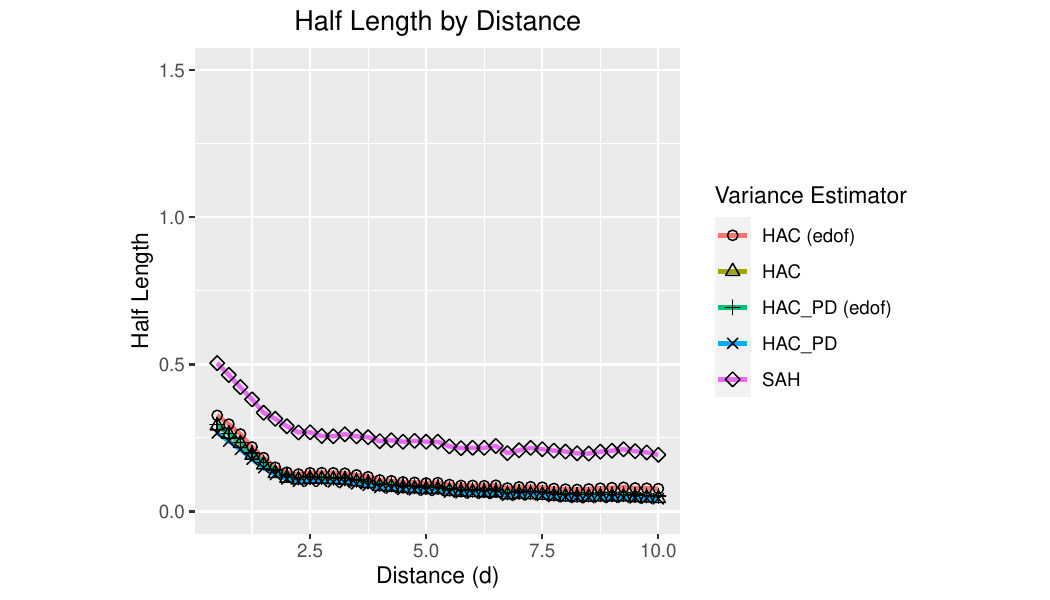}
	\caption{ Polygon-intervention simulation results on the coverage rates and half lengthes of two-sided 95\% confidence intervals with the Hajek estimator and different variance estimators in the additive effect case (\ref{additive_effect}). The sample size is 144. HAC refers to the CI with the HAC variance estimator in (\ref{HAC_formula}) and a normal critical value. The length and coverage of the HAC CI is assesed with respect to the cases where HAC estimator returns a nonnegative value. HAC\_PD refers to the CI with positive-semidefinite HAC variance estimator in (\ref{formula:HAC_PD}) and a normal critical value.  HAC (edof) refers to the CI with HAC variance estimator in (\ref{HAC_formula}) and empirical degree of freedom adjustment. HAC\_PD (edof) refers to the CI with HAC variance estimator in (\ref{formula:HAC_PD})  and empirical degree of freedom adjustment. SAH refers to the CI with  SAH variance estimator (\ref{varianceSAH}).  }
	\label{fig:AME_hajek_additive_polygon}
\end{figure}

\begin{figure}[h]
	\centering
	\includegraphics[width=0.5\textwidth]{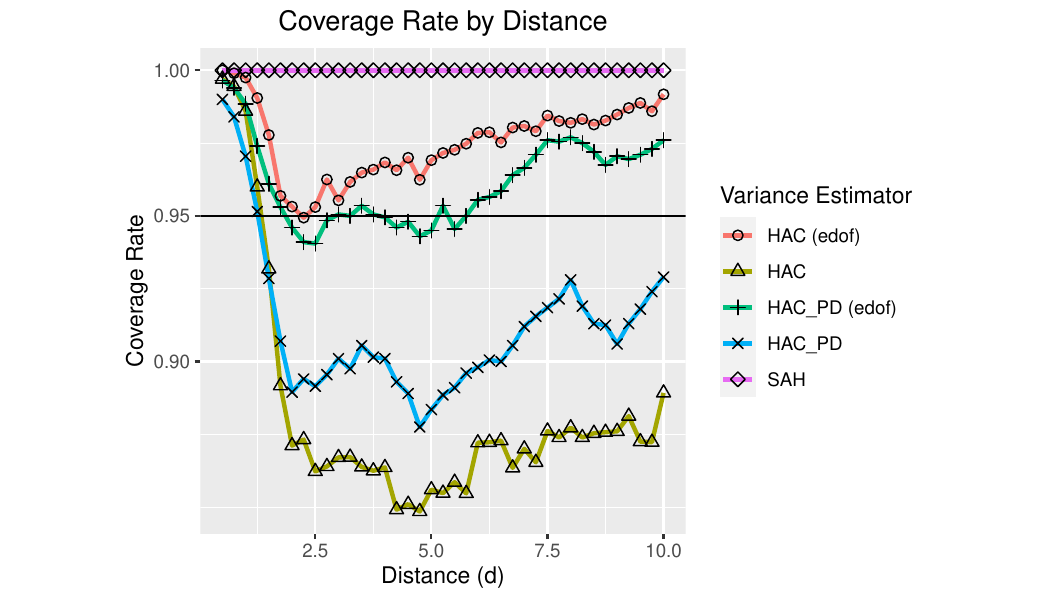}\includegraphics[width=0.5\textwidth]{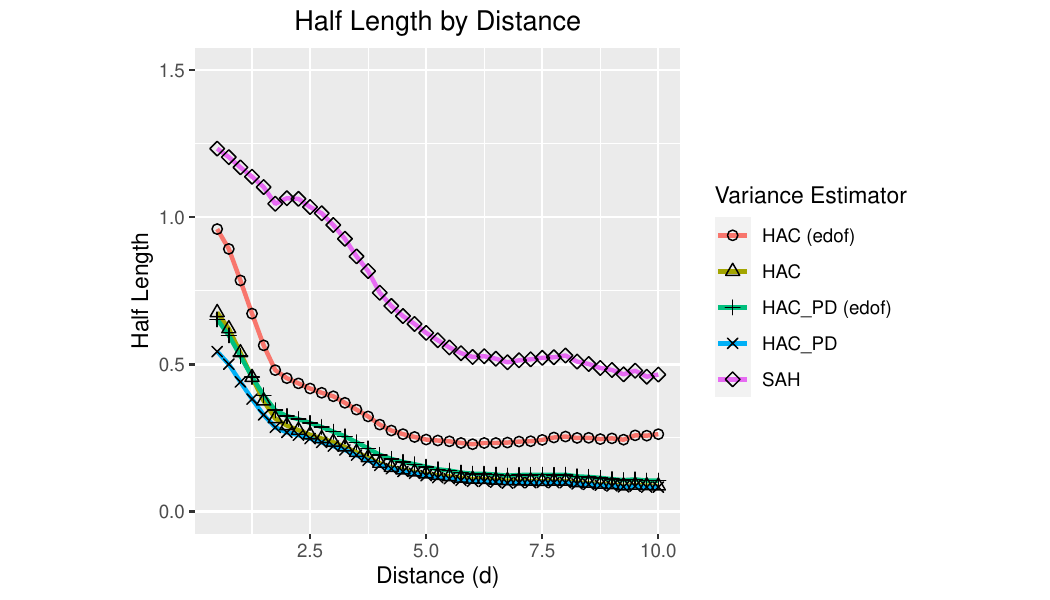}
	\caption{ Polygon-intervention simulation results on the coverage rates and half lengths of two-sided 95\% confidence intervals with the Hajek estimator and different variance estimators in the interactive effect case (\ref{interactive_effect}). The sample size is 144. HAC refers to the CI with the HAC variance estimator in (\ref{HAC_formula}) and a normal critical value. The length and coverage of the HAC CI is assesed with respect to the cases where HAC estimator returns a nonnegative value. HAC\_PD refers to the CI with positive-semidefinite HAC variance estimator in (\ref{formula:HAC_PD}) and a normal critical value.  HAC (edof) refers to the CI with HAC variance estimator in (\ref{HAC_formula}) and empirical degree of freedom adjustment. HAC\_PD (edof) refers to the CI with HAC variance estimator in (\ref{formula:HAC_PD})  and empirical degree of freedom adjustment. SAH refers to the CI with  SAH variance estimator (\ref{varianceSAH}).  }
	\label{fig:AME_hajek_interactive_polygon}
\end{figure}

We report simulation results for the Smoothed Hajek estimator with a polygon intervention simulation. We use a triangular kernel and bandwidth 1. Figure \ref{fig:MSE_Hajek_polygon_sm} reports results on MSE, Figure  \ref{fig:AME_hajek_additive_polygon_sm} reports results on coverage and half length for the additive case, and Figure  \ref{fig:AME_hajek_interactive_polygon_sm} on coverage and half length for the interactive case.

\begin{figure*}[h]
	\centering
	\begin{subfigure}[b]{0.5\textwidth}
		\centering
		\includegraphics[scale=0.45]{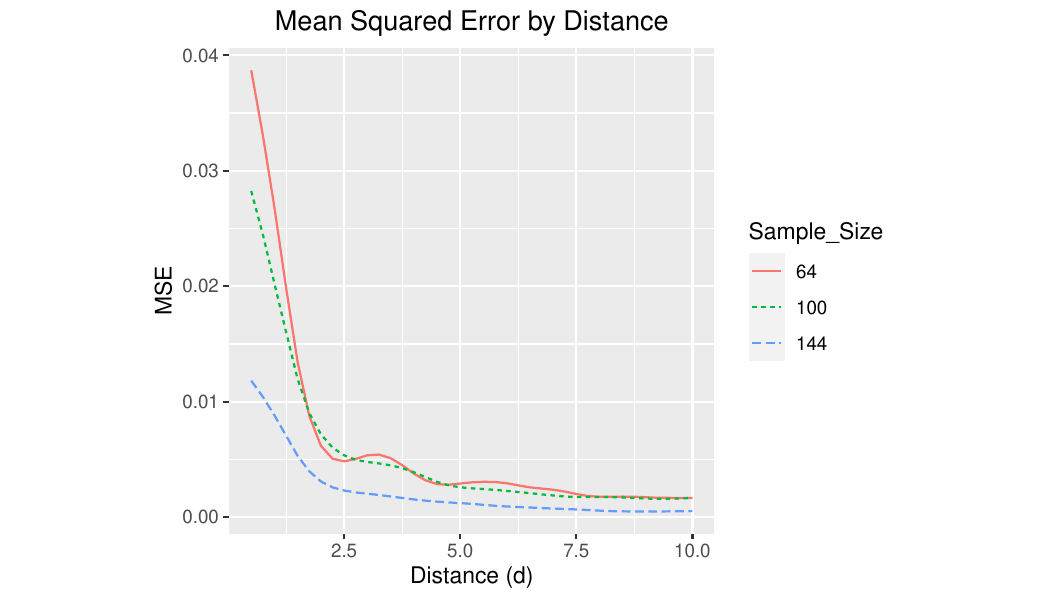}
		\caption{MSE for the additive case (\ref{additive_effect})}
	\end{subfigure}%
	~ 
	\begin{subfigure}[b]{0.5\textwidth}
		\centering
		\includegraphics[scale=0.45]{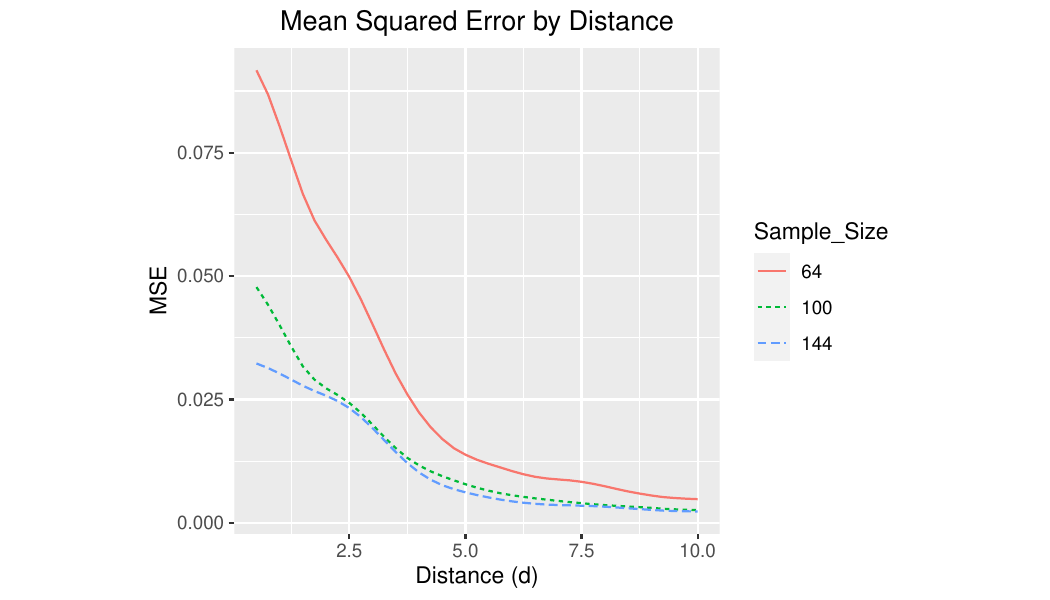}
		\caption{MSE for the interactive case (\ref{interactive_effect})}
	\end{subfigure}
	\caption{The left and right figures report the Mean Squared Errors of the Smoothed Hajek estimator in the additive-effect case and the interactive-effect case, respectively.  }
	\label{fig:MSE_Hajek_polygon_sm}
\end{figure*}

\begin{figure}[h]
	\centering
	\includegraphics[width=0.5\textwidth]{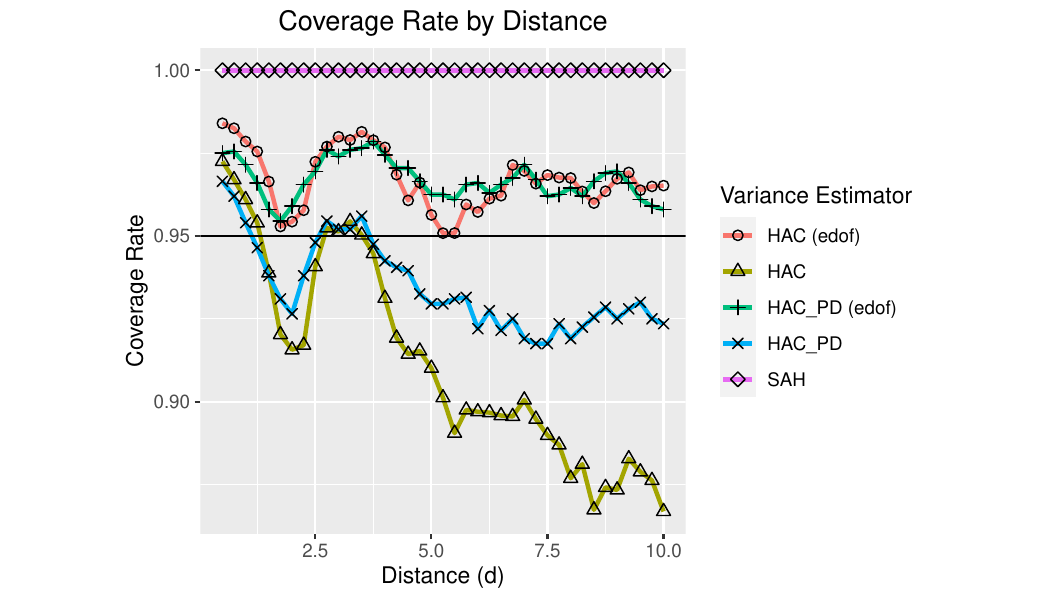}\includegraphics[width=0.5\textwidth]{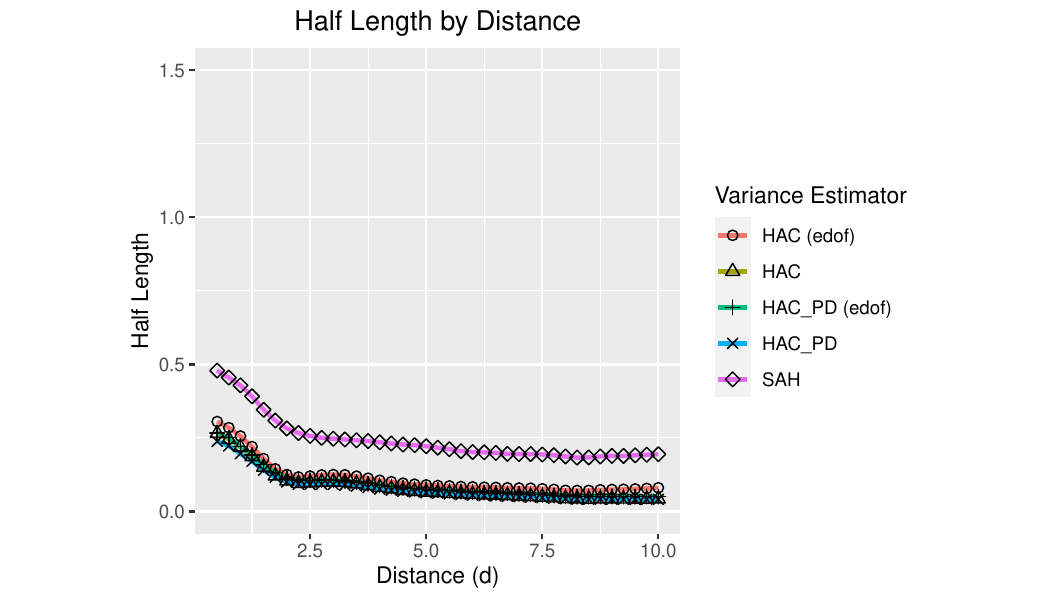}
	\caption{ Polygon-intervention simulation results on the coverage rates and half lengthes of two-sided 95\% confidence intervals with the Smoothed Hajek estimator and different variance estimators in the additive effect case (\ref{additive_effect}). The sample size is 144. HAC refers to the CI with the HAC variance estimator in (\ref{HAC_formula}) and a normal critical value. The length and coverage of the HAC CI is assesed with respect to the cases where HAC estimator returns a nonnegative value. HAC\_PD refers to the CI with positive-semidefinite HAC variance estimator in (\ref{formula:HAC_PD}) and a normal critical value.  HAC (edof) refers to the CI with HAC variance estimator in (\ref{HAC_formula}) and empirical degree of freedom adjustment. HAC\_PD (edof) refers to the CI with HAC variance estimator in (\ref{formula:HAC_PD})  and empirical degree of freedom adjustment. SAH refers to the CI with  SAH variance estimator (\ref{varianceSAH}).  }
	\label{fig:AME_hajek_additive_polygon_sm}
\end{figure}

\begin{figure}[h]
	\centering
	\includegraphics[width=0.5\textwidth]{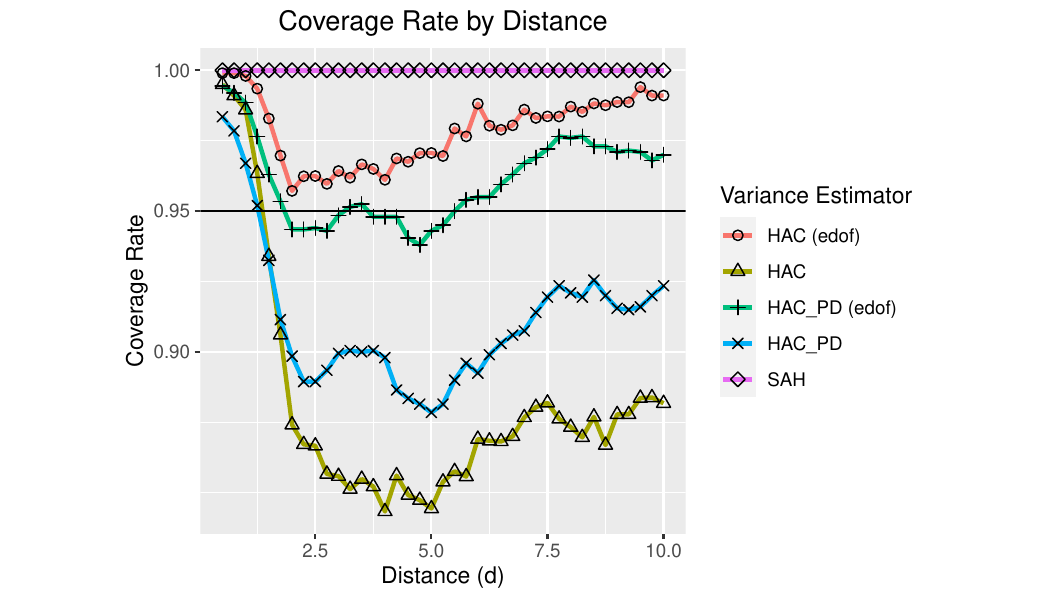}\includegraphics[width=0.5\textwidth]{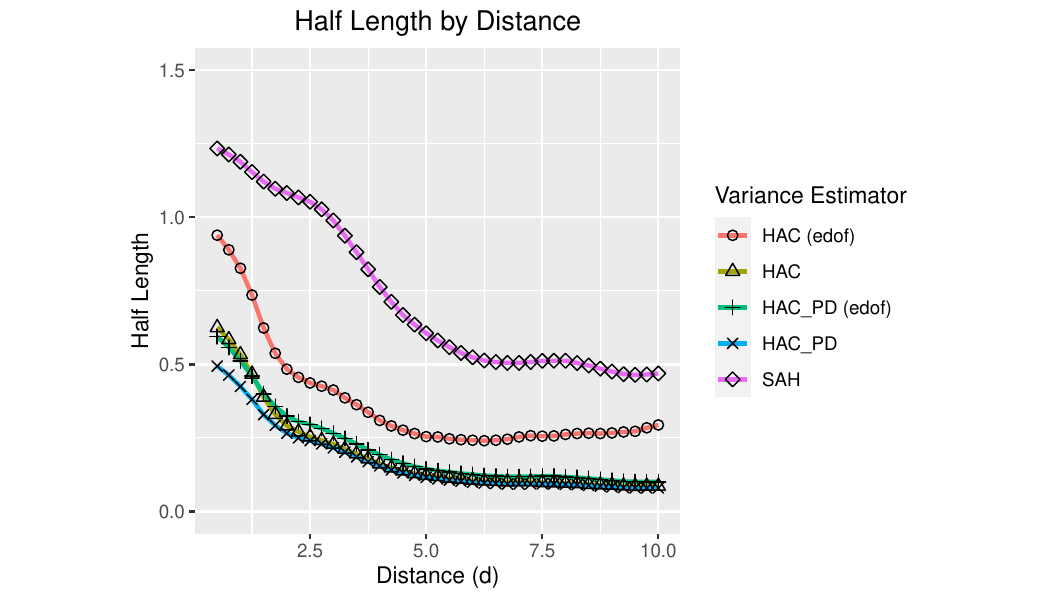}
	\caption{ Polygon-intervention simulation results on the coverage rates and half lengthes of two-sided 95\% confidence intervals with the Smoothed Hajek estimator and different variance estimators in the interactive effect case (\ref{interactive_effect}). The sample size is 144. HAC refers to the CI with the HAC variance estimator in (\ref{HAC_formula}) and a normal critical value. The length and coverage of the HAC CI is assesed with respect to the cases where HAC estimator returns a nonnegative value. HAC\_PD refers to the CI with positive-semidefinite HAC variance estimator in (\ref{formula:HAC_PD}) and a normal critical value.  HAC (edof) refers to the CI with HAC variance estimator in (\ref{HAC_formula}) and empirical degree of freedom adjustment. HAC\_PD (edof) refers to the CI with HAC variance estimator in (\ref{formula:HAC_PD})  and empirical degree of freedom adjustment. SAH refers to the CI with  SAH variance estimator (\ref{varianceSAH}).  }
	\label{fig:AME_hajek_interactive_polygon_sm}
\end{figure}

\clearpage

\subsection{Additional Results for Empirical Applications}

Figure~\ref{fig:jaya-forests} provides more details on the construction of the new outcome variable for use in our re-analysis of the \citet{Jayachandran267} experiment.

\begin{figure}[b]\centering 
\includegraphics[width=0.4\textwidth]{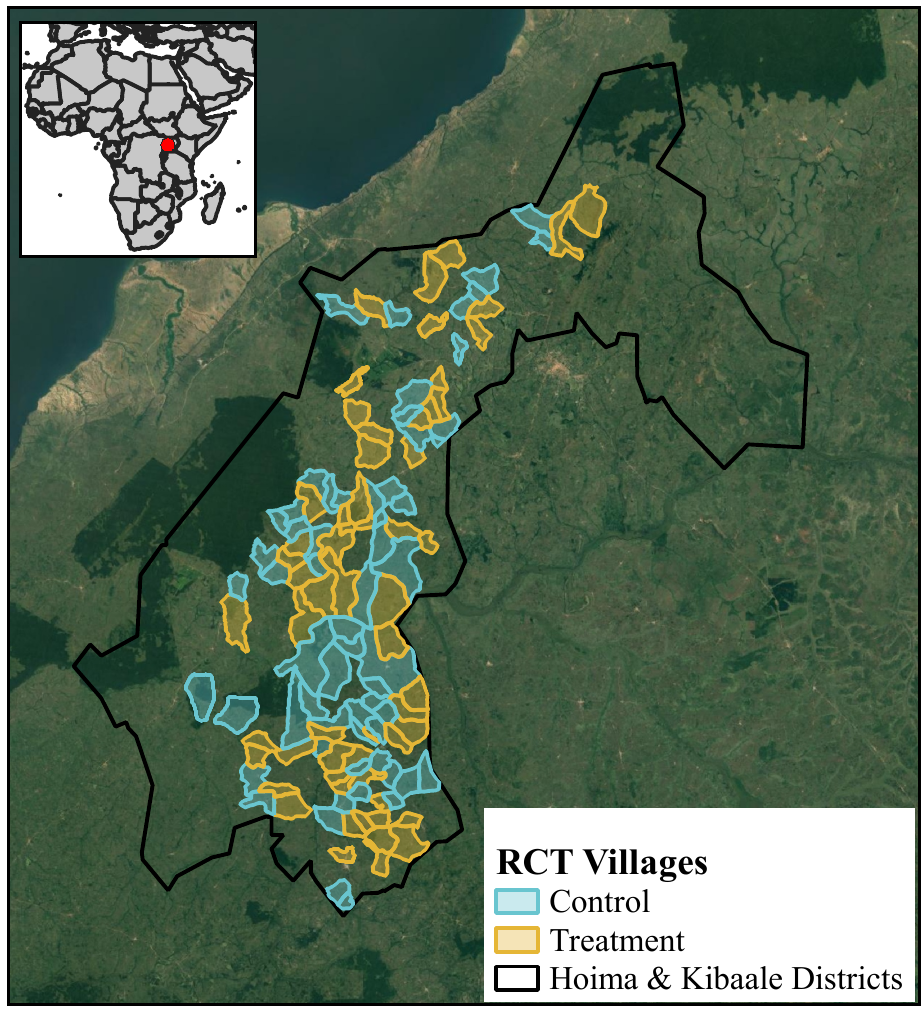} \\
\includegraphics[width=.8\textwidth]{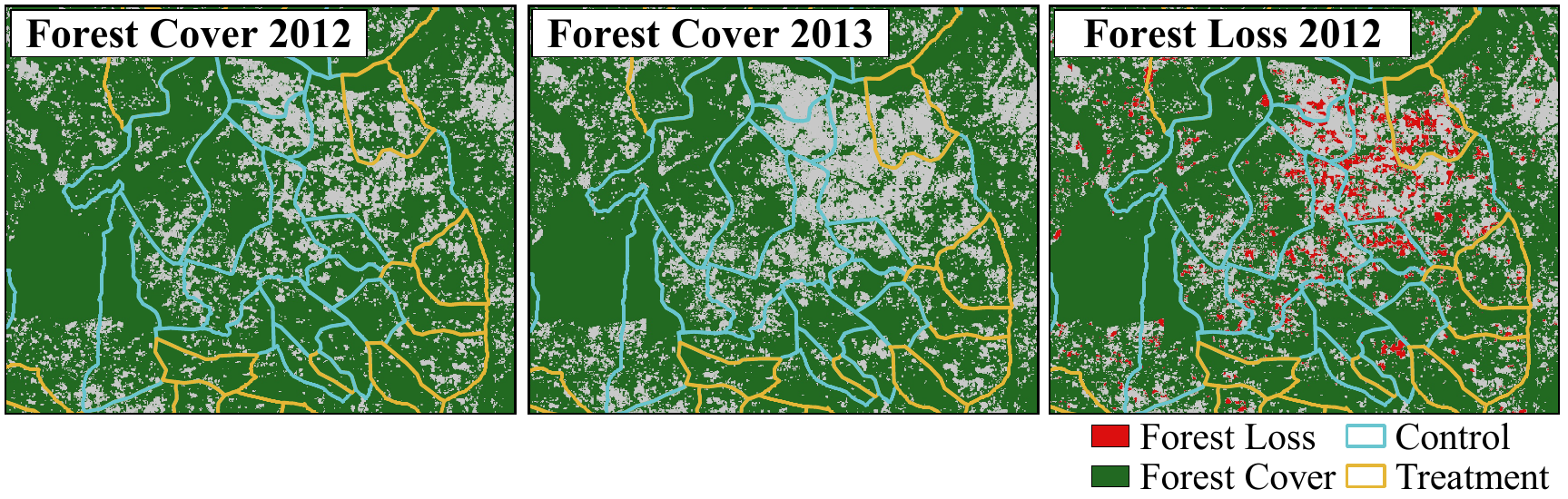}
\caption{\label{fig:jaya-forests} Top plot: Study area of randomized control trial for a PES program in Hoima and Kibaale district in Uganda, from \citet{Jayachandran267}. Boundaries of treatment (60) and control (61) villages were digitized using publicly available data and published maps. Bottom plots: The Global Forest Cover (GFC) dataset over a subset of the study area showing forest cover for 2012, 2013, and forest loss in 2012 \citep{hansen2013high}.}
\end{figure}


\clearpage

\section{Weaker assumptions on the extent of interference}
\subsection{Weaker Assumptions on the Extent of Interference}\label{Section:weak_interference}

This section extends our inferential results on the Hajek estimator by relaxing the local interference assumption C\ref{assn:local-inf} such that spatial dependence does not have to be contained within a strict distance cutoff. 
We allow for more general spatial ``near-epoch dependence'' \citep{jenish2012spatial} and provide results on root-N consistency, asymptotic normality, and HAC variance estimation.\footnote{The near-epoch dependence overcomes limitations of the spatial mixing literature, for example, as in \cite{jenish_prucha2009_clt_spatial_interaction}. The mixing condition may be too stringent for design-based causal inference settings, in which outcomes are modeled as a function of Bernoulli random variables. See \cite{andrews1984non} and \cite{doukhan2002rates}.   } 
We do not pursue the most general results (e.g. the most relaxed conditions on data moments and correlation structures). Rather, we impose assumptions that are commensurable with the assumptions in the main text. Proofs of root-N consistency and asymptotic normality are standard and follow as special cases of the results in \cite{jenish2012spatial}. The proof of the consistency of the HAC variance estimator is relatively new, as far as we know.\footnote{The logic of our proof closely follows \cite{conley99_spatial}, but some estimates are derived in the near-epoch dependency setup rather than the spatial-mixing setup. \cite{davidson2020new} studies a similar problem but the setup and assumption are different from this paper.}

We introduce the near-epoch dependence concept. Our version is a simplified version from \cite{jenish2012spatial}. Denote the sample size by $N$. Let $\node_N$ be a set of intervention nodes equipped with a distance metric $\gamma$ that satisfies positivity and the triangle inequality.  Let $H_N=\{H_{i,N}, i\in \mathcal{S}_N\}$ and $V_N=\{V_{i,N},i\in \mathcal{S}_N\}$ be two arbitrary sets of random variables. Define $\mathcal{F}_i(s)=\sigma\left( V_{j,N},j\in \mathcal{S}_n, \gamma(i,j)\leq s\right)$, the $\sigma$-field generated by the random variables $\{V_{j,N}\}$ located in the s-neighborhood of the intervention node $i$. Denote the $L_p$-norm of a random variable as $||X||_p=\left(E[|X|^p]\right)^{\frac{1}{p}}$.

\begin{definition}[$L_2$-NED]
	Let $H_N=\{H_{i,N}, i\in\node_N, N\geq 1\}$ be a random field with $||H_{i,N}||_2
 <\infty$. Let $V_N=\{V_{i,N},i\in\node_N, N\geq 1\}$ be a random field, where $|\node_N|=N$, and let $g_N=\{g_{i,N},i\in\node_N,N\geq 1\}$ be a set of positive constants. Then the random field $H_N$ is said to be $L_2(g_N)$-near-epoch dependent on the random field $V_N$ if, for all $i$,
	\begin{equation}
		||H_{i,N}-\E[H_{i,N}|\mathcal{F}_{i}(s)]||_2\leq g_{i,N}\psi(s),
	\end{equation}
	for a non-increasing sequence $\psi(s)\geq 0$ with $\lim_{s\to\infty}\psi(s)=0$. $H_N$ is said to be $L_2-$NED on $V_N$ of size $-\lambda$ if $\psi(s)=O(s^{-\mu})$ for some $\mu>\lambda>0$. Further more, if $\sup_{N}\sup_{i\in\node_N}{g_{i,N}}<\infty$, then $H_N$ is said to be uniformly $L_2$-NED on $V_N$.  
\end{definition}

We establish several implications of the $L_2$-NED property. The following lemma is inspired by Theorem 17.5 in \cite{davidson1994stochastic}.

Let $\mathcal{F}_i^c(s)=\sigma\left( V_{j,N},j\in \mathcal{S}_n, \gamma(i,j)> s\right)$, the $\sigma$-field generated by the random variables $\{V_{j,N}\}$ that are more than $s$ distance away from the intervention node $i$. The first lemma below states that $L_2$-NED random variables are nearly constant when conditioned on random variables $\{V_{j,N}\}$ that are far away. This is an adaption of the mixingale property from the time series literature to our spatial process setting.\footnote{For a reference, see Chapter 17.2 of \cite{davidson1994stochastic}.}
\begin{lemma}\label{lemma:mixingale}
	Let $V_N$ be a set of independent random variables, and $H_N$ be a set of zero-mean $L_2(g_N)$-NED random variables on $V_N$, then $||\E[H_{i,N}|\mathcal{F}_i^c(s)]||_2\leq g_{i,N}\psi(s)$.
\end{lemma}

\begin{proof}
We have the following chain of inequalities:
\begin{align}
	& ||\E[H_{i,N}|\mathcal{F}_i^c(s)]||_2 \\
	 \leq & ||\E\left[\left(H_{i,N}- \E[H_{i,N}|\mathcal{F}_i(s)]\right)|\mathcal{F}_i^c(s)\right]||_2 + ||\E[\left(\E[H_{i,N}|\mathcal{F}_i(s)]\right)|\mathcal{F}_i^c(s)]||_2 \label{83}\\
	 \leq & ||H_{i,N}- \E[H_{i,N}|\mathcal{F}_i(s)]||_2 + ||\E[\left(\E[H_{i,N}|\mathcal{F}_i(s)]\right)|\mathcal{F}_i^c(s)]||_2 \label{84}\\
	 \leq & g_{i,N}\psi(s) + 0 \label{85}.
\end{align}
(\ref{83}) is by the triangle inequality and (\ref{84}) is by Jensen's inequality. (\ref{85}) follows by definition and that $\E[\left(\E[H_{i,N}|\mathcal{F}_i(s)]\right)|\mathcal{F}_i^c(s)]=0$. Note that $\E[\left(\E[H_{i,N}|\mathcal{F}_i(s)]\right)|\mathcal{F}_i^c(s)]=0$ by Theorem 10.22 in \cite{davidson1994stochastic} following the facts that $H_{i,N}$ is mean zero and that $\mathcal{F}_i(s)$ and $\mathcal{F}_i^c(s)$ are independent $\sigma$-fields.  
\end{proof}

We shall inherit C\ref{assn:bern-des} and C\ref{assn:bounded-y}. We relax C\ref{assn:local-inf} with the following weak-dependence assumption.
\begin{assn}\label{assn:NED}
 For any $d$-circle average outcomes and for all sample sizes, $\{\mu_i(\Y,d), i\in\mathcal{S}_N\}$ is uniform $L_2$-NED of size -4 on the random field $\{Z_{i}, i\in\mathcal{S}_N\}$.\footnote{Lower-level conditions for establishing the NED property can be found in \cite{jenish2012spatial}. We omit the detail for brevity.}
\end{assn}

Denote the corresponding NED-coefficient function as $\psi_d():\mathbb{R}_+\to \mathbb{R}_+$.\footnote{To compare the strength of dependence between two circle-averages with different distance values, one can compare the corresponding NED-coefficient functions. For example, it may be reasonable to assume $\psi_{d_1}(t)\geq\psi_{d_2}(t)$ for all $t\in\mathbb{R}$ if $d_1\geq d_2$. } Now recall the class of estimators defined in Section \ref{Section:HAHTComparison}:
\begin{equation}
	\widehat{\tau}(\mu_1,\mu_0;d)= \mu_1-\mu_0 +\left( \frac{1}{Np}\sum_{i=1}^n \Z_i(\mu_i(\Y;d)-\mu_1)- \frac{1}{N(1-p)}\sum_{i=1}^n (1-\Z_i)(\mu_i(\Y;d)-\mu_0)\right).
\end{equation} 
This class contains the HT estimator ($\mu_1=\mu_0=0$) and the linearized Hajek estimator ($\mu_1=\bar{\mu}_1(d)$ and $\mu_0=\bar{\mu}_0(d)$ ) as special cases. The following lemma states that provided $\{\mu_i(\Y;d), i\in\mathcal{S}_N\}$ is uniform $L_2$-NED with size $-4$, so is the transformation $ \widehat{\tau}_{i,N}(\mu_1,\mu_0;d)=\left( \frac{\Z_i}{Np} - \frac{(1-\Z_i)}{N(1-p)}\right)\mu_i(\Y;d) + \frac{1}{N}(\mu_1-\mu_0) - \left(\frac{\Z_i\mu_1}{Np} - \frac{(1-\Z_i)\mu_0}{N(1-p)}\right)-\tau_i(d;\eta) $. Note $\E[\widehat{\tau}_{i,N}(\mu_1,\mu_0;d)]=0$.

\begin{lemma}\label{lemma:transformation}
	Under C\ref{assn:bern-des} and C\ref{assn:NED}, for each pair $\mu_1,\mu_0\in\mathbb{R}$, $\{\widehat{\tau}_{i,N}(\mu_1,\mu_0;d), i\in\mathcal{S}_N\}$ is uniform $L_2(g_N^{\tau})$-NED with size $-4$ on the random field $\{Z_{i}, i\in\mathcal{S}_N\}$. We have $g_{i,N}^{\tau}\leq\frac{C}{Np(1-p)}$, where $C$ is independent of $N$, $\mu_1$ and $\mu_0$.
\end{lemma}

\begin{proof}
	Note $\E[\widehat{\tau}_{i,N}(\mu_1,\mu_0;d)|\mathcal{F}_i(s)]$ is equal to:
	\begin{equation}
		 \left( \frac{\Z_i}{Np} - \frac{(1-\Z_i)}{N(1-p)}\right)\E[\mu_i(\Y;d)|\mathcal{F}_i(s)] + \frac{1}{N}(\mu_1-\mu_0) - \left(\frac{\Z_i\mu_1}{Np} - \frac{(1-\Z_i)\mu_0}{N(1-p)}\right)-\tau_i(d),
	\end{equation}
	because the $\sigma$ field generated by $Z_i$ is a subset of $\mathcal{F}_i(s)$.  Noting that $|\frac{\Z_i}{Np} - \frac{(1-\Z_i)}{N(1-p)}|\leq \frac{1}{Np(1-p)}$, we have
	\begin{align}
			&||\widehat{\tau}_{i,N}(d)-\E[\widehat{\tau}_{i,N}(d)|\mathcal{F}_i(s)]||_2 \leq \frac{1}{Np(1-p)} ||\mu_i(\Y,d)- \E[\mu_i(\Y,d)|\mathcal{F}_i(s)]||_2 \\
			\lesssim & \frac{1}{Np(1-p)}\psi_d(s), \label{89}
	\end{align}
	where by using the symbol $\lesssim$ we suppress a constant independent of $N$. Importantly, note that the constant is independent of $\mu_1$ and $\mu_0$. 
\end{proof}

The next lemma is a covariance inequality which, together with a later assumption on the density of intervention nodes, implies a $O_p(\frac{1}{\sqrt{N}})$ convergence rate for the HT and Hajek estimators.\footnote{See Lemma A.3 in \cite{jenish2012spatial} for a more general covariance inequality.} 
\begin{lemma}\label{lemma:covariance_inequality_raw}
Let $d_{ij}=\gamma(i,j)$ denote the distance between intervention node i and intervention node j. 
Under C\ref{assn:bern-des}, C\ref{assn:bounded-y} and C\ref{assn:NED}, for any $\mu_1$ and $\mu_0$ such that $|\mu_1|<B$ and $|\mu_0|<B$ for some positive constant $B$, 
\begin{equation}
	|\textnormal{Cov}\left(\widehat{\tau}_{i,N}(\mu_1,\mu_0;d),\widehat{\tau}_{j,N}(\mu_1,\mu_0;d)\right)|\leq 
	\frac{C(B,p)}{N^2}\psi_d(\frac{d_{ij}}{3}) 
\end{equation}
where $C(B,p)$ denotes a generic constant dependent on $B$ and $p$, and independent of $N$.
\end{lemma}

\begin{proof}
	We write $\widehat{\tau}_{i,N}(\mu_1,\mu_0;d)$ as $\widehat{\tau}_i(d)$, and $\widehat{\tau}_{j,N}(\mu_1,\mu_0;d)$ as $\widehat{\tau}_j(d)$ for brevity. Note $E[\widehat{\tau}_i(d)]=0$ and $E[\widehat{\tau}_j(d)]=0$.  we have:
	\begin{align}
		&\left|\textnormal{Cov}\left(\widehat{\tau}_{i}(d),\widehat{\tau}_{j}(d)\right)\right|= \left|\E[\widehat{\tau}_i(d)\widehat{\tau}_j(d)]\right| \\
		 \leq & \left|\E[\left(\widehat{\tau}_i(d)-\E[\widehat{\tau}_i(d)|\mathcal{F}_i\left(\frac{d_{ij}}{3}\right)]\right)\widehat{\tau}_j(d)]\right|  + \left|\E[\E[\widehat{\tau}_i(d)|\mathcal{F}_i\left(\frac{d_{ij}}{3}\right)]\widehat{\tau}_j(d)]]\right| \label{92}\\
		\leq & ||\widehat{\tau}_i(d)-\E[\widehat{\tau}_i(d)|\mathcal{F}_i\left(\frac{d_{ij}}{3}\right)]||_2 \times ||\widehat{\tau}_j(d)||_2 \label{93}\\
		 + & \left|\E[\E[\widehat{\tau}_i(d)|\mathcal{F}_i\left(\frac{d_{ij}}{3}\right)]\E[\widehat{\tau}_j(d)|\mathcal{F}_i\left(\frac{d_{ij}}{3}\right)]]\right|\\
		 \lesssim& \frac{1}{N^2} C_1(B,p) \psi_d(\frac{d_{ij}}{3}) + ||\E[\widehat{\tau}_i(d)|\mathcal{F}_i\left(\frac{d_{ij}}{3}\right)]||_2\times ||\E[\widehat{\tau}_j(d)|\mathcal{F}_i\left(\frac{d_{ij}}{3}\right)]||_2 \label{95}\\
		 \lesssim &\frac{C_1(B,p)}{N^2} \psi_d(\frac{d_{ij}}{3})+ \frac{1}{N}C_2(B,p)\times ||\E[\widehat{\tau}_j(d)|\mathcal{F}_j^c\left(\frac{d_{ij}}{3}\right)]||_2 \label{96}\\
		 \lesssim & \frac{C_1(B,p)}{N^2} \psi_d(\frac{d_{ij}}{3})+ \frac{C_2(B,p)}{N^2}\psi_d(\frac{d_{ij}}{3}) \leq \frac{C(B,p)}{N^2} \psi_d(\frac{d_{ij}}{3}) \label{97}
	\end{align}
	where $C_1(B,p)$ and $C_2(B,p)$ denote generic constants depending on $B$ and $p$, and by using the symbol $\lesssim$ we suppress constants that are independent of $N$.  (\ref{92}) follows by the triangle inequality. (\ref{93}) follows by the Cauchy-Schwarz inequality. (\ref{95}) follows by bounding the size of $\widehat{\tau}_j(d)$ using C\ref{assn:bounded-y} and by using (\ref{89}). (\ref{96}) follows by bounding the size of $\widehat{\tau}_i(d)$ using C\ref{assn:bounded-y} and by the Jensen's inequality and the fact that $\mathcal{F}_i(\frac{d_{ij}}{3})\subset\mathcal{F}^c_j(\frac{d_{ij}}{3}) $. (\ref{97}) follows by using Lemma \ref{lemma:mixingale} and reading off the size of $g_{i,N}$ from ($\ref{89}$).

\end{proof}

We need a stronger assumption to control the concentration of the intervention nodes. The following assumption prevents nodes from concentrating around a few locations.
\begin{assn}\label{assn:node_spacing}
	For each radius $s\geq 0$ and all large sample sizes $N$, there exists a $C_b>0$ such that $\sup_{i\in\node_N}|\mathcal{B}_N(i;s)|\leq b(s)$, with $b(s)=C_bs^2$.
\end{assn}

The following lemma establishes the root-N convergence rate of the HT estimator and the asymptotic equivalence of the Hajek estimator. We note that all lemmas and propositions below are stated for each distance value $d$.
\begin{lemma}\label{lemma:covariance_inequality}
Under C\ref{assn:bern-des}, C\ref{assn:bounded-y}, C\ref{assn:NED}, and C\ref{assn:node_spacing}, $\E[\hat{\tau}_{\HT}(d)]=\textnormal{AME}(d;\eta)$, and $\Var\left(\hat{\tau}_{\HT}(d)\right)=O(\frac{1}{N})$. Consider the estimator $\widehat{\tau}_{\HA}(d)$ defined in (\ref{eq:hajek}) and its asymptotic linear expansion $ \widehat{\tau}^{\Taylor}_{\HA}(d)$ defined in Lemma \ref{lemma:ha-var},  we have $\sqrt{N}( \widehat{\tau}^{\Taylor}_{\HA}(d)-\widehat{\tau}_{\HA}(d))=o_p(1)$.	
\end{lemma}

\begin{proof}
	Proof of unbiasedness of the HT estimator is by Proposition \ref{prop:identification}. To bound the variance, note
\begin{align*}
	\Var\left(\hat{\tau}_{\HT}(d)\right) = \Var\left(\sum_{i=1}^N\widehat{\tau}_{i,N}(0,0;d)\right)=\sum_{i=1}^N \sum_{j=1}^N \Cov\left(\widehat{\tau}_{i,N}(0,0;d),\widehat{\tau}_{j,N}(0,0;d)\right)
\end{align*}
Write $\widehat{\tau}_{i,N}(0,0;d)$ as $\widehat{\tau}_i(d)$ and $\widehat{\tau}_{j,N}(0,0;d)$ as $\widehat{\tau}_j(d)$. We have, for each $i$,
\begin{align}
	& \left|\sum_{j=1}^N \Cov\left(\widehat{\tau}_i(d),\widehat{\tau}_j(d)\right)\right|\leq 	\sum_{j=1}^N \left|\Cov\left(\widehat{\tau}_i(d),\widehat{\tau}_j(d)\right)\right|\lesssim  \frac{1}{N^2}\sum_{j=1}^N\psi_d(\frac{d_{ij}}{3}),
\end{align}
where by using $\lesssim$ we suppress a constant independent of $N$. To prove the lemma, we only need to show the quantity $\sum_{j=1}^N\psi_d(\frac{d_{ij}}{3})$ is uniformly bounded for all $i$ and all sample sizes $N$.

By C\ref{assn:NED}, there exists a $C_2$ such that $\psi_d(\frac{k}{3})\leq C_2\min\{1,\frac{1}{k^{2+\epsilon}}\}$.\footnote{ For proper constants $C_0$, $C_1$, $\tilde{C}_1$, and $C_2$, $\psi_d(\frac{k}{3})\leq \min\{C_0, \frac{C_1 }{(\frac{k}{3})^{2+\epsilon}}\} \leq \min\{C_0, \frac{\tilde{C}_1}{k^{2+\epsilon}}\}\leq \max\{C_0,\tilde{C}_1\} \times \min\{1,\frac{1}{k^{2+\epsilon}}\}\leq C_2\min\{1,\frac{1}{k^{2+\epsilon}}\} $}
We have the following inequality:
\begin{align}
&	\sum_{j=1}^N\psi_d(\frac{d_{ij}}{3}) \leq \lim_{K\to\infty}\left(\sum_{k=0}^{K}\left(b(k+1)-b(k)\right)\psi_d(\frac{k}{3})\right) \label{99}\\
&	\leq C_2\lim_{K\to\infty}\left(\sum_{k=0}^{K}\left(b(k+1)-b(k)\right)\min\{1,\frac{1}{k^{2+\epsilon}}\}\right) \label{100}\\
&=	C_2\lim_{K\to\infty}\left( -b(0)1+b(1)(1-1)+\sum_{k=2}^K b(k)\left( \frac{1}{(k-1)^{2+\epsilon}}-\frac{1}{k^{2+\epsilon}}\right)+b(K+1)\times \frac{1}{K^{2+\epsilon}}  \right) \label{101}\\
& \lesssim -b(0)+ \lim_{K\to\infty}\left(\sum_{k=2}^K b(k)\left( \frac{1}{(k-1)^{2+\epsilon}}-\frac{1}{k^{2+\epsilon}}\right)\right) + \lim_{K\to\infty}\left(b(K+1)\times \frac{1}{K^{2+\epsilon}}\right) \label{102}\\
& \lesssim -b(0)+ C_b\lim_{K\to\infty}\left(\sum_{k=2}^K k^2 \left( \frac{1}{(k-1)^{2+\epsilon}}-\frac{1}{k^{2+\epsilon}}\right)\right) + \lim_{K\to\infty}\left(\frac{C_b(K+1)^2}{K^{2+\epsilon}}\right) \label{103}\\
& <\infty, \label{104}
\end{align}
where $b(\cdot):\mathbb{R}_{+} \to \mathbb{R}_{+}$ is the neighborhood size function defined in C\ref{assn:node_spacing}. (\ref{99}) is by an inclusion criteria and the fact that $\psi$ is a non-increasing function. (\ref{100}) follows because $b$ by definition is a non-decreasing function. (\ref{101}) and (\ref{102}) are purely algebraic. (\ref{103}) follows by C\ref{assn:node_spacing} and the fact that $\frac{1}{(k-1)^{2+\epsilon}}-\frac{1}{k^{2+\epsilon}}\geq 0$ for all $k\geq 2$. (\ref{104}) follows because $\sum_{k=2}^\infty k^2 \left( \frac{1}{(k-1)^{2+\epsilon}}-\frac{1}{k^{2+\epsilon}}\right)<\infty$ by calculus.\footnote{ $\sum_{k=2}^\infty k^2 \left( \frac{1}{(k-1)^{2+\epsilon}}-\frac{1}{k^{2+\epsilon}}\right) =  \sum_{k=2}^\infty k^2 \left( \frac{k^{2+\epsilon}-(k-1)^{2+\epsilon}}{(k-1)^{2+\epsilon}k^{2+\epsilon}} \right)\lesssim \sum_{k=2}^{\infty}\frac{\int_{k-1}^k x^{1+\epsilon}}{(k-1)^{2+\epsilon} k^{\epsilon} } \leq  \sum_{k=2}^{\infty} \frac{k^{1+\epsilon}}{(k-1)^{2+\epsilon} k^{\epsilon} }\lesssim  \sum_{k=2}^{\infty} \frac{1}{ (k-1)^{1+\epsilon} }<\infty. $ Note each term is nonnegative. } Note the bound is uniform in $i$ and sample sizes $N$.

Linearization proof for the Hajek estimator is identical to the one in Lemma \ref{lemma:ha-var}.
\end{proof}

The following lemma is needed to bound the bias of the HAC variance estimator later.
\begin{lemma}\label{lemma:tail}
Under C\ref{assn:NED} and C\ref{assn:node_spacing}, for any $\{b_N\}_{N=1}^{\infty}$ such that $b_N\to\infty$, we have
\begin{equation}
		\lim_{N\to\infty}\sup_{i\in\node_N} \sum_{ \{j:d_{ij}>b_N\}}\psi_d(\frac{d_{ij}}{3}) = 0.
\end{equation}	
\end{lemma}
\begin{proof}
	For each $i$, similar to the proof in Lemma \ref{lemma:covariance_inequality} we have 
	\begin{align}
				  & 0\leq \sum_{ \{j:d_{ij}>b_N\}}\psi_d(\frac{d_{ij}}{3}) \leq \left(\sum_{k=\floor{b_N}}^{\infty}\left(b(k+1)-b(k)\right)\psi_d(\frac{k}{3})\right)\\
				 \leq &  b(\floor{b_N})\psi_d(\frac{\floor{b_N}}{3}) + \sum_{k=\floor{b_N}+1}^{\infty}b(k)\left( \frac{1}{(k-1)^{2+\epsilon}}-\frac{1}{k^{2+\epsilon}}\right). 
	\end{align}
	Note the right-hand side is independent of $i$. Thus we have:
	\begin{align}
			 &\lim_{N\to\infty}\sup_{i\in\node_N} \sum_{ \{j:d_{ij}>b_N\}}\psi_d(\frac{d_{ij}}{3})\\
			 \leq &  	\lim_{N\to\infty} \left( b(\floor{b_N})\psi_d(\frac{\floor{b_N}}{3}) + \sum_{k=\floor{b_N}+1}^\infty b(k)\left( \frac{1}{(k-1)^{2+\epsilon}}-\frac{1}{k^{2+\epsilon}}\right)\right) =0,
	\end{align}
    since under C\ref{assn:NED} and C\ref{assn:node_spacing} we have $\sum_{k=2}^{\infty}b(k)\left( \frac{1}{(k-1)^{2+\epsilon}}-\frac{1}{k^{2+\epsilon}}\right)<\infty $ as shown in Lemma \ref{lemma:covariance_inequality}  and $b(\floor{b_N})\psi_d(\frac{\floor{b_N}}{3})\to 0$.
	
\end{proof}

\begin{prop}\label{prop:clt_weak_dependence}
Under C\ref{assn:bern-des}, C\ref{assn:bounded-y}, C\ref{assn:NED}, and C\ref{assn:node_spacing} and if $N\times \AVar(\widehat{\tau}_{\HA}(d))$ is uniformly bounded below for all large $N$, we have,
\begin{equation}
	\frac{\widehat{\tau}_{\HA}(d)-\textnormal{AME}(d;\eta)}{\sqrt{\AVar(\widehat{\tau}_{\HA}(d))}} \overset{d}{\to} N(0,1).
\end{equation}
\end{prop}
\begin{proof}
	    The proof of $\sqrt{N}\left(\widehat{\tau}_{\HA}(d)-\textnormal{AME}(d;\eta)\right)=\sqrt{N}\left(\widehat{\tau}_{\HA}^\Taylor(d)-\textnormal{AME}(d;\eta)\right)+o_p(1)$ is similar as in Section \ref{Section:Normality}. We only need to establish the asymptotic distribution for $\widehat{\tau}_{\HA}^\Taylor(d)$. 
	    
	    The proof follows from Theorem 2 in \cite{jenish2012spatial} (JP hereafter) with a slight modification of the proof. Assumption 1 in \cite{jenish2012spatial} is replaced by C\ref{assn:node_spacing}. Assumption 1 in JP is used to prove the covariance inequality in Lemma A.3 in JP, which can be replaced by the inequality in the proof of Lemma \ref{lemma:covariance_inequality}. Assumption 1 is also used in established Theorem A.1 in JP for the $\alpha$-mixing CLT. Because our assignment variables $\{Z_{i,N},i\in\mathcal{S}_N\}$ are independent by C\ref{assn:bern-des} and by C\ref{assn:node_spacing}, we can replace the $\alpha-$mixing CLT with a CLT with bounded-degree dependency graph. Assumption 2(a) and Assumption 4(a) in JP are trivially satisfied by C\ref{assn:bounded-y} and choosing $c_{i,n}=\frac{c}{N}$ for all $i$ by some constant $c>0$ independent of $N$.  Assumption 2(b) and Assumption 3 in JP are trivially satisfied by C\ref{assn:bern-des}.  Assumption 4(b) is satisfied by the premise of this lemma. Assumption 4(c) is satisfied by let $d_{in}=\frac{d}{N}$ for all $i$ by some constant $d>0$ independent of $N$. Then the proposition follows.
\end{proof}

We study the asymptotic variance of the Hajek estimator under the current set of assumptions:

\begin{lemma}
	Under C\ref{assn:bern-des}, the variance of the linearized Hajek estimator can be expressed as
	\begin{align}
		& \Var \left(\widehat{\tau}_{\HA}^{\Taylor}(d)\right)\\
		=&\frac{1}{N^2p}\sum_{i=1}^N \E\left[\left(\muit- \bar {\mu}^1(d)\right)^2\right] + \frac{1}{N^2(1-p)}\sum_{i=1}^N \E\left[\left(\muic- \bar {\mu}^0(d)\right)^2\right]  \\
		+ & \frac{1}{N^2} \sum_{i=1}^N\sum_{j\not =i}\sum_{a=0}^1 \sum_{b=0}^1(-1)^{a+b} \E[\left(\muiab - \bar{\mu}^a(d)\right) \left(\mujab - \bar{\mu}^b(d)\right)] 
	\end{align}
	
\end{lemma}

\begin{proof}
	
	We have the expansion, similar to the derivation in Lemma \ref{lemma:ha-var-homo},	
	\begin{align}
		& \Var \left(\widehat{\tau}_{\HA}^{\Taylor}(d)\right)\\
		=&\frac{1}{N^2p}\sum_{i=1}^N \E\left[\left(\muit- \bar {\mu}^1(d)\right)^2\right] + \frac{1}{N^2(1-p)}\sum_{i=1}^N \E\left[\left(\muic- \bar {\mu}^0(d)\right)^2\right]  \label{0620_111}\\
		-&  \frac{1}{N^2}\sum_{i=1}^N\E^2\left[ \muit - \muic - (\bar{\mu}^1(d) - \bar{\mu}^0(d) )\right] \label{0620_112}\\
		+ & \frac{1}{N^2} \sum_{i=1}^N\sum_{j\not =i}\sum_{a=0}^1 \sum_{b=0}^1(-1)^{a+b} \E[\left(\muiab - \bar{\mu}^a(d)\right) \left(\mujab - \bar{\mu}^b(d)\right)]  \label{0620_113}\\
		- & \frac{1}{N^2}  \sum_{i=1}^N\sum_{j\not =i}\sum_{a=0}^1 \sum_{b=0}^1 (-1)^{a+b} \E[\muia - \bar{\mu}^a(d)] E[\mujb - \bar{\mu}^b(d)],  \label{0620_114}\\
		= & (\ref{0620_111}) + (\ref{0620_113})+0.
	\end{align}
	where we use the fact that $\sum_{i=1}^N\E[\muia-\bar{\mu}^a(d)]=0$ by definition for $a\in\{0,1\}$.
\end{proof}

Now we provide results on HAC variance estimation. We first establish several lemmas. 

\begin{lemma}\label{lemma:covariance_ned}
	Under C\ref{assn:bern-des}, C\ref{assn:bounded-y}, C\ref{assn:NED}, for any pair of $i$ and $j$ and for any $\mu_1$ and $\mu_0$ such that $|\mu_1|<B$ and $|\mu_0|<B$ for some positive constant $B$,  we have
	\begin{align}
		||\widehat{\tau}_{j,N}(\mu_1,\mu_0;d)-\E[\widehat{\tau}_{j,N}(\mu_1,\mu_0;d)|\mathcal{F}_i(s)]||_2\lesssim\begin{cases}
			\frac{1}{N} & \textnormal{ if } s\leq d_{ij}\\
			\frac{1}{N}\psi_d(s-d_{ij}) & \textnormal{ if } s> d_{ij}
		\end{cases},
	\end{align}
	where by using the symbol $\lesssim$ we suppress a constant dependent of B but independent of $N$.  Further,
	\begin{align}
		&||\widehat{\tau}_{i,N}(\mu_1,\mu_0;d)\widehat{\tau}_{j,N}(\mu_1,\mu_0;d)-\E[\widehat{\tau}_{i,N}(\mu_1,\mu_0;d)\widehat{\tau}_{j,N}(\mu_1,\mu_0;d)|\mathcal{F}_i(s)]||_2\\
		 \lesssim & \begin{cases}
			\frac{1}{N^2} & \textnormal{ if } s\leq d_{ij}\\
			\frac{1}{N^2}\psi_d(s-d_{ij}) & \textnormal{ if } s> d_{ij},
		\end{cases}.
	\end{align}
	where by using the symbol $\lesssim$ we suppress a constant dependent of B but independent of $N$.
\end{lemma}

\begin{proof}
	For brevity we write $\widehat{\tau}_{i,N}(\mu_1,\mu_0;d)$ as $\widehat{\tau}_{i,N}(d)$ and $\widehat{\tau}_j(\mu_1,\mu_0;d)$ as $\widehat{\tau}_j(d)$.
	
	For the case where $s\leq d_{ij}$ the statements are proved by the triangle inequality and C\ref{assn:bounded-y}. For the case where $s>d_{ij}$, we have:
	\begin{align}
		& 	||\widehat{\tau}_{j,N}(d)-\E[\widehat{\tau}_{j,N}(d)|\mathcal{F}_i(s)]||_2 \\
		\leq & ||\widehat{\tau}_{j,N}(d)-\E[\widehat{\tau}_{j,N}(d)|\mathcal{F}_{j}(s-d_{ij})]||_2+||\E[\widehat{\tau}_{j,N}(d)|\mathcal{F}_j(s-d_{ij})]-\E[\widehat{\tau}_{j,N}(d)|\mathcal{F}_i(s)]||_2 \label{108}\\
		\lesssim & \frac{1}{N}\psi_d(s-d_{ij}) + ||\E[\widehat{\tau}_{j,N}(d)-\E[\widehat{\tau}_{j,N}(d)|\mathcal{F}_j(s-d_{ij})]|\mathcal{F}_i(s)] ||_2 \label{109}\\
		\lesssim & \frac{1}{N}\psi_d(s-d_{ij}) +||\widehat{\tau}_{j,N}(d)-\E[\widehat{\tau}_{j,N}(d)|\mathcal{F}_j(s-d_{ij})] ||_2 \label{110}\\
		\lesssim & \frac{1}{N}\psi_d(s-d_{ij}) \label{111}.
	\end{align}
	(\ref{108}) is by triangle inequality. (\ref{109}) follows the estimates in (\ref{89}) and the fact that $\mathcal{F}_j(s-d_{ij})\subset \mathcal{F}_i(s)$. (\ref{110}) follows by the Jensen's inequality, and (\ref{111}) follows by (\ref{89}). 
	
	\begin{align}
		& ||\widehat{\tau}_{i,N}(d)\widehat{\tau}_{j,N}(d)-\E[\widehat{\tau}_{i,N}(d)\widehat{\tau}_{j,N}(d)|\mathcal{F}_i(s)]||_2 \\
		\leq  & ||\widehat{\tau}_{i,N}(d)\widehat{\tau}_{j,N}(d)-\widehat{\tau}_{i,N}(d)\E[\widehat{\tau}_{j,N}(d)|\mathcal{F}_i(s)]||_2 \label{115} \\
		& + ||\widehat{\tau}_{i,N}(d)\E[\widehat{\tau}_{j,N}(d)|\mathcal{F}_i(s)]-\E[\widehat{\tau}_{i,N}(d)|\mathcal{F}_i(s)]\E[\widehat{\tau}_{j,N}(d)|\mathcal{F}_i(s)]||_2 \nonumber  \\
		&+||\E[\widehat{\tau}_{i,N}(d)|\mathcal{F}_i(s)]\E[\widehat{\tau}_{j,N}(d)|\mathcal{F}_i(s)]-\E[\widehat{\tau}_{i,N}(d)\widehat{\tau}_{j,N}(d)|\mathcal{F}_i(s)]||_2  \nonumber 
		\\
		\lesssim &   \frac{1}{N}||\widehat{\tau}_{j,N}(d)-\E[\widehat{\tau}_{j,N}(d)|\mathcal{F}_i(s)]||_2 + \frac{1}{N}||\widehat{\tau}_{i,N}(d)-\E[\widehat{\tau}_{i,N}(d)|\mathcal{F}_i(s)]||_2 \label{117} \\
		& + \frac{1}{N}||\widehat{\tau}_{i,N}(d)-\E[\widehat{\tau}_{i,N}(d)|\mathcal{F}_i(s)]||_2 \nonumber ,
	\end{align}
		where by using the symbol $\lesssim$ we suppress a constant dependent of B but independent of $N$.
	(\ref{117})	is bounded by, up to a constant independent of $N$, $\frac{1}{N^2}\psi_d(s-d_{ij}) + \frac{1}{N^2} \psi_d(s)$ if $s>d_{ij}$, and $\frac{1}{N^2}$ if $s\leq d_{ij}$. (\ref{115}) 
is
by triangle inequality. (\ref{117}) is by C\ref{assn:bounded-y}. 
	
\end{proof}

Let $K:\mathbb{R}_+\to \mathbb{R}$ be the kernel function used for the HAC variance estimator. We make the following assumption on $K$. 
\begin{assn}\label{assn:kernel}
	The kernel function $K:\mathbb{R}_+\to \mathbb{R}$ has the following properties:
	\begin{enumerate}
		\item $K(0)=1$.
		\item It is supported on $[0,1]$.\footnote{This means that $K(x)=0$ for all $x>1$.}
		\item It is uniformly bounded:  there exists a $K_{max}<\infty$ such that $\sup_{x\in \mathbb{R}_+}|K(x)|<K_{max}$.
		\item It is locally Lipschitz at 0: there exists a $\epsilon\in(0,1)$ and a positive constant $C>0$ such that
		\begin{equation}
			|K(x)-1|\leq C|x|, 
		\end{equation}
		for any $x\in [0,\epsilon)$.
	\end{enumerate}
\end{assn}

Lemmas below provide preliminary estimates for some random quantities. We define the following quantity:
\begin{equation}
	\widehat{\epsilon}_{i,N}(\mu_1,\mu_0;d)= \frac{\Z_i}{Np}\left(\mu_i(\Y;d)-\mu_1 \right)-\frac{1-\Z_i}{N(1-p)}\left(\mu_i(\Y;d)-\mu_0 \right).
\end{equation}
In terms of dependency, $\widehat{\epsilon}_{i,N}(\mu_1,\mu_0;d)$ behaves similarly as $\widehat{\tau}_{i,N}(\mu_1,\mu_0;d)$. We summarize them in the lemma below:
\begin{lemma}\label{lemma:epsilon}
	Under C\ref{assn:bern-des}, C\ref{assn:bounded-y}, C\ref{assn:NED} and C\ref{assn:node_spacing}. For any pair $i$ and $j$ and for any $\mu_1$ and $\mu_0$ such that $|\mu_1|<B$ and $|\mu_0|<B$ for some positive constant $B$,
	\begin{equation}
		\left|\Cov\left(\widehat{\epsilon}_{i,N}(\mu_1,\mu_0;d),\widehat{\epsilon}_{j,N}(\mu_1,\mu_0;d) \right)\right|\lesssim \frac{1}{N^2}\psi_d\left(\frac{d_{ij}}{3}\right)
	\end{equation}
	where by using the symbol $\lesssim$ we suppress a constant dependent of B but independent of $N$. Moreover, 
		\begin{align}
		||\widehat{\epsilon}_{j,N}(\mu_1,\mu_0;d)-\E[\widehat{\epsilon}_{j,N}(\mu_1,\mu_0;d)|\mathcal{F}_i(s)]||_2\lesssim\begin{cases}
			\frac{1}{N} & \textnormal{ if } s\leq d_{ij}\\
			\frac{1}{N}\psi_d(s-d_{ij}) & \textnormal{ if } s> d_{ij}
		\end{cases},
	\end{align}
	where by using the symbol $\lesssim$ we suppress a constant dependent of B but independent of $N$. Moreover, 
	\begin{align}
		&||\widehat{\epsilon}_{i,N}(\mu_1,\mu_0;d)\widehat{\epsilon}_{j,N}(\mu_1,\mu_0;d)-\E[\widehat{\epsilon}_{i,N}(\mu_1,\mu_0;d)\widehat{\epsilon}_{j,N}(\mu_1,\mu_0;d)|\mathcal{F}_i(s)]||_2\\
		\lesssim & \begin{cases}
			\frac{1}{N^2} & \textnormal{ if } s\leq d_{ij}\\
			\frac{1}{N^2}\psi_d(s-d_{ij}) & \textnormal{ if } s> d_{ij},
		\end{cases}.
	\end{align}
	where by using the symbol $\lesssim$ we suppress a constant dependent of B but independent of $N$.
\end{lemma}
\begin{proof}
The first statement follows from Lemma \ref{lemma:covariance_inequality_raw} upon noting that $\widehat{\epsilon}_{i,N}(\mu_1.\mu_0;d)$ differs from $\widehat{\tau}_{i,N}(\mu_1.\mu_0;d)$ by a constant. The rest statements follow the same calculation as in Lemma \ref{lemma:covariance_ned}.
\end{proof}

Define the quantity
\begin{equation}
	\widehat{v}_{i,N}(\mu_1,\mu_0;d,b_N)=\sum_{j\in \mathcal{B}(i;b_N)}\widehat{\epsilon}_{i,N}(\mu_1,\mu_0;d)\widehat{\epsilon}_{j,N}(\mu_1,\mu_0;d)K(\frac{d_{ij}}{b_N}),
\end{equation}
and $v_{i,N}(\mu_1,\mu_0;d,b_N)=\E[\widehat{v}_{i,N}(\mu_1,\mu_0;d,b_N)]$
\begin{lemma}\label{lemma:covcov_inequality} 

	Under C\ref{assn:bern-des}, C\ref{assn:bounded-y},  C\ref{assn:NED}, C\ref{assn:node_spacing}, and C\ref{assn:kernel}, for all $i\in\node_N$ and for any $\mu_1$ and $\mu_0$ such that $|\mu_1|<B$ and $|\mu_0|<B$ for some positive constant $B$, we have, uniformly for all $i$ and all sample sizes $N$,
	\begin{enumerate}[label=(\roman*)]
		\item $|\widehat{v}_{i,N}(\mu_1,\mu_0;d,b_N)|\lesssim\frac{b_N^2}{N^2}$ and $\left|v_{i,N}(\mu_1,\mu_0;d,b_N)\right|\lesssim\frac{b_N^2}{N^2}$,
		\item 	If $s> b_N$, 
		\begin{equation}
		||\widehat{v}_{i,N}(\mu_1,\mu_0;d,b_N)-\E[\widehat{v}_{i,N}(\mu_1,\mu_0;d,b_N)|\mathcal{F}_i(s)]||_2\lesssim \frac{1}{N^2}\sum_{j\in \mathcal{B}(i;b_N)}\psi_d(s-d_{ij}),
		\end{equation}
		\item If $s>3b_N$,
		\begin{equation}
			|\textnormal{Cov}\left(\widehat{v}_{i,N}(d,b_N),\widehat{v}_{j,N}(d,b_N)\right)|\lesssim\frac{b_N^2}{N^4}\times\max\{\sum_{\{k: d_{ik} \leq b_N\}}\psi_d(\frac{d_{ij}}{3}-d_{ik}),\sum_{\{l: d_{jl} \leq b_N\}}\psi_d(\frac{d_{ij}}{3}-d_{jl})\},
		\end{equation}
	\end{enumerate}
	where by using the symbol $\lesssim$ we suppress a constant dependent of B but independent of $N$.
\end{lemma}
\begin{proof}
	$(i)$ follows by C\ref{assn:bern-des}, C\ref{assn:bounded-y}, C\ref{assn:node_spacing} and C\ref{assn:kernel} and an application of the triangle inequality.
	We write $\widehat{v}_{i,N}(\mu_1,\mu_0;d,b_N)$ as $\widehat{v}_{i,N}(d,b_N)$ and $\widehat{\epsilon}_{i,N}(\mu_1,\mu_0;d)$ as $\widehat{\epsilon}_{i,N}(d)$ for brevity. All the constants omitted below depend on $B$.
	\begin{align}
		&||	\widehat{v}_{i,N}(d,b_N)-\E[	\widehat{v}_{i,N}(d,b_N)|\mathcal{F}_i(s)]||_2\\
		& \leq K_{max}\times \sum_{j\in\mathcal{B}(i;b_N)} ||\widehat{\epsilon}_{i,N}\left(d\right)\widehat{\epsilon}_{j,N}\left(d\right) - E[\widehat{\epsilon}_{i,N}\left(d\right)\widehat{\epsilon}_{j,N}\left(d\right)|\mathcal{F}_i(s)]   ||_2 \label{0622_130}\\
		& \lesssim  \frac{K_{max}}{N^2} \left( |j: d_{ij}\geq s, d_{ij}\leq b_N|+ \sum_{\{j: d_{ij} \leq s\}}\psi_d(s-d_{ij})\right) \label{0622_131}
	\end{align}
	(\ref{0622_130}) is by the triangle inequality and C\ref{assn:kernel}. (\ref{0622_131}) is by the estimates in Lemma \ref{lemma:epsilon} and C\ref{assn:bounded-y}.
	If $s>b_N$, the term $|j: d_{ij}\geq s, d_{ij}\leq b_N|$ vanishes, proving the result.

	As in Lemma \ref{lemma:mixingale}, we have, if $s>b_N$,
	\begin{equation}
		||\E[\widehat{v}_{i,N}(d,b_N)-v_{i,N}(d,b_N)|\mathcal{F}_i^c(s)]||_2\lesssim \frac{1}{N^2}\sum_{j\in \mathcal{B}(i;b_N)}\psi_d(s-d_{ij}).
	\end{equation}
	 Similarly as in Lemma \ref{lemma:covariance_inequality_raw}, if $d_{ij}>3b_N$,
    \begin{align*}
		 & \left|\textnormal{Cov}\left(\widehat{v}_{i,N}(d,b_N),\widehat{v}_{j,N}(d,b_N)\right)\right|\\
		\leq & \left|\E[\left(\widehat{v}_{i,N}(d,b_N)-v_{i,N}(d,b_N) \right)\left(\widehat{v}_{j,N}(d,b_N)-v_{j,N}(d,b_N)\right)]  \right|\\
		\leq & \left|\E[\left(\widehat{v}_{i,N}(d,b_N)-v_{i,N}(d,b_N) - \E[\left(\widehat{v}_{i,N}(d,b_N)-v_{i,N}(d,b_N)\right)|\mathcal{F}_i(\frac{d_{ij}}{3})] \right)\left(\widehat{v}_{j,N}(d,b_N)-v_{j,N}(d,b_N)\right)]  \right|\\
		& + \left|\E[ \E[\left(\widehat{v}_{i,N}(d,b_N)-v_{i,N}(d,b_N)\right)|\mathcal{F}_i(\frac{d_{ij}}{3})]E[\left(\widehat{v}_{j,N}(d,b_N)-v_{j,N}(d,b_N)\right)|\mathcal{F}_i(\frac{d_{ij}}{3})]]\right|\\
		\lesssim &  \frac{b_N^2}{N^2}\times  \frac{1}{N^2}\max\{\sum_{\{k: d_{ik} \leq b_N\}}\psi_d(\frac{d_{ij}}{3}-d_{ik}),\sum_{\{l: d_{jl} \leq b_N\}}\psi_d(\frac{d_{ij}}{3}-d_{jl})\}.
    \end{align*}

\end{proof}

Recall the definition of the HAC variance estimator from (\ref{HAC_formula}). An calculation similar to that at the beginning of Section \ref{subsection:HAC} shows that the HAC variance estimator $\widehat{\V}_{\HAC}(d)$ with a kernel function $K$ can be expressed as:
\begin{align}
	&\widehat{\V}_{\HAC}(d;b_N) \\
	= & \frac{1}{N_1^2} \sum_{i=1}^{N} Z_i\hat{e}_i^2(d) + \frac{1}{N_0^2} \sum_{i=1}^{N} Z_i\hat{e}_i^2(d) + \frac{1}{N_1^2} \sum_{i=1}^{N} \sum_{j \neq i }Z_iZ_j\hat{e}_i(d)\hat{e}_j(d)K(\frac{d_{ij}}{b_N})\\
	& - \frac{2}{N_1 N_0} \sum_{i=1}^{N} \sum_{j \neq i}(1-Z_i)Z_j\hat{e}_i(d)\hat{e}_j(d)K(\frac{d_{ij}}{b_N}) + \frac{1}{N_0^2} \sum_{i=1}^N \sum_{j \neq i}(1-Z_i)(1-Z_j)\hat{e}_i(d)\hat{e}_j(d)K(\frac{d_{ij}}{b_N}).
\end{align}

Define the following quantities
\begin{align}
	& 	\widehat{\V}_{1}(d;b_N)= \frac{1}{N^2p^2} \sum_{i=1}^{N} Z_i\hat{e}_i^2(d) + \frac{1}{N^2(1-p)^2} \sum_{i=1}^{N} (1-Z_i)\hat{e}_i^2(d)\\
	& + \frac{1}{N^2p^2} \sum_{i=1}^{N} \sum_{j \neq i }Z_iZ_j\hat{e}_i(d)\hat{e}_j(d)K(\frac{d_{ij}}{b_N})\\
& - \frac{2}{N^2p(1-p)} \sum_{i=1}^{N} \sum_{j \neq i}(1-Z_i)Z_j\hat{e}_i(d)\hat{e}_j(d)K(\frac{d_{ij}}{b_N})\\
&  + \frac{1}{N^2(1-p)^2} \sum_{i=1}^N \sum_{j \neq i}(1-Z_i)(1-Z_j)\hat{e}_i(d)\hat{e}_j(d)K(\frac{d_{ij}}{b_N})
\end{align}
Note $ \widehat{\V}_{1}(d;b_N)=\sum_{i=1}^N v_i(\widehat{\bar{\mu}}_1(d),\widehat{\bar{\mu}}_0(d);d,b_N)$. Also define
\begin{align}
	\widehat{\V}_{2}(d;b_N)= \sum_{i=1}^N v_i(\bar{\mu}_1(d),\bar{\mu}_0(d);d,b_N)
\end{align}

The following assumption is similar as C\ref{assn:homo}.
\begin{assn}\label{assn:extended_homophily}
	Given a kernel function $K$ and a sequence of bandwidth $\{b_N\}_{N=1}^{\infty}$,
\begin{equation}
	\lim\inf_{N}\left(\frac{1}{N}\sum_{i=1}^N\sum_{j=1}^N K(\frac{d_{ij}}{b_N})\left(\tau_i(d)-\textnormal{AME}(d;\eta)\right)\left(\tau_j(d)-\textnormal{AME}(d;\eta)\right)\right)\geq 0.\footnote{Note if $b_N$ is fixed and $K$ is the uniform kernel, this reduces to C\ref{assn:homo}.}
\end{equation}	
\end{assn}

\begin{remark}
The following lemma shows that some limitation on the correlation of effect heterogeneity when nodes are far apart implies C\ref{assn:extended_homophily}. The lemma is given with a high-level condition. It can be checked when one imposes specific low-level conditions on the majorizing function  (i.e., $\phi_d$ function below.) and on the spacing of intervention nodes.
		\begin{lemma}
			Under C\ref{assn:bounded-y}, C\ref{assn:node_spacing} and C\ref{assn:kernel}, and 
			suppose for all large $N$ and any pair $i$ and $j$ we have 
			\begin{equation}
				|\tau_i(d)-\textnormal{AME}(d;\eta)|\times  |\tau_j(d)-\textnormal{AME}(d;\eta)|\leq \phi_d(d_{ij}),
			\end{equation}
			and $\lim_N\sup_{i\in \node_N} \sum_{j\in \mathcal{B}(i,b_N)}\phi_d(d_{ij})\to 0 $ for any $b_N\to\infty$.
			\newline 
			We have 
			$$\lim\inf_{N}\left(\frac{1}{N}\sum_{i=1}^N\sum_{j=1}^N K(\frac{d_{ij}}{b_N})\left(\tau_i(d)-\textnormal{AME}(d;\eta)\right)\left(\tau_j(d)-\textnormal{AME}(d;\eta)\right)\right)=0. 
			$$
		\end{lemma}
		\begin{proof}
			Note we have $\sum_{i=1}^N\sum_{j=1}^N \left(\tau_i(d)-\textnormal{AME}(d;\eta)\right)\left(\tau_j(d)-\textnormal{AME}(d;\eta)\right)=0$ by definition. Let $\tilde{b}_N=o(b_N^{\frac{1}{3}})$. Then, for large $N$,
			\begin{align*}
				&\lim_{N}\left|\left(\frac{1}{N}\sum_{i=1}^N\sum_{j=1}^N K(\frac{d_{ij}}{b_N}) \left(\tau_i(d)-\textnormal{AME}(d;\eta)\right)\left(\tau_j(d)-\textnormal{AME}(d;\eta)\right)\right)\right|\\=&  \lim_{N}\left|\frac{1}{N}\sum_{i=1}^N\sum_{j=1}^N \left(K(\frac{d_{ij}}{b_N})-1\right) \left(\tau_i(d)-\textnormal{AME}(d;\eta)\right)\left(\tau_j(d)-\textnormal{AME}(d;\eta)\right)\right|\\
				\lesssim& \lim_N \left|\frac{1}{N}\sum_{i=1}^N\sum_{j\in \mathcal{B}(i,\tilde{b}_N)} \frac{d_{ij}}{b_N} \right| + \left|\frac{1}{N}\sum_{i=1}^N\sum_{j\not \in \mathcal{B}(i,\tilde{b}_N)} | \left(\tau_i(d)-\textnormal{AME}(d;\eta)\right)\left(\tau_j(d)-\textnormal{AME}(d;\eta)\right)| \right|\\
				\lesssim & \lim_{N}\frac{\tilde{b}_N^3}{b_N} + \lim_{N}\sup_{i\in\node_N}\sum_{j \not \in \mathcal{B}(i,\tilde{b}_N)} | \left(\tau_i(d)-\textnormal{AME}(d;\eta)\right)\left(\tau_j(d)-\textnormal{AME}(d;\eta)\right)|=o(1).
			\end{align*}		
	\end{proof}
\end{remark}

We conclude the section with a proof of the consistency of the HAC variance estimator and establish the asymptotic validity of the Wald confidence interval.
\begin{prop}Let $b_N=o(N^{\frac{1}{6}})$ and $b_N\to\infty$.
For $\alpha<1$ and under C\ref{assn:bern-des}, C\ref{assn:bounded-y},  C\ref{assn:NED}, C\ref{assn:node_spacing}, C\ref{assn:kernel} and C\ref{assn:extended_homophily},
	\begin{enumerate}[label=(\roman*)] 
		\item $N\times\left( \widehat{\V}_{\HAC}(d;b_N)-\widehat{\V}_1(d,b_N)\right)=o_p(1)$.
		\item $N\times\left(\widehat{\V}_1(d;b_N)-\widehat{\V}_2(d,b_N)\right)=o_p(1)$. 
		\item $N\times \left(\widehat{\V}_2(d;b_N)- \E[\widehat{\V}_2(d;b_N)] \right)=o_p(1)$.
		\item $\lim\inf_{N}\left(N\times \left(\E[\widehat{\V}_2(d;b_N)]-\AVar\left(\widehat{\tau}_{\HA}(d)\right)\right)\right) \geq 0 $.
		\item $\lim_{N\to\infty}\mathbf{Prob}\left(z_{\frac{\alpha}{2}}\leq \frac{\widehat{\tau}_{\HA}(d)-\textnormal{AME}(d;\eta) }{\sqrt{	\widehat{\textnormal{V}}_{\HAC}(d,b_N)}} \leq z_{1-\frac{\alpha}{2}}\right)\geq 1-\alpha$
	\end{enumerate}
\end{prop}

\begin{proof}

We first prove (i). Notice by C{\ref{assn:bern-des}}, $\frac{N_1}{N}-p=O_p(N^{-\frac{1}{2}})$. We have
\begin{equation}
\left|\frac{N^2}{N_1^2}-\frac{1}{p^2}\right|=\left|\frac{p^2-\left(\frac{N_1}{N}\right)^2}{p^2\left(\frac{N_1}{N}\right)^2}\right|=\left|\frac{(\frac{N_1}{N}-p)(\frac{N_1}{N}+p)}{p^2\left(\frac{N_1}{N}\right)^2}\right|=O_p(N^{-\frac{1}{2}})
\end{equation}

We have
\begin{align}
	& N\times(\frac{1}{N_1^2} \sum_{i=1}^{N} Z_i\hat{e}_i^2(d)-\frac{1}{N^2p^2} \sum_{i=1}^{N} Z_i\hat{e}_i^2(d)) \\
	=& \left(\frac{N^2}{N_1^2}-\frac{1}{p^2}\right)\frac{1}{N}\sum_{i=1}^{N} Z_i\hat{e}_i^2(d) =O_p(N^{-\frac{1}{2}}),
\end{align}
where we use the fact that $\frac{1}{N}\sum_{i=1}^N Z_i\hat{e}_i^2(d)=O(1)$ by C\ref{assn:bounded-y}.\footnote{Note we define $\widehat{\bar{\mu}}^1(d)=0$ if $N_1=0$ and $\widehat{\bar{\mu}}^0(d)=0$ if $N_0=0$.}

Also,
\begin{align}
	&\left|N\times(\frac{1}{N_1^2} \sum_{i=1}^{N} \sum_{j \neq i }Z_iZ_j\hat{e}_i(d)\hat{e}_j(d)K(\frac{d_{ij}}{b_N})-\frac{1}{N^2p^2} \sum_{i=1}^{N} \sum_{j \neq i }Z_iZ_j\hat{e}_i(d)\hat{e}_j(d)K(\frac{d_{ij}}{b_N}))\right|\\
	\leq & \left(\frac{N^2}{N_1^2}-\frac{1}{p^2}\right)  \left|\frac{1}{N}\sum_{i=1}^N\sum_{j \neq i }Z_iZ_j\hat{e}_i(d)\hat{e}_j(d)K(\frac{d_{ij}}{b_N})\right|\\
	\lesssim & \left(\frac{N^2}{N_1^2}-\frac{1}{p^2}\right) \left(b_N^2 \times K_{max}\right) 
	= O_p(N^{-\frac{1}{2}}b^2_N)=o_p(1) \label{0622_144}
\end{align}
where in (\ref{0622_144}) we use C\ref{assn:bounded-y}, C\ref{assn:node_spacing} and C\ref{assn:kernel}. Other terms can be processed similarly. This proves (i).

Now we prove (ii). We have:
\begin{align}
	& \widehat{\V}_1(d;b_N)-\widehat{\V}_2(d;b_N) \\
	 = &\sum_{i=1}^N \sum_{j\in\mathcal{B}(i;b_N)}K(\frac{d_{ij}}{b_N})\times \\
	 & \underbrace{\big(\widehat{\epsilon}_{i,N}(\widehat{\bar{\mu}}^1(d),\widehat{\bar{\mu}}^0(d);d)\widehat{\epsilon}_{j,N}(\widehat{\bar{\mu}}^1(d),\widehat{\bar{\mu}}^0(d);d) -\widehat{\epsilon}_{i,N}(\bar{\mu}^1(d),\bar{\mu}^0(d);d)\widehat{\epsilon}_{j,N}(\bar{\mu}^1(d),\bar{\mu}^0(d);d) \big)}_{(u_{ij})} \label{147_0622}
\end{align}

For ($u_{ij}$) we have the expressions
\begin{align}
	&u_{ij}=\frac{Z_iZ_j}{N^2p^2}\left(-\mu_i(\Y;d)\left(\widehat{\bar{\mu}}^1(d)-\bar{\mu}^1(d)\right)-\mu_j(\Y;d)\left(\widehat{\bar{\mu}}^1(d)-\bar{\mu}^1(d)\right) + \left(\widehat{\bar{\mu}}^1(d)\right)^2-\left(\bar{\mu}^1(d)\right)^2 \right)\\
	&+\frac{(1-Z_i)(1-Z_j)}{N^2(1-p)^2}\left(-\mu_i(\Y;d)\left(\widehat{\bar{\mu}}^0(d)-\bar{\mu}^0(d)\right)-\mu_j(\Y;d)\left(\widehat{\bar{\mu}}^0(d)-\bar{\mu}^0(d)\right) + \left(\widehat{\bar{\mu}}^0(d)\right)^2-\left(\bar{\mu}^0(d)\right)^2 \right)\\
	&-\frac{(1-Z_i)Z_j}{N^2(1-p)^2}\left(-\mu_i(\Y;d)\left(\widehat{\bar{\mu}}^1(d)-\bar{\mu}^1(d)\right)-\mu_j(\Y;d)\left(\widehat{\bar{\mu}}^0(d)-\bar{\mu}^0(d)\right) + \widehat{\bar{\mu}}^0(d)\widehat{\bar{\mu}}^1(d)-\bar{\mu}^0(d)\bar{\mu}^1(d) \right)\\
	&-\frac{(1-Z_j)Z_i}{N^2(1-p)^2}\left(-\mu_i(\Y;d)\left(\widehat{\bar{\mu}}^0(d)-\bar{\mu}^0(d)\right)-\mu_j(\Y;d)\left(\widehat{\bar{\mu}}^1(d)-\bar{\mu}^1(d)\right) + \widehat{\bar{\mu}}^0(d)\widehat{\bar{\mu}}^1(d)-\bar{\mu}^0(d)\bar{\mu}^1(d) \right)
\end{align}
Some calculation will show that by C\ref{assn:bern-des}, C\ref{assn:bounded-y}, C\ref{assn:NED}, C\ref{assn:node_spacing}, and C\ref{assn:kernel} we have
\begin{align}
	 \widehat{\V}_1(d;b_N)-\widehat{\V}_2(d;b_N) =O_p(\frac{1}{N} b_N^2 N^{-\frac{1}{2}}).
\end{align}
We then have $N\times( \widehat{\V}_1(d;b_N)-\widehat{\V}_2(d;b_N))=o_p(1)$ because $b_N=o(N^{\frac{1}{6}})$. This proves (ii).

Now we prove (iii). Write for brevity $\widehat{v}_{i,N}(\bar{\mu}^1(d),\bar{\mu}^0(d);d,b_N)=\widehat{v}_{i,N}(d;b_N)$. Note that $\bar{\mu}^1(d)$ and $\bar{\mu}^0(d)$ are uniformly bounded by C\ref{assn:bounded-y} for all sample sizes $N$. 
\begin{align}
	& \Var\left(\widehat{\V}_2(d;b_N)\right) = \sum_{i=1}^N \sum_{j=1}^N \Cov(\widehat{v}_{i,N}(d;b_N),\widehat{v}_{j,N}(d;b_N))\\
	\leq &\sum_{i=1}^N\sum_{j\in \mathcal{B}(i;6b_N)}|\Cov(\widehat{v}_{i,N}(d;b_N),\widehat{v}_{j,N}(d;b_N))| +  \sum_{i=1}^N\sum_{j\not \in \mathcal{B}(i;6b_N)}|\Cov(\widehat{v}_{i,N}(d;b_N),\widehat{v}_{j,N}(d;b_N))|\\
	\lesssim & \sum_{i=1}^N\sum_{j\in \mathcal{B}(i;6b_N)}\sqrt{\Var\left(\widehat{v}_{i,N}(d;b_N)\right) \times \Var\left(\widehat{v}_{j,N}(d;b_N)\right)} \\
	+ & \sum_{i=1}^N\sum_{j\not \in \mathcal{B}(i;6b_N)} \frac{b_N^2}{N^4}\times\max\{\sum_{\{k: d_{ik} \leq b_N\}}\psi_d(\frac{d_{ij}}{3}-d_{ik}),\sum_{\{l: d_{jl} \leq b_N\}}\psi_d(\frac{d_{ij}}{3}-d_{jl})\}\\
	\lesssim & N\times 36b^2_N \times \frac{b_N^4}{N^4} + \sum_{i=1}^N\sum_{j\not \in \mathcal{B}(i;6b_N)} \frac{b_N^2}{N^4}\times C_bb_N^2 \times \frac{1}{b_N^{4+\epsilon}}\\\
	\lesssim & \frac{b^6_N}{N^3} + \frac{1}{N^2}\frac{1}{b_N^{\epsilon}}
\end{align}
Thus we have $N^2\times  \Var\left(\widehat{\V}_2(d,b_N)\right) =o(1) $ because $b_N=o(N^{\frac{1}{6}})$.

Now we prove (iv).  Write for brevity $\widehat{\epsilon}_{i,N}(\bar{\mu}^1(d),\bar{\mu}^0(d);d)=\widehat{\epsilon}_{i,N}(d)$.\footnote{Note in general 
		$E[\widehat{\epsilon}_{i,N}(d)] = \frac{1}{N}\left(\tau_i(d)-\left(\bar{\mu}^1(d)-\bar{\mu}^0(d)\right)\right)\not =0$.} First note we have:
\begin{equation}
	\E[\widehat{\V}_2(d;b_N)]=\sum_{i=1}^N\sum_{j\in\mathcal{B}(i;b_N)}K\left(\frac{d_{ij}}{b_N}\right)\E[\widehat{\epsilon}_{i,N}(d)\widehat{\epsilon}_{j,N}(d)],
\end{equation}
and,
\begin{equation}
	\AVar\left( \widehat{\tau}_{\HA}(d) \right)= \sum_{i=1}^N\sum_{j=1}^N\Cov\left(\widehat{\epsilon}_{i,N}(d),\widehat{\epsilon}_{j,N}(d)\right).
\end{equation}
We have the identity:
\begin{align}
		& \sum_{i=1}^N\sum_{j=1}^N K\left(\frac{d_{ij}}{b_N}\right) \E[\widehat{\epsilon}_{i,N}(d)\widehat{\epsilon}_{j,N}(d)] - \sum_{i=1}^N\sum_{j=1}^N K\left(\frac{d_{ij}}{b_N}\right) \Cov\left(\widehat{\epsilon}_{i,N}(d),\widehat{\epsilon}_{j,N}(d)\right)\\
		&=\sum_{i=1}^N\sum_{j=1}^N K\left(\frac{d_{ij}}{b_N}\right) \E[\widehat{\epsilon}_{i,N}(d)]\E[\widehat{\epsilon}_{j,N}(d)]\\
		& = \sum_{i=1}^N\sum_{j=1}^N K\left(\frac{d_{ij}}{b_N}\right) \left(\tau_i(d)-\textnormal{AME}(d;\eta)\right)\left(\tau_j(d)-\textnormal{AME}(d;\eta)\right). \label{179}
\end{align}

Note by Lemma \ref{lemma:epsilon},
\begin{align}
	|\Cov(\widehat{\epsilon}_{i,N}(d),\widehat{\epsilon}_{j,N}(d))| \lesssim\frac{1}{N^2}\psi_d(\frac{d_{ij}}{3}).
\end{align}
Let $\tilde{b}_N=o(b_N^{\frac{1}{3}})$. For large $N$, we then have:
\begin{align}
& \left|\sum_{i=1}^N\sum_{j=1}^N K(\frac{d_{ij}}{b_N})\Cov\left(\widehat{\epsilon}_{i,N}(d),\widehat{\epsilon}_{j,N}(d)\right) - \AVar\left( \widehat{\tau}_{\HA}(d) \right)\right| \label{181}\\
\leq  &   \sum_{i=1}^N\sum_{j=1}^N \left| K(\frac{d_{ij}}{b_N})-1\right|\times |\Cov\left(\widehat{\epsilon}_{i,N}(d),\widehat{\epsilon}_{j,N}(d)\right)|\\
\leq &   \sum_{i=1}^N\sum_{j\in \mathcal{B}(i;\tilde{b}_N)} \left| K(\frac{d_{ij}}{b_N})-1\right|\times |\Cov\left(\widehat{\epsilon}_{i,N}(d),\widehat{\epsilon}_{j,N}(d)\right)|\\
+ & (K_{\max}+1) \sum_{i=1}^N\sum_{j\not \in \mathcal{B}(i;\tilde{b}_N)}  |\Cov\left(\widehat{\epsilon}_{i,N}(d),\widehat{\epsilon}_{j,N}(d)\right)|\\
 \lesssim &   \frac{1}{N^2}\sum_{i=1}^N\sum_{j\in \mathcal{B}(i;\tilde{b}_N)} \frac{d_{ij}}{b_N} + \frac{1}{N}\sup_{i\in\node_N}\sum_{j\not \in \mathcal{B}(i;\tilde{b}_N)}\psi_d(\frac{d_{ij}}{3}) \label{0623_176}\\
 \lesssim & \frac{1}{N} \frac{\tilde{b}_N^3}{b_N} + \frac{1}{N}\sup_{i\in\node_N}\sum_{j\not \in \mathcal{B}(i;\tilde{b}_N)}\psi_d(\frac{d_{ij}}{3}) = o(\frac{1}{N}) \label{0623_177}
\end{align}
(\ref{0623_176}) is by C\ref{assn:kernel}. (\ref{0623_177}) is by C\ref{assn:node_spacing} and Lemma \ref{lemma:tail}. 

Thus we have $N\times \left(\sum_{i=1}^N\sum_{j=1}^N K(\frac{d_{ij}}{b_N})\Cov\left(\widehat{\epsilon}_{i,N}(d),\widehat{\epsilon}_{j,N}(d)\right) - \AVar\left( \widehat{\tau}_{\HA}(d) \right)\right)=o(1)$. We have:
\begin{align}
	& \lim\inf_{N}\left(N\times \left(\E[\widehat{\V}_2(d,b_N)]-\AVar\left(\widehat{\tau}_{\HA}(d)\right)\right)\right)\\
	\geq & \lim\inf_{N} N\times \left( \sum_{i=1}^N\sum_{j=1}^N K\left(\frac{d_{ij}}{b_N}\right) \E[\widehat{\epsilon}_{i,N}(d)]\E[\widehat{\epsilon}_{j,N}(d)] \right) \\
	+ & \lim\inf_N N\times \left(\sum_{i=1}^N\sum_{j=1}^N K(\frac{d_{ij}}{b_N})\Cov\left(\widehat{\epsilon}_{i,N}(d),\widehat{\epsilon}_{j,N}(d)\right) - \AVar\left( \widehat{\tau}_{\HA}(d) \right)\right) \geq 0 \label{189},
\end{align}
where (\ref{189}) is by (\ref{179}) and C\ref{assn:extended_homophily}, and the calculation from (\ref{181}) to (\ref{0623_177}). 

This proves (iv). (v) follows similarly as in Proposition \ref{prop:HACvarconsistency}.

\end{proof}

\clearpage
\section{Notations}
\subsection{Notation}

\begin{center}
	\begin{tabular}[h]{|c|c|c|}
\hline
Notations & Definitions& First Appear in\\
\hline
$\mathcal{S}$ & The set of intervention nodes &  Page \pageref{notation:node} \\
\hline
$\mathcal{X}$ & Two-dimensional set locations for outcomes & Page \pageref{notation:outcomeset}  \\
\hline
$\Z$ & Ordered vector of experimental assignment variable & Page \pageref{notation:randomtreatment} \\
\hline
$\z$ & Realized assignment &  Page \pageref{notation:realizedtreatment}\\
\hline 
$Y_x(\z)$ & Potential outcome when assignment is $\z$ at location $x$ & Page \pageref{notation:potentialoutcome} \\
\hline 
$Y_x$ & Observed outcome at location $x$ & Page \pageref{notation:observedoutcome} \\
\hline
$\Y(\z)$ & Full set of potential outcomes when assignment is $\z$ & Page \pageref{notation:potentialoutcome_vec}\\
\hline
$\Y$ & Full set of observed outcomes & Page \pageref{notation:observedoutcome_vec}\\
\hline
$\mu_i(\Y(\z); \Omega_d)$ &  circle averages with $\Y(\z)$ at distance d& Page \pageref{def:circleaverage} \\
\hline
$\mu_i(\Y; \Omega_d)$ &  observed circle averages with $\Y(\z)$ at distance d& Page \pageref{def:obs_circleaverage} \\
\hline
$Y_{x}(z_i; \eta)$ & Marginalized potential outcome at point x & Page \pageref{notation:AME_at_x} \\
\hline
$\mu_{i}(z_i; d, \eta)$ & Marginalized circle averages & Page \pageref{notation:AME_at_i} \\
\hline
$\tau_{ix}(\eta)$ & Individual marginalized effect & Page \pageref{notation:individual_ame_x} \\
\hline
$\tau_{i}(d; \eta) $ & average of individual marginalized effect at distance d & Page \pageref{notation:individual_ame_i} \\
\hline
	\end{tabular}
\end{center}



\end{document}